\documentclass[12pt]{article}

\usepackage[english]{babel}
\usepackage{dsfont}

\usepackage[letterpaper,top=2cm,bottom=2cm,left=2.5cm,right=2.5cm,marginparwidth=1.75cm]{geometry}

\usepackage{amsmath,amssymb,bm,amsthm}
\usepackage{thmtools}
\usepackage{thm-restate}

\newtheorem{proposition}{Proposition}

\newtheorem{theorem}{Theorem}

\usepackage{multicol}
\usepackage{graphicx}
\usepackage{tikz}
\usepackage{freetikz}

\usepackage{pgfplots}
\usetikzlibrary{patterns}
\usepgfplotslibrary{fillbetween}

\pgfplotsset{compat=newest}
\pgfplotsset{plot coordinates/math parser=false}

\usepackage[toc,page,header]{appendix}
\usepackage{titletoc}
\usepackage{tocvsec2}                   
\usepackage{setspace}
\usepackage{tocloft}

\usepackage{amsmath}

\def\bs{\ensuremath\boldsymbol}

\usepackage{graphicx}
\usepackage{enumerate}
\usepackage[colorlinks=true, allcolors=blue]{hyperref}
\usepackage{authblk}

\allowdisplaybreaks[1]

\title{In the recovery of sparse vectors from quadratic measurements, the presence of linear terms breaks the square root bottleneck}
\author{Augustin Cosse\\
\textcolor{blue}{augustin.cosse@univ-littoral.fr}}
\affil{Universit\'e du Littoral C\^ote d'Opale\\
Laboratoire de Mathématiques Pures et Appliquées Joseph Liouville }

\begin{document}
\maketitle

\begin{abstract}
Motivated by recent results in the statistical physics of spin glasses, we study the recovery of a sparse vector $\bs x_0\in \mathbb{S}^{n-1}$, $\|\bs x_0\|_{\ell_0} = k<n$, from $m$ quadratic measurements of the form $ (1-\lambda)\langle \bs A_i, \bs x_0\bs x_0^T \rangle + \lambda \langle\bs c_i,\bs x_0 \rangle $ where $\bs A_{i}, \bs c_{i}$ have i.i.d Gaussian entries. This can be related to a constrained version of the 2-spin Hamiltonian with external field for which it was shown (in the absence of any structural constraint and in the asymptotic regime) in~\cite{belius2022triviality} that the geometry of the energy landscape becomes trivial above a certain threshold $\lambda > \lambda_c\in (0,1)$. Building on this idea we study the evolution of the so-called square root bottleneck for $\lambda\in [0,1]$ in the setting of the sparse rank one matrix recovery/sensing problem.  We show that recovery of the vector $\bs x_0$ can be guaranteed as soon as $m\gtrsim k^2 (1-\lambda)^2/\lambda^2$,  $\lambda \gtrsim k^{-1/2}$ provided that this vector satisfies a sufficiently strong incoherence condition, thus retrieving the linear regime for an external field $(1-\lambda)/\lambda \lesssim k^{-1/2}$. Our proof relies on an interpolation between the linear and quadratic settings, as well as on standard convex geometry arguments.
\end{abstract}


\section{Introduction}

\subsection{\label{quadraticEquations}Quadratic equations, external field and square root bottleneck}

An important question in classical (i.e. linear) compressed sensing is the design of measurement matrices.  The traditional view of compressed sensing considers an object that is supposed to have a sparse representation in some basis and tries to recover this object from linear measurements $\bs A\bs x = \bs A\bs x_0 = \bs b$ where $\bs A\in \mathbb{R}^{m\times n}$.    When $m$ is larger than $n$, by selecting sufficiently diverse measurement vectors, the system $\bs A\bs x = \bs b$ can be inverted and the solution $\bs x_0$ can be recovered.  A more interesting situation occurs when $m<n$.  In this case, despite the limited number of measurements,  because of the existence of a sparse representation for $\bs x_0$, it remains possible to recover this signal exactly provided that $k^2\lesssim m$.  Such a sample complexity induces a ``square-root bottleneck" $k\lesssim \sqrt{m}$ ($k\lesssim m^{1/2+\varepsilon}$ for some $\varepsilon$ to be exact~\cite{bourgain2011explicit}) on the size of the support of $\bs x_0$.  For a fixed number of measurements, only those vectors whose support is smaller than $\sqrt{k}$ are guaranteed to be successfully recovered.  In theory one can in fact go beyond this bottleneck bound and show that with high probability over the measurement matrices $\bs A$, recovery remains possible for $k=\Omega(m)$~\cite{candes2006stable}. Although the probabilistic result implies the existence of efficient measurement matrices beyond the square root bottleneck, checking that a particular matrix is an efficient measurement matrix or explicitly constructing one is hard in practice~\cite{bandeira2013certifying}.  

What is perhaps equally interesting is that a similar phenomenon (albeit different in nature) seems to arise when considering quadratic measurements~\cite{li2013sparse} (or the optimization of quadratic objectives~\cite{amini2008high, wang2016statistical}).  Moreover, in this case, due to the additional difficulty introduced by the non linearity, even for random matrices, exact recovery of a signal $\bs x_0\in \mathbb{R}^n$ with at most $k$ non vanishing entries seems to require $m\gtrsim k^2\vee n$ despite the $O(k)$ intrinsic dimension of the problem. This difficulty is essentially a computational one and it can be related to the joint nature of the unknown.  Simultaneoulsy enforcing a number of vanishing entries and the rank one structure on the matrix appearing in convex relaxations implies that the optimal number of parameters is much less than the typical performance of both of the entrywise $\ell_1$ norm ($m\gtrsim k^2$) or nuclear norm ($m\gtrsim n$) penalties.  Each of those norms will thus yield suboptimal sample complexities. In fact it is shown in~\cite{oymak2015simultaneously} that any combination of those will still remain suboptimal while it is also shown that minimizing a weighted combination of the $\ell_0$ norm and rank will recover the solution uniquely for $m\gtrsim k$ measurements only (up to log factors).  In this setting again,  a similar``square root bottleneck" thus seems to appear.  While theory posits the well posedness of the original problem for any vector $\bs x_0$,  in the absence of any sufficiently serious prior on the unknowns, any computational endeavour currently seems bound to fail for measurement operators with $m<k^2$. Possible subterfuges then include the use of generative priors (see for example~\cite{bora2017compressed,
aubin2019spiked,hand2018phase}) or the constraining of the initialization of iterative algorithms (see for example~\cite{soltanolkotabi2019structured, wang2017sparse}). Such approaches come at the expense of stronger conditions which are not always satisfied and/or require prior knowledge on the solution which is not always available.  

In this paper, we are interested in investigating this curiosity under the light of a recent characterization of the landcape of the $p$-spin Hamiltonian through the Kac-Rice formula~\cite{belius2022triviality}.  Statistical physics has since long been a provider of a number of tools that revealed particularly useful to the signal processing community (see for example~\cite{mezard2009information, krzakala2012statistical}) by relying on the interpretation of cost functions as Hamiltonians associated to particular physical systems (such as spin glasses or more general systems of particles).  If $H_n(\bs \sigma)$, $\bs \sigma \in \mathbb{R}^n$ denotes the pure $p$-spin Hamiltonian restricted to the unit sphere $\mathbb{S}^{n-1}$ and $H_n^h(\bs \sigma) = H_n(\bs \sigma) + n h \langle \bs \varphi,  \bs \sigma\rangle $ denotes the extension of this Hamiltonian to the presence of an external field $h\bs \varphi\in \mathbb{R}^n$, if we use $E_n$ to define the event under which the only critical points of $H_{n}^h$ are one maximum and one minimum, then the main result of~\cite{belius2022triviality} states that for $h^2$ sufficiently large, $\lim_{n\rightarrow \infty} P(E_n) = 1$.

Those recent findings raise a natural question: \emph{Can we observe a similar phenomenon when considering the recovery of {\bfseries \upshape structured} signals from quadratic equations under the presence of sufficiently strong linear terms?} \emph{In particular, if the landscape of the $p$-spin Hamiltonian becomes trivial for a sufficiently strong external field, {\bfseries \upshape  is it possible to break the square root bottleneck under a similar regime when trying to recover an unknown sparse vector from a system of quadratic equations?}} Our main result provides a positive answer to this question. 

In what follows, we will use $\|\bs X\|_p$ to denote the Schatten $p$-norms, and $\|\bs X\|_{\ell_p}$ to denote the entrywise $\ell_p$ norm of a matrix $\bs X$.  I.e. $\|\bs X\|_{\ell_p} = \left(\sum_{ij} |\bs X_{ij}|^p\right)^{1/p}$.  For any random variable $X$, we use $\|X\|_{\psi_p}$ to denote the Orlicz $p$-norm of $X$ (see e.g.~\cite{wellner2013weak} section 2.2.). Moreover, given a matrix $\bs X \in \mathbb{R}^{n+1\times n+1}$ of the form 
\begin{align}
\bs X = \left[\begin{array}{cc}
1 & \bs x^* \\
\bs x& \tilde{\bs X}
\end{array}\right]
\end{align}
along the line of~\cite{barak2012hypercontractivity}, we refer to $\bs x$ as encoding the first order pseudo-moments of $\bs X$ and to $\tilde{\bs X}$ as encoding the second order pseudo-moments. Given a sparse vector $\bs x_0$ we use $S$ to denote the support of $\bs x_0$ and $\Omega$ to denote the support of the rank one matrix $\bs X_0 = [1,\bs x_0^*][1, \bs x^*_0]^*$. 

To formalize the discussion above and to provide an answer to the questions raised, we consider the following setting.  In the long tradition of~\cite{candes2006robust, donoho2006compressed} we define our information operator $I_m: \mathbb{S}^{n-1}\mapsto \mathbb{R}^m$ that samples $m$ observations about a signal $\bs x_0$ as 
\begin{align}
   I_m(\bs x_0) = \left( \langle \tilde{\bs A}_1, \bs X_0\rangle, \ldots, \langle \tilde{\bs A}_m,  \bs X_0\rangle \right)\label{measurementOperatorIm}
\end{align}
where the sampling kernels $\tilde{\bs A}_i$ are defined as an interpolation between a linear part $\bs c_i$ and a quadratic part $\bs A_i$:
\begin{align}
\tilde{\bs A}_i = \left[\begin{array}{cc}
0 & \lambda \bs c_i^T\\
\lambda \bs c_i & (1-\lambda)\bs A_i\\
\end{array}\right]\label{definitionAitilde}
\end{align}
Note that we always have injectivity on the rank one manifold as shown by the following proposition

\begin{proposition}\label{wellPosedness}
Let $\bs x_0\in \mathbb{S}^{n-1}$ be a $k$-sparse signal,  and let $I_m$ denote the measurement operator defined in~\eqref{measurementOperatorIm}. Then $m\gtrsim k$ measurements are enough to recover $\bs x_0$ (in theory). 
\end{proposition}
\begin{proof}
Assume that there exist two vectors $\bs x_0, \bs y_0$ with support $S$ and $V$, $|S| = |V|=k$ that satisfy $I_m(\bs x_0) = I_m(\bs y_0)$.  Let $W = S\cup V$,  $|W|\leq 2k$. 
Simply use the injectivity of the linear map $\mathcal{A}_{\lambda}$ defined on the space of symmetric matrices as 
\begin{align}
\begin{array}{lll}
\mathcal{S}_{n+1}&\rightarrow  & \mathbb{R}^m\\
\bs X&\mapsto & \left\{2\lambda \bs c_i^*\bs X\bs e_1 + (1-\lambda)\langle \bs A_i, \tilde{\bs X}\rangle  \right\}_{1\leq i\leq m}
\end{array}\label{definitionLinearMapInterp}
\end{align}
on either $\mathcal{S}_{\tilde{\Omega}}^{n+1}$ where $\tilde{\Omega} = \Omega\setminus \left(\left\{(1, [n])\right\}\cup \left\{([n],1)\cup (1,1)\right\}\right)$ or $\mathcal{S}^{n+1}_{\Omega'}$ where $\Omega' = \left\{(W,1)\cup (1,W)\right\}$ together with $\bs X = \bs x\bs x^*$ and $\bs x_1 = 1$.
\end{proof}
Proposition~\ref{wellPosedness} of course does not imply that the signal $\bs x_0$ can be recovered efficiently (i.e. in a computationally tractable manner).  In this paper, we turn to the following semidefinite programs which we label as $\mathsf{SDP}_{1}(\tilde{\bs A}, \lambda)$ and $\mathsf{SDP}_{2}(\tilde{\bs A}, \lambda)$ 
\begin{align}
\begin{split}
    \min \quad &\sum_{k\in[n]} \left|\bs e_k^*\bs X\bs e_1\right|\\
    \text{subject to}\quad & \text{\upshape Tr}(\bs X)=2\\
    &\left\langle \left[\begin{array}{cc}
    0 & \lambda \bs c_i^*\\
    \lambda \bs c_i& (1-\lambda)\bs A_i
    \end{array}\right] , \bs X\right\rangle = b_i, \quad i=1, \ldots m\\
    &\bs X_{11}=1, \bs X\succeq 0
\end{split}\label{sdpquadraticLinearL1column}
\end{align}
\begin{align}
\begin{split}
    \min \quad &\|\bs X\|_{\ell_1}\\
    \text{subject to}\quad & \text{\upshape Tr}(\bs X)=2\\
    &\left\langle \left[\begin{array}{cc}
    0 & \lambda \bs c_i^*\\
    \lambda \bs c_i& (1-\lambda)\bs A_i
    \end{array}\right] , \bs X\right\rangle = b_i, \quad i=1, \ldots m\\
    &\bs X_{11}=1, \bs X\succeq 0
\end{split}\label{sdpquadraticLinear}
\end{align}

Our question now becomes: if $\bs x_0$ is a vector on the unit sphere $\bs x_0\in \mathbb{S}^{n-1}$ with $|\text{supp}(\bs x_0)| = k< n$, \emph{in what regime on $\lambda, m$ and $k$ can the semidefinite programs $\mathsf{SDP}_1(\tilde{\bs A}, \lambda)$ and $\mathsf{SDP}_2(\tilde{\bs A}, \lambda)$ recover $\bs x_0$ uniquely?}

A first observation is that if we completely disregarded the second order pseudo-moments in $\mathsf{SDP}_1(\tilde{\bs A}, 1)$, this program would be exactly equivalent to linear compressed sensing (see for example~\cite{candes2007sparsity}).  In this regime, we should thus reasonably expect to recover $\bs x_0$ from $m\gtrsim k$ measurements.  A second observation is that for $\lambda=0$, we should not hope to recover the complete matrix $\bs X_0 = [1, \bs x_0^*]^*[1, \bs x_0^*]$ since no observations are made on the first order pseudo-moments.  The fact that $\mathsf{SDP}_1(\tilde{\bs A}, 0)$ and $\mathsf{SDP}_2(\tilde{\bs A}, 0)$ will be unable to recover $\bs X_0$ does not imply that those programs will fail in recovering the signal $\bs x_0$. In fact for a sufficiently large number of observations (in the ``large $m$ regime", $m\gtrsim k^2$), in light of classical results on quadratic problems, we should expect recovery of the second order pseudo-moments from which the solution $\bs x_0$ can ultimately be extracted.

\begin{figure}
\centering
    \begin{tikzpicture}[every node/.style={inner sep=0,outer sep=0}]
    \node[circle, fill, minimum size=5pt,inner sep=0pt, outer sep=0pt] (e0) at (0,0) {};
    
    \coordinate (e1) at (1,2) ;
    \coordinate (e2) at (1.8,1.2) ;
    \coordinate (e3) at (1.6,-1) ;
    \draw[gray, thick, line width=.4mm] (e0) -- (e1) -- (e2) --(e3)--(e0);
    \draw[gray, thick, line width=.4mm] (e0) -- (e2);
    \draw[gray, thick, line width=.4mm] (e2) -- (e3);
    \draw[gray, thick, line width=.4mm] (e1) -- (e2);
\draw[gray, thick, line width=.4mm,dash pattern=on 5pt off 5pt on 5pt](e0) -- (2*.7,4*.7);
\draw[gray, thick, line width=.4mm,dash pattern=on 5pt off 5pt on 5pt](e0) -- (1.8*1.4,1.2*1.4);

\path[fill=gray!20,thick, line width=.4mm] (-.3,1.4) to (-.3+1.6*1,1.4+1*.3) to (-.4+1.6*1,0.5+1*.3) to (-.4,0.5) to (-.3,1.4);

\draw[gray!30, thick, line width=.4mm] (-.3,1.4) -- (-.3+1.6*1,1.4+1*.3) -- (-.4+1.6*1,0.5+1*.3) -- (-.4,0.5) -- (-.3,1.4);

\draw[gray!30, thick, line width=.4mm] (-.4,0.5) -- (-.4+1.6*1,0.5+1*.3);

\draw[gray, thick, line width=.4mm] (-.2,1) -- (1.4,1.3);
\draw[gray, thick, line width=.4mm,dash pattern=on 5pt off 5pt on 5pt] (-.2,1) -- (-.2-1.6*0.6,1-0.6*.3);

\path[fill=white!20,thick, line width=.6mm] (e0) to (e1) to (e2) to (e0);
\draw[gray, thick, line width=.4mm] (e0) -- (e1) -- (e2) --(e3)--(e0);
\draw[gray, thick, line width=.4mm] (e0) -- (e1) -- (e2) --(e0);
\node[circle, fill, minimum size=5pt,inner sep=0pt, outer sep=0pt] (emid) at (0.55,1.11) {};
\node[] (eabove) at (0.55,1.9) {$\bs a\bs a^*$};

\node[] (eabove) at (-1.4,1.2) {$\mathcal{A}(\bs X) = \bs b$};

    \end{tikzpicture}
\caption{\label{intersectionCone01} Illustration of the configuration corresponding to Proposition~\ref{propositionOnlyLinearVectorL1} in which the intersection between the positive semidefinite cone and the affine subspace $\left\{\bs X|\mathcal{A}(\bs X) = \bs b\right\}\equiv \left\{\bs X|\text{\upshape Trace}(\bs X) = 2\right\}\cap \left\{\bs X|\bs X_{11}=1\right\}\cap \left\{\bs X|\bs X_{1j} = a_j\right\}$ with $\|\bs a\|_2^2 = 1$ reduces to the matrix $\bs a\bs a^*$.}
\end{figure}
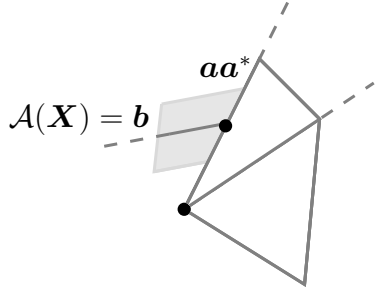

%

\subsection{Main results}
    
In order to introduce our main conclusion, we start by recalling the notion of coherence $\mu_0$ for a signal $\bs x_0$ as
\begin{align}
\mu_0 \equiv \frac{\|\bs x_0\|_\infty}{\|\bs x_0\|_2}\label{definitionCoherence}
\end{align}
Note that for a $k$-sparse vector of $\mathbb{S}^{n-1}$ ($\|\bs x_0\|_2 = 1$), we always have $1/\sqrt{k}\leq \mu_0 \leq 1$ where $\mu_0 = 1$ for a perfectly coherent signal and $\mu_0 = 1/\sqrt{k}$ for a perfectly incoherent one (i.e. $\bs x_0 = (1/\sqrt{k})\mathbf{1}_S$ or gaussian i.i.d.).  We further define $\tilde{\mu}_0$ as $\tilde{\mu}_0 = \sqrt{k}\mu_0$ so that $1\leq \tilde{\mu}_0\leq \sqrt{k}$. We are now ready to state the main quantitative observation of this paper. 

\begin{theorem}\label{mainTheorem}
   As soon as $m\gtrsim k^2\tilde{\mu}_0^4(1-\lambda)^2/\lambda^2 $ as well as $\lambda \gtrsim \tilde{\mu}_0k^{-1/2}$ (up to log factors), $\mathsf{SDP}_2(\tilde{\bs A}, \lambda)$ recovers the solution $\bs x_0$ with probability $1-o_n(1)$. 
\end{theorem} 

As expected, the result of Theorem~\ref{mainTheorem} only holds for $\lambda \gtrsim \tilde{\mu}_0k^{-1/2}$ and is only meaninful for $(1-\lambda)/\lambda<k^{-\varepsilon/2}$ for $\varepsilon>0$ (i.e. the square-root bottleneck only ``unlocks" for sufficiently large external fields).  We for example recover the linear regime for any $\lambda>1-O(k^{-1/2})$.  What is perhaps more surprising is that recovery above this threshold will require $\bs x_0$ to satisfy a strong incoherence condition. I.e. for $\lambda>1-O(k^{-1/2})$, the sample complexity grows as $m\gtrsim k^3\mu_0^4$ and ``remaining" in the linear regime $m\gtrsim k$ will require $\bs x_0$ to obey the condition $\mu_0\lesssim k^{-1/2}$.  For more coherent vectors, other approaches (e.g. iterative) might be needed.  

Another interesting phenomenon is that the numerical simulations seem to indicate a blow up in the sample complexity for arbitrarily small but non vanishing external fields (see Figs~\ref{phaseTransitionRecoveryColumnInnerMatrix} and~\ref{phaseIllustrationmanual}) which suggest that the lower bound on $\lambda$ might be superfluous. 

The rest of the paper is organized as follows. We start by investigating the simpler $\mathsf{SDP}_1(\tilde{\bs A},1)$ in section~\ref{sectionLinearMeasurementsVectorL1}.  We then extend those results to the setting of $\mathsf{SDP}_2(\tilde{\bs A}, 1)$ in section~\ref{sectionLinearMeasurementsMatrixL1}. The final result is obtained by considering $\mathsf{SDP}_2(\tilde{\bs A}, \lambda)$ as a perturbation of $\mathsf{SDP}_2(\tilde{\bs A}, 1)$ in section~\ref{sectionFinalResult}.

\section{\label{sectionLinearMeasurementsVectorL1}Unique recovery in the setting of $\mathsf{SDP}_1(\tilde{\bs A}, \lambda =1)$}

In order to prepare for the matrix $\ell_1$ norm, we start by extending the classical compressed sensing results from~\cite{candes2007sparsity} to the recovery of a sparse matrix under linear measurements. 

The main result of this section is the following proposition which essentially consists in a lift of the standard compressed sensing setting to the world of semidefinite programming.

\begin{proposition}\label{propositionOnlyLinearVectorL1}
As soon as $m\gtrsim k$, the semidefinite program $\mathsf{SDP}_1(\tilde{\bs A}, 1)$ recovers the solution $\bs x_0$ with probability $1-o_n(1)$ (on the genericity of the measurements and signs of $\bs x_0$). 
\end{proposition}

\begin{proof}
The result can be proved either by explicitely constructing a dual certificate (see~\cite{hand2013conditions, li2013sparse,candes2013phaselift} and references therein) or by showing uniqueness on the set $ \left\{(i,j)|\text{$i=1$ or $j=1$} \right\}$ (which can be done through traditional compressed sensing arguments) and then use the fact that on $\mathcal{S}_n^+\cap \left\{\bs X_{11}=1\right\}\cap\left\{\bs X_{i,1}=a_i\right\}\cap \left\{\text{\upshape Tr}(\bs X) = 2\right\}$, we necessarily have $\bs X_{i,i} = a_i^2$ from which follows $\bs X = \bs a\bs a^*$. Let $\tilde{\Omega} = \left\{(i,j)|\text{$i=1$ or $j=1$ or $i=j$}\right\}$. We claim that when a matrix is specified on this set as $\bs X_{\tilde{\Omega}} = (\bs a\bs a^*)_{\tilde{\Omega}}$, the only PSD completion for this matrix is the rank one matrix $\bs a\bs a^*$. To see this, first note that
\begin{align}
    1=\sum_{i=1}^n \bs X_{ii} \geq \sum_{i=1}^n \bs X_{i,1}^2 =\|\bs a\|_2^2 = 1\label{argumentcolumnTrace01}
\end{align}
As a result $\sum_{i=1}^n \bs X_{ii} - \bs X_{i,1}^2 = 0$, which, from the non negativity of the $\bs X_{ii} - \bs X_{i,1}^2$ necessarily implies $\bs X_{ii} = \bs X_{i,1}^2$. From this, any candidate solution $\bs Z$ to $\mathsf{SDP}_1(\tilde{\bs A},1)$ can read
\begin{align}
    \bs Z = \bs a\bs a^* + \bs X\succeq 0
\end{align}
with $\bs X_{1,j} = \bs X_{i,1} = \bs X_{ii} = 0$ for every $(i,j)\in [n]$. $\bs X$ thus reduces to a component in the Tangent space $T = \left\{\bs a\bs a^*\bs X + \bs X\bs a\bs a^* - \bs a\bs a^*\bs X\bs a\bs a^*\right\}$ to the PSD cone at $\bs a\bs a^*$. Using $\bs X_{1j} = \bs X_{i1} = 0$ as well as $\bs a_1 = 1$ we get 
\begin{align}
   & \bs 0= \bs a\left(\bs a^*\bs X\right)_1 + \bs X\bs a\bs a_1 - \langle \bs a\bs a^*, \bs X\rangle \bs a\bs a_1\\
   \iff & \bs X\bs a = \left[-(\bs a^*\bs X)_1 + \langle \bs a\bs a^*, \bs X\rangle \right]\bs a
\end{align}
and $\bs a$ is an eigenvector of $\bs X$ with eigenvalue $\lambda = \left[-(\bs a^*\bs X)_1 + \langle \bs a\bs a^*, \bs X\rangle \bs a_1\right]$. From this we can write $\bs X$ as $\bs X= \lambda_a \bs a\bs a^*$ but since the trace of $\bs X$ is zero, this implies $\lambda_a = 0$ and hence $\bs X = \bs 0$. In other words, the intersection between the PSD cone and the affine space 
\begin{align}
    (\bs X)_{\tilde{\Omega}} = \left[\begin{array}{cc}
    1 & \bs a^T\\
    \bs a & \text{diag}(\bs a\bs a^*)
    \end{array}\right]\label{argumentcolumnTrace0end}
\end{align}
is a single matrix as illustrated in Fig~\ref{intersectionCone01}.
We now provide an alternative proof based on a matrix certificate. In the case of $\mathsf{SDP}_1(\tilde{\bs A}, 1)$, if we let $\mathcal{A}$ to denote the linear map encoding the constraints $\langle \bs e_1\bs c_i^* + \bs c_i\bs e_1^*, \bs X\rangle  = b_i$,  any solution $\bs X$ is necessarily of the form
\begin{align}
    \bs X = \left[\begin{array}{cc}
    1 & \bs x^*\\
    \bs x&\tilde{\bs X}
    \end{array}\right]
\end{align}
Moreover, if we let $\bs H = \bs X-  \bs X_0$ with 
\begin{align}
\bs H = \left[\begin{array}{cc}
0& \bs h^*\\
\bs h & \tilde{\bs H}
\end{array}\right]    
\end{align}
We define $\Omega'$ as $\Omega' = \left\{(1,j)|j\in S\right\}\cup\left\{(j,1)|j\in S\right\}$. For any vector $\bs Y$ in the range of $\mathcal{A}_{0}^*$,  one must necessarily have
\begin{align}
    \|\bs x\|_1 & = \|\left(\bs x-\bs x_0 + \bs x_0\right)_{S}\|_1 + \|\left(\bs x-\bs x_0\right)_{S^c} + \left(\bs x_0\right)_{S^c}\|\\
    & \geq  \frac{1}{2}\langle \bs H_{\Omega'}, \text{\upshape sign}(\bs x_0)\bs e_1^* + \bs e_1\text{\upshape sign}(\bs x_0)^*\rangle  + \|\bs h_{S^c}\|_1 + \|\bs x_{0}\|_1 \\
    & \geq \frac{1}{2}\langle \bs H_{\Omega'}, \text{\upshape sign}(\bs x_0)\bs e_1^* + \bs e_1\text{\upshape sign}(\bs x_0)^* - \bs Y_{\Omega'}\rangle + \|\bs x_0\|_1 + \|\bs h_{(\Omega')^c}\|_1 - \langle \bs Y_{(\Omega')^c}, \bs H_{(\Omega')^c}\rangle  
\end{align}
As a result, if one can find a vector $\bs Y$ in the range of $\mathcal{A}_0^*$ such that 
\begin{align}
    \langle \bs H_{\Omega'}, \text{\upshape sign}(\bs x_0)\bs e_1^* + \bs e_1\text{\upshape sign}(\bs x_0)^* - \bs Y_{\Omega'}\rangle + \|\bs h_{S^c}\|_1 - \langle \bs Y_{(\Omega')^c}, \bs H_{(\Omega')^c}\rangle >0\label{conditionDualCertificatel1vector} 
\end{align}
for all $\bs H$, then $\|\bs x\|_1>\|\bs x_0\|$ and any solution to $\mathsf{SDP}_1(\tilde{\bs A}, 1)$ necessarily satisfies 
\begin{align}
    \bs X = \left[\begin{array}{cc}
    1 & \bs x_0^*\\
    \bs x_0 & \tilde{\bs X}
    \end{array}\right]
\end{align}
In order to construct such a $\bs Y$, we consider the following extension of the classical compressed sensing ansatz
\begin{align}
    \bs Y &= \frac{1}{2}\sum_{i=1}^m \bs c_i\bs c_{i, S}^* \bs C^{-1}\text{\upshape sign}(\bs x_0)\bs e_1^*+ \bs e_1 \sum_{i=1}^m\text{\upshape sign}(\bs x_0)^* \bs C^{-1} \bs c_{i, S} \bs c_i^* \\
& =  \frac{1}{2}\sum_{i=1}^m \left(\bs c_i\bs e_1^*+ \bs e_1 \bs c_i^*\right)\bs c_{i, S}^* \bs C^{-1}\text{\upshape sign}(\bs x_0)
\end{align}
where $\bs C$ is the symmetric matrix defined as 
\begin{align}
\bs C = m^{-1} \sum_{i=1}^m \bs c_{i, S}\bs c_{i, S}^*
\end{align}
We start by showing invertibility of $\bs C$ through the following proposition (an equivalent of which can be found in~\cite{candes2007sparsity})
\begin{proposition}
    Let $\bs c_i$ be i.i.d. random Gaussian vectors with mean zero and variance $1$. Let $S$ be a subset of the entries of size $|S| = k$. We have 
    \begin{align}
      \left\|\frac{1}{m} \sum_{i=1}^m\bs c_{i, S}\bs c_{i, S}^* - \bs I\right\|< \delta 
    \end{align}
    with probability at least $1-n^{-\alpha}$ as long as $m\gtrsim k \log(n)$
\end{proposition}

\begin{proof}
    The proof follows from an epsilon net argument combined with an application of the non commutative Bernstein inequality which we recall below.

\begin{proposition}[Bernstein, matrix version see~\cite{koltchinskii2011nuclear}]\label{BernsteinMatrix}
    Let $Z_1, \ldots, Z_n$ be independent random matrices with dimensions $m_1\times m_2$ that satisfy $\mathbb{E}Z_i = 0$ and $\|Z_i\|_{\psi_\alpha}\leq U$ almost surely for some constant $U$ where $\|Z\|_{\psi_\alpha}$ is the $\alpha$-Orclicz norm defined as 
    \begin{align}
        \|\|Z\|\|_{\psi_{\alpha}}\equiv \inf\left\{u>0:\mathbb{E}\exp\left(\|Z\|^\alpha/u^{\alpha}\right)\leq 2\right\}, \alpha \geq 1
    \end{align}
    and all $i=1, \ldots, n$. Define 
    \begin{align}
        \sigma_Z = \max \left\{\left\|\frac{1}{n}\sum_{i=1}^n \mathbb{E}Z_iZ_i^T\right\|^{1/2}, \left\|\frac{1}{n}\sum_{i=1}^n \mathbb{E}Z_i^TZ_i\right\|^{1/2}\right\}
    \end{align}
    Then there is a constant $C>0$ such that for all $t>0$ with probability at least $1-e^{-t}$
    \begin{align}
        &\left\|\frac{Z_1+\ldots + Z_n}{n}\right\|\\
        &\leq C\max \left\{\sigma_z\sqrt{\frac{t+\log(m)}{n}}, \left\|\|Z\|\right\|_{\psi_\alpha}\left(\log\left(\frac{\left\|\|Z\|\right\|_{\psi_\alpha}}{\sigma_z}\right)\right)^{1/\alpha} \frac{t+\log(m)}{n}\right\}
    \end{align}
\end{proposition}

First note that 
    \begin{align}
        \sigma^2 =  \left\|\mathbb{E}\frac{1}{m} \sum_{i=1}^m \bs c_{i, S} \|\bs c_{i, S}\|^2 \bs c_{i, S}^*\right\| \leq \left\|k\bs I + \mathds{1}\mathds{1}^* - \text{diag}(\mathds{1}\mathds{1}^*) \right\|\lesssim k
    \end{align}
\begin{align}
    \left\|\bs c_{i, S}\bs c_{i, S}^* - \bs I\right\|_{\psi_1}\lesssim k
\end{align}
The bound on the Orlicz norm can be derived from lemma 2.2.1 in~\cite{wellner2013weak}. Applying Proposition~\ref{BernsteinMatrix}, we thus get with probability at least $1-e^{-t}$, 
    \begin{align}
        \left\|\frac{1}{m}\sum_{i=1}^m \bs c_{i, S }\bs c^*_{i, S } - \bs I\right\| \lesssim C \sqrt{\frac{k t}{m}}\vee \frac{ kt}{m}
    \end{align}
\end{proof}

We have $\mathcal{P}_{\Omega'}(\bs Y) = \text{\upshape sign}(\bs x_0)\bs e_1^* + \bs e_1\text{\upshape sign}(\bs x_0)^*$. Now fix any $\ell\in S^c$, define $\bs v_\ell$ to be the $\ell^{th}$ row of $\sum_{i=1}^m\bs c_{i}\bs c_{i, \Omega}^*$ (or similarly the $\ell^{th}$ column of $\sum_{i=1}^m\bs c_{i, \Omega}\bs c_{i}^*$) and let 
\begin{align}
    \bs y_\ell = \left\langle \bs v_\ell, \left(\sum_{i=1}^m \bs c_{i, S}\bs c_{i, S}^*\right)^{-1} \text{\upshape sign}(\bs x_0)\right\rangle = \langle \bs w_\ell, \text{\upshape sign}(\bs x_0)\rangle 
\end{align}
where $\bs w_\ell = \bs C^{-1}\bs v_\ell$. The vector $\bs v_\ell$ is given by a sum of subexponential random variables. Again using the non commutative Bernstein inequality, noting that 
\begin{align}
    \left\|\frac{1}{m} \mathbb{E} \sum_{i=1}^m c_{i, \ell} \|\bs c_{i, S}\|_2^2 \right\| \lesssim k, \quad \left\|\frac{1}{m}\mathbb{E}\sum_{i=1}^m c_{i, \ell} \bs c_{i, S }\bs c_{i, S}^*\right\|\lesssim k
\end{align}
as well as 
\begin{align}
    \|\left\|c_{i,\ell} \bs c_{i, S}\right\|\|_{\psi_1}\lesssim \|c_{i, \ell}\|_{\psi_2} \|\bs c_{i, S}\|_{\psi_2} \lesssim k
\end{align}
Using the fact that $P(\|\bs c_{i, S}\|>t)\leq 2\exp\left( - \frac{t^2}{2k}\right)$. Using this with Proposition~\ref{BernsteinMatrix} gives 
\begin{align}
    \left\|\frac{1}{m} \sum_{i=1}^m c_{i, \ell} \bs c_{i, \Omega}^*\right\| \leq \sqrt{\frac{k t}{m}}\vee \frac{kt}{m}
\end{align}
with probability at least $1-e^{-t}$. Taking $t =\log^{\alpha}(n)$ gives $\|\sum_{i=1}^m c_{i, \ell}\bs c_{i, S}\|\lesssim C\sqrt{k m}$ (i.e. up to log factors) with probability at least $1-n^{-\alpha}$. 

To conclude, if we define the events $E_{v}$ and $E_{C}$ as $E_{C} = \left\{\|\bs C\|\gtrsim 1\right\}$ and $E_{v} = \left\{\sup_{\ell} \|\bs v_\ell\| \lesssim \sqrt{k/m}\right\}$, we have $P(E_v\cap E_{C})\geq 1- P(E_v^c) - P(E_{C}^c)\geq 1-o_n(1)$ (using a union bound to express $P(E_{v})$). Together, those results then give 
\begin{align}
    \sup_{\ell}\|\bs w_\ell\|\leq \|\bs C\|^{-1}\sup_{\ell}\|\bs v_\ell\| \lesssim \sqrt{k/m}
\end{align}
with probablity at least $1-o_n(1)$. In particular, using random signs for $\bs x_0$ and following the same reasoning as in~\cite{candes2007sparsity}, this implies $\left\|\mathcal{P}_{\Omega^c}(\bs Y)\right\|_{\infty}\leq \delta$ with probability $1-o_n(1)$ as soon as $m\gtrsim k \delta^{-1}$
\end{proof}

\section{\label{sectionLinearMeasurementsMatrixL1}Unique recovery in the setting of $\mathsf{SDP}_2(\tilde{\bs A}, \lambda =1)$}

Quite surprisingly, the extension from the vector $\ell_1$ norm to the matrix $\ell_1$ norm introduces an artefactual condition on the incoherence of the unknown vector $\bs x_0$.  This condition seems to arise mostly as a side effect of the proof technique. We discuss it in more detail in section~\ref{numericalExperimentsQuadratic} where we provide empirical evidence in favor of its superfluous character. The quantitative formulation of this observation is summarized by Proposition~\ref{propositionSDP2lambda1} below. 

\begin{proposition}\label{propositionSDP2lambda1}
As soon as $m\gtrsim k \tilde{\mu}_0^6$, the semidefinite program $\mathsf{SDP}_2(\tilde{\bs A}, 1)$ recovers the solution $\bs x_0$ with probability $1-o_n(1)$ (on the genericity of the measurements and signs of $\bs x_0$).
\end{proposition}

The rest of the section is organized as follows. Section~\ref{matrixL1linearMeas} gives the structure of the proof, whereas the auxilliary lemmas (which are essentially subexponential concentration results) are proved in section~\ref{auxialliaryResultsLinearMeasMatrixL1}. As indicated above, the numerical simulations are discussed in section~\ref{numericalExperimentsQuadratic}. 

\subsection{\label{matrixL1linearMeas}Proof outline}

For clarity, we will now use $\bs x_0$ to denote the vector $[1, \tilde{\bs x}_0]$ where $\tilde{\bs x}_0$ refers to the unknown vector which we want to recover.  $\mu_0$ thus stands for the coherence of $\tilde{\bs x}_0$. For any solution $\bs X$ of $\mathsf{SDP}_2(\tilde{\bs A}, 1)$,  if we let $\bs H = \bs X - \bs X_0$ and introduce the decomposition
\begin{align}
\bs H = \left[\begin{array}{cc}
1&\bs h^* \\
\bs h& \tilde{\bs H}
\end{array}\right]
\end{align}
we can write
\begin{align}
\|\bs X\|_{\ell_1} &\geq \|\bs X_0\|_{\ell_1} + \langle \bs X_{\Omega} - \bs X_0, \text{sign}(\bs x_0) \text{sign}(\bs x_0)^*\rangle + \left\|\bs X_{\Omega^c}\right\|_{\ell_1}\\
&\geq  \|\bs X_0\|_{\ell_1} + \langle \bs H_{\Omega}, \text{sign}(\bs x_0)\bs e_1^* + \bs e_1\text{sign}(\bs x_0)^*)\rangle + \langle \bs H, \bs e_1\bs e_1^*\rangle\\
& + \langle \tilde{\bs H} , \text{sign}(\tilde{\bs x}_0)\text{sign}(\tilde{\bs x}_0)^*\rangle + \|\bs H_{\Omega}^c\|_{\ell_1}\\
&   \geq \|\bs X_0\|_{\ell_1} + \langle \tilde{\bs H}_{\Omega}, \text{sign}(\tilde{\bs x}_0)\text{sign}(\tilde{\bs x}_0)^*\rangle + \|\bs H_{\Omega^c}\|_{\ell_1} \\
&+ \langle \bs H_{\Omega}, \text{sign}(\tilde{\bs x}_0)\bs e_1^* + \bs e_1\text{sign}(\tilde{\bs x}_0)^* \rangle \label{introducingCertificate01}  
\end{align}
We now let $B$ to denote the set of indices $B = \left\{(i,i)|i \in S^c\right\}$ with $S = \text{supp}(\bs x_0)$, from which $\text{Tr}(\tilde{\bs H}_{\Omega^c}) = \text{Tr}(\tilde{\bs H}_B) = \text{Tr}(\bs X_B)\geq 0$.  Moreover, using $\text{Tr}(\bs H) = 0$, and for any vector $\bs Y$ in the range of $\mathcal{A}^*_1$, we can write
\begin{align}
\|\bs X\|_{\ell_1} &   \geq \|\bs X_0\|_{\ell_1} + \langle \tilde{\bs H}_{\Omega}, \text{sign}(\tilde{\bs x}_0)\text{sign}(\tilde{\bs x}_0)^*\rangle + \|\bs H_{\Omega^c}\|_{\ell_1} \\
&+ \langle \bs H_{\Omega}, \text{sign}(\tilde{\bs x}_0)\bs e_1^* + \bs e_1\text{sign}(\tilde{\bs x}_0)^* \rangle  + \gamma\text{Tr}(\bs H), \bs H\rangle \\
&   \geq \|\bs X_0\|_{\ell_1} + \langle \tilde{\bs H}_{\Omega}, \text{sign}(\tilde{\bs x}_0)\text{sign}(\tilde{\bs x}_0)^*\rangle + \|\bs H_{\Omega^c}\|_{\ell_1} \\
&+ \langle \bs H_{\Omega}, \text{sign}(\tilde{\bs x}_0)\bs e_1^* + \bs e_1\text{sign}(\tilde{\bs x}_0)^* \rangle   + \gamma\text{Tr}(\bs H_{\Omega^c}) \\
&+ \gamma\text{Tr}(\bs H_{T\cap \Omega}) +\gamma\text{Tr}(\bs H_{T^\perp\cap \Omega})\\
&   \geq \|\bs X_0\|_{\ell_1} + \langle \tilde{\bs H}_{\Omega}, \text{sign}(\tilde{\bs x}_0)\text{sign}(\tilde{\bs x}_0)^*\rangle + \|\bs H_{\Omega^c}\|_{\ell_1} \\
&+ \langle \bs H_{\Omega}, \text{sign}(\tilde{\bs x}_0)\bs e_1^* + \bs e_1\text{sign}(\tilde{\bs x}_0)^* \rangle  + \gamma\text{Tr}(\bs H_{B}) \\
&+ 2^{-1}\gamma \langle \bs x_0\bs x_0^*, \bs H\rangle  +\gamma\text{Tr}(\bs H_{T^\perp\cap \Omega})\\
& \geq \|\bs X_0\|_{\ell_1} + \langle \tilde{\bs H}_{\Omega}, \text{sign}(\tilde{\bs x}_0)\text{sign}(\tilde{\bs x}_0)^*\rangle + \|\bs H_{\Omega^c}\|_{\ell_1} \\
&+ \langle \bs H_{\Omega}, \text{sign}(\tilde{\bs x}_0)\bs e_1^* + \bs e_1\text{sign}(\tilde{\bs x}_0)^* + 2^{-1}\gamma\left( \tilde{\bs x}_0\bs e_1^* + \bs e_1\tilde{\bs x}_0^*\right)\rangle\\
& + \gamma\text{Tr}(\bs H_{B}) + \gamma \text{Tr}(\bs H_{T^\perp\cap \Omega}) + 2^{-1}\gamma\langle \tilde{\bs x}_0\tilde{\bs x}_0^*, \bs H\rangle 
\end{align}
Moreover, we define the tangent space $\tilde{T}$ as $\tilde{T} = \left\{\bs X = \tilde{\bs x}_0\tilde{\bs x}^* + \tilde{\bs x}\tilde{\bs x}_0^*|\tilde{\bs x}, \tilde{\bs x}_0\in \mathbb{R}^n\right\}$. From this we can write
\begin{align}
\|\bs X\|_{\ell_1}& \geq \|\bs X_0\|_{\ell_1} + \langle \tilde{\bs H}_{\Omega}, \text{sign}(\tilde{\bs x}_0)\text{sign}(\tilde{\bs x}_0)^*\rangle + \|\bs H_{\Omega^c}\|_{\ell_1} \\
&+ \langle \bs H_{\Omega}, \text{sign}(\tilde{\bs x}_0)\bs e_1^* + \bs e_1\text{sign}(\tilde{\bs x}_0)^* + 2^{-1}\gamma\left( \tilde{\bs x}_0\bs e_1^* + \bs e_1\tilde{\bs x}_0^*\right)\rangle\\
&   + \gamma\text{Tr}(\bs H_{B}) + \gamma \text{Tr}(\bs H_{T^\perp\cap \Omega}) + 2^{-1}\gamma\langle \tilde{\bs x}_0\tilde{\bs x}_0^*, \bs H\rangle \\
&\geq \|\bs X_0\|_{\ell_1} + \langle \tilde{\bs H}_{\tilde{T}\cap \Omega}, \text{sign}(\tilde{\bs x}_0)\text{sign}(\tilde{\bs x}_0)^*\rangle + \|\bs H_{\Omega^c}\|_{\ell_1} \\
&+  \langle \tilde{\bs H}_{\tilde{T}^\perp\cap \Omega}, \text{sign}(\tilde{\bs x}_0)\text{sign}(\tilde{\bs x}_0)^*\rangle\\
&+ \langle \bs H_{\Omega}, \text{sign}(\tilde{\bs x}_0)\bs e_1^* + \bs e_1\text{sign}(\tilde{\bs x}_0)^* + 2^{-1}\gamma\left( \tilde{\bs x}_0\bs e_1^* + \bs e_1\tilde{\bs x}_0^*\right)\rangle\\
&   + \gamma\text{Tr}(\bs H_{B}) + \gamma \text{Tr}(\bs H_{T^\perp\cap \Omega}) + 2^{-1}\gamma\langle \tilde{\bs x}_0\tilde{\bs x}_0^*, \bs H\rangle\\
&\geq \|\bs X_0\|_{\ell_1} + \langle \tilde{\bs H}_{\Omega}, \mathcal{P}_{\tilde{T}\cap \Omega}\left(\text{sign}(\tilde{\bs x}_0)\text{sign}(\tilde{\bs x}_0)^*\right)\rangle + \|\bs H_{\Omega^c}\|_{\ell_1} \\
&+ \langle \bs H_{\Omega}, \text{sign}(\tilde{\bs x}_0)\bs e_1^* + \bs e_1\text{sign}(\tilde{\bs x}_0)^* + 2^{-1}\gamma\left( \tilde{\bs x}_0\bs e_1^* + \bs e_1\tilde{\bs x}_0^*\right)\rangle\\
& + \gamma\text{Tr}(\bs H_{B}) + \gamma \text{Tr}(\bs H_{T^\perp\cap \Omega}) + 2^{-1}\gamma\langle \tilde{\bs x}_0\tilde{\bs x}_0^*, \bs H\rangle\label{gettingRidOfTperpOmega}
\end{align}
In~\eqref{gettingRidOfTperpOmega} we use $\tilde{\bs H} = \tilde{\bs X} - \tilde{\bs X}_0$ from which $\mathcal{P}_{\tilde{T}^\perp\cap \Omega}(\tilde{\bs H}) =\mathcal{P}_{\tilde{T}^\perp\cap \Omega}(\tilde{\bs X})  \succeq  0$. 
Using $\|\tilde{\bs x}_0\|_2 = 1$ and developing,  we obtain 
\begin{align}
\|\bs X\|_{\ell_1}&\geq \|\bs X_0\|_{\ell_1} +  \langle \tilde{\bs H}_{ \Omega}, \left(\text{sign}(\tilde{\bs x}_0)\tilde{\bs x}_0^* + \tilde{\bs x}_0\text{sign}(\tilde{\bs x}_0)^*\right) \rangle \|\tilde{\bs x}_0\|_{\ell_1} \\
&- \langle \tilde{\bs H}_{ \Omega}, \tilde{\bs x}_0\tilde{\bs x}_0^*\rangle \|\tilde{\bs x}_0\|^2_{\ell_1} + \|\bs H_{\Omega^c}\|_{\ell_1} \\
&+ \langle \bs H_{\Omega}, \text{sign}(\tilde{\bs x}_0)\bs e_1^* + \bs e_1\text{sign}(\tilde{\bs x}_0)^* + 2^{-1}\gamma\left( \tilde{\bs x}_0\bs e_1^* + \bs e_1\tilde{\bs x}_0^*\right)\rangle\\
&   + \gamma\text{Tr}(\bs H_{B}) + \gamma \text{Tr}(\bs H_{T^\perp\cap \Omega}) + 2^{-1}\gamma\langle \tilde{\bs x}_0\tilde{\bs x}_0^*, \bs H\rangle\label{gettingRidOfTOmega}
\end{align}
An application of the Hanson-Wright inequality (see Proposition~\ref{HansonWright} below), noting that 
\begin{align}
\langle \tilde{\bs H}_{\Omega}, \text{sign}(\tilde{\bs x}_0)\tilde{\bs x}_0^*\rangle = \langle \tilde{\bs H}_{\Omega}\text{diag}(|\tilde{\bs x}_0|), \text{sign}(\tilde{\bs x}_0)\text{sign}(\tilde{\bs x}_0)^*\rangle  
\end{align}
hence using $\text{Tr}(\bs H_{T\cap \Omega}) + \text{Tr}(\bs H_{T^\perp \cap \Omega}) + \text{Tr}(\bs H_{B}) = \text{Tr}(\bs H) = 0$ together with $\bs H_B, \bs H_{T^\perp\cap \Omega}\succeq 0$
\begin{align}
\left|\mathbb{E}_{s(t)}\langle \tilde{\bs H}_{\Omega}, \text{sign}(\tilde{\bs x}_0)\tilde{\bs x}_0^*\rangle\right| &= \left|\mathbb{E}_{s(t)}\langle \tilde{\bs H}_{\Omega}\text{diag}(|\tilde{\bs x}_0|), \text{sign}(\tilde{\bs x}_0)\text{sign}(\tilde{\bs x}_0)^*\rangle \right|\\
& \leq \mu_0 \sqrt{k} \|\bs H_{\Omega}\|_F
\end{align}
and similarly 
\begin{align}
\left|\mathbb{E}_{s(t)}\langle \tilde{\bs H}_{\Omega}, \tilde{\bs x}_0\tilde{\bs x}_0^*\rangle\right| &\lesssim \mu_0^2 \sqrt{k} \|\bs H_{\Omega}\|_F
\end{align}
then gives 
\begin{align}
\|\bs X\|_{\ell_1}&\geq \|\bs X_0\|_{\ell_1} - \mu_0 (1+\log^{\alpha}(n))\|\bs H_\Omega\|_F \|\tilde{\bs x}_0\|_{\ell_1} - \mu_0\|\bs H_\Omega\|_F\sqrt{k} \|\tilde{\bs x}_0\|_{\ell_1} \\
&- \left\| \bs H_{ \Omega}\right\|_F \mu_0^2 (1+\log^{\alpha}(n)) \|\tilde{\bs x}_0\|^2_{\ell_1} - \mu_0^2 \|\bs H_\Omega\|_F \sqrt{k}\|\tilde{\bs x}_0\|^2_{\ell_1} \\
& +   \|\bs H_{\Omega^c}\|_{\ell_1} + \langle \bs H_{\Omega}, \text{sign}(\tilde{\bs x}_0)\bs e_1^* + \bs e_1\text{sign}(\tilde{\bs x}_0)^* + 2^{-1}\gamma\left( \tilde{\bs x}_0\bs e_1^* + \bs e_1\tilde{\bs x}_0^*\right)\rangle\label{afterApplyingHansonWright01}\\
& + \gamma\text{Tr}(\bs H_{B}) + \gamma \text{Tr}(\bs H_{T^\perp\cap \Omega}) - 2^{-1}\gamma \mu_0^2 \left(1+\log^{\alpha}(n)\right) \|\bs H_{\Omega}\|_F\\
&- 2^{-1}\gamma \mu_0^2 \left(1+\log^{\alpha}(n)\right) \|\bs H_{\Omega}\|_F \sqrt{k}\label{afterApplyingHansonWright02}
\end{align}
with probability at least $1-o_n(1)$. 
\begin{proposition}[see Theorem 6.2.1 in~\cite{vershynin2018high}\label{HansonWright}]
Let $\bs x = (x_1, \ldots, x_n)\in \mathbb{R}^n$ be a random vector with independent mean zero, subgaussian coordinates. Let $\bs A$ be an $n\times n$ matrix. Then for every $t\geq 0$, we have
\begin{align}
P\left(|\bs x^T\bs A\bs x - \mathbb{E} \bs x^T\bs A\bs x|\geq t\right)\leq 2\exp\left[-\left(\frac{t^2}{K^2\|\bs A\|_F^2}\wedge \frac{t}{K^2\|\bs A\|}\right)\right]
\end{align}
where $K = \vee_i \|x_i\|_{\psi_2}$.
\end{proposition}
To control~\eqref{afterApplyingHansonWright01} we use an application of Hoeffding’s inequality (see for example~\cite{boucheron2003concentration})
From this we get
\begin{align}\begin{split}
\|\bs X\|_{\ell_1}&\geq \|\bs X_0\|_{\ell_1} - \mu_0 (1+\log^{\alpha}(n))\|\bs H_\Omega\|_F \|\tilde{\bs x}_0\|_{\ell_1} - \mu_0\|\bs H_\Omega\|_F\sqrt{k} \|\tilde{\bs x}_0\|_{\ell_1} \\
&- \left\| \bs H_{ \Omega}\right\|_F \mu_0^2 (1+\log^{\alpha}(n)) \|\tilde{\bs x}_0\|^2_{\ell_1} - \mu_0^2 \|\bs H_\Omega\|_F \sqrt{k}\|\tilde{\bs x}_0\|^2_{\ell_1} \\
& +   \|\bs H_{\Omega^c}\|_{\ell_1}  - \|\bs H_{\Omega}\|_F \log(n) \left(1+2^{-1}\gamma \mu_0\right)\\
& + \gamma\text{Tr}(\bs H_{B}) + \gamma \text{Tr}(\bs H_{T^\perp\cap \Omega}) - 2^{-1}\gamma \mu_0^2 \left(1+\log^{\alpha}(n)\right) \|\bs H_{\Omega}\|_F\\
&- 2^{-1}\gamma \mu_0^2  \|\bs H_{\Omega}\|_F \sqrt{k}\end{split}\label{afterApplyingHansonWright03}
\end{align}
Note that $\text{Tr}(\bs H_{T\cap \Omega^c}) = 0$. Hence $\text{Tr}(\bs H_B) + \text{Tr}(\bs H_{T^\perp\cap \Omega }) = \text{Tr}(\bs H_{T^\perp})$. Let $\bs H_{T} = \bs x_0(\bs x_{S^c})^* + \bs x_{S^c}\bs x_0^*$.  Using the duality of the $\ell_1$ norm, we also have 
\begin{align}
    \|\bs H_{\Omega^c}\|_{\ell_1} &\geq \langle \mathcal{P}_{\Omega^c}\left(\bs H_{T}\right) + \mathcal{P}_{\Omega^c}\left(\bs H_{T^\perp}\right), \bs x_0\text{\upshape sign}(\bs x_{S^c})\rangle\|\bs x_0\|^{-1}_\infty\\
    &\geq \langle \bs x_{S^c}\bs x_0^* + \bs x_0\bs x_{S^c} + \bs H_{T^\perp }, \bs x_0\text{\upshape sign}(\bs x_{S^c})\rangle \|\bs x_0\|^{-1}_\infty\\
    &\geq \langle \bs x_{S^c}\bs x_0^* + \bs x_0\bs x_{S^c} , \bs x_0 \text{\upshape sign}(\bs x_{S^c})\rangle \|\bs x_0\|^{-1}_\infty \\
    &\geq \|\bs x_0\|_{2}\|\bs x_{S^c}\|_{1} \|\bs x_0\|^{-1}_\infty\\
&= 2\mu_0^{-1}\|\bs x_{S^c}\|_{1}\label{lowerBoundl1NormFinal01}
\end{align}

Finally to relate the norm of $\bs H_{\Omega}$ to $\bs H_{\Omega^c}$ and $\bs H_{T^\perp\cap \Omega}$, we will need Propositions~\ref{propositionInjectivityOnTcapOmega} , ~\ref{relationTTperpOmegaLinearMeasurements} and~\ref{PropositionAstarAHOmegaComplement} below which are respectively proved in sections~\ref{sectionpropositionInjectivityOnTcapOmega}, ~\ref{sectionrelationTTperpOmegaLinearMeasurements} and~\ref{sectionProofPropositionAstarAHOmegaComplement}. We start by showing injectivity of the linear transformation $\mathcal{H}$ whose action on symmetric matrices is defined through
\begin{align}
\begin{array}{lll}
\mathcal{S}^{n+1} &\rightarrow& \mathbb{R}^{m+1}\\
\bs X & \mapsto & \displaystyle m^{-1}\sum_{i=1}^m (\bs e_1\bs c_i^* + \bs c_i\bs e_1^*) \langle \left(\bs e_1\bs c_i^* + \bs c_i\bs e_1^*\right), \bs X \rangle+ 2(\bs e_1\bs e_1^*)\langle \bs e_1\bs e_1^*, \bs X \rangle \\
& &- 2\bs I\langle \bs e_1\bs e_1^*, \bs X \rangle - 2\bs e_1\bs e_1^* \langle \bs I, \bs X\rangle  + 6\bs I \langle \bs I, \bs X\rangle  
\end{array}\label{definitionMathcalH}
\end{align}
In what follows, we will use $\mathcal{D}(\bs X)$ to denote the deterministic part of $\mathcal{H}$. I.e. 
\begin{align}
\mathcal{D}(\bs X)& = 2(\bs e_1\bs e_1^*)\langle \bs e_1\bs e_1^*, \bs X \rangle - 2\bs I\langle \bs e_1\bs e_1^*, \bs X \rangle - 2\bs e_1\bs e_1^* \langle \bs I, \bs X\rangle  + 6\bs I \langle \bs I, \bs X\rangle  
\end{align}

\begin{restatable}{proposition}{restateInjectivityOnTcapOmega}
\label{propositionInjectivityOnTcapOmega}
Let $\mathcal{H}$ be defined as in~\eqref{definitionMathcalH}. Let $T$ and $T^\perp$ denote the tangent space to the positive semidefinite cone at $\bs X_0 = \bs x_0\bs x_0^*$ where $\bs x_0 = [1, \tilde{\bs x}_0]$ and $\tilde{\bs x}_0$ is a $k$-sparse vector of the unit sphere, i.e.  $\tilde{\bs x}_0\in \mathbb{S}^{n+1}$. We let $\Omega$ denote the support of $\bs X_0$ and use $\mathcal{P}_{T\cap \Omega}$ to encode the orthogonal projector onto the intersection $T\cap \Omega$. Then as soon as $m\gtrsim k\delta^{-1}$, with probability at least $1-n^{-\alpha}$ we have 
\begin{align}
\left\|\mathcal{P}_{T\cap \Omega}\mathcal{H}\mathcal{P}_{T\cap \Omega}(\bs X) - \mathcal{P}_{T\cap \Omega}(\bs X)\right\|\lesssim \delta\|\bs X\|_F
\end{align}

\end{restatable}

\begin{restatable}{proposition}{restaterelationTTperpOmegaLinearMeasurements}
\label{relationTTperpOmegaLinearMeasurements}
Let $\mathcal{H}$ be defined as in~\eqref{definitionMathcalH}. Let $T$ and $T^\perp$ denote the tangent space to the positive semidefinite cone at $\bs X_0 = \bs x_0\bs x_0^*$ where $\bs x_0 = [1, \tilde{\bs x}_0]$ and $\tilde{\bs x}_0$ is a $k$-sparse vector of the unit sphere, i.e.  $\tilde{\bs x}_0\in \mathbb{S}^{n+1}$. We let $\Omega$ denote the support of $\bs X_0$ and use $\mathcal{P}_{T\cap \Omega}$ to encode the orthogonal projector onto the intersection $T\cap \Omega$. Then as soon as $m\gtrsim k\delta^{-1}$, with probability at least $1-n^{-\alpha}$ , for any matrix $\bs X\in T^\perp\cap \Omega$, we have 
\begin{align}
\left\|\mathcal{P}_{T\cap \Omega}\mathcal{H}(\bs X) \right\|\lesssim (1+ \delta)\|\bs X\|_F
\end{align}

\end{restatable}

Combining Propositions~\ref{propositionInjectivityOnTcapOmega} and~\ref{relationTTperpOmegaLinearMeasurements}, and noting that
\begin{align}
\left\|\left(\mathcal{P}_{T\cap \Omega}\mathcal{H}\mathcal{P}_{T\cap \Omega}\right)^{-1}\right\| 
&= \left\|\left(\imath - \left(\imath - \mathcal{P}_{T\cap \Omega}\mathcal{H}\mathcal{P}_{T\cap \Omega}\right)\right)^{-1}\right\| \\
&= \sum_{k=0}^\infty \left\|\imath - \mathcal{P}_{T\cap \Omega}\mathcal{H}\mathcal{P}_{T\cap \Omega}\right\|^k\\
& =  \frac{1}{1-\delta}
\end{align}
with probability at least $1-n^{-\alpha}$ as soon as $m\gtrsim k\log^{\alpha}(n)$.  From this, we can write 
\begin{align}
\left\|\bs H_{T\cap \Omega}\right\|_F& =\left\| \left(\mathcal{P}_{T\cap \Omega}\mathcal{H}\mathcal{P}_{T\cap \Omega}\right)^{-1} \mathcal{H}\bs H_{T\cap \Omega}\right\|\\
&\leq \left\|\left(\mathcal{P}_{T\cap \Omega}\mathcal{H}\mathcal{P}_{T\cap \Omega}\right)^{-1}\right\|\left\|\mathcal{P}_{T\cap \Omega}\mathcal{H}\bs H_{T\cap \Omega}\right\|_F\\
&\leq  \left\|\mathcal{P}_{T\cap \Omega}\mathcal{H}\bs H_{T\cap \Omega}\right\|_F\\
&\leq  \left\|\mathcal{P}_{T\cap \Omega}\mathcal{H}\left(\bs H - \bs H_{T^\perp \cap \Omega} - \bs H_{\Omega^c}\right)\right\|_F\\
&\lesssim \|\bs H_{T^\perp\cap \Omega}\|_F + \|\mathcal{P}_{T\cap \Omega}\mathcal{H}\bs H_{\Omega^c}\|_F\label{relationBetweenHTOmegaHTperpOmegaHOmegaComplement001}
\end{align}

To obtain a finer control on the second term in~\eqref{relationBetweenHTOmegaHTperpOmegaHOmegaComplement001}, we use Proposition~\ref{PropositionAstarAHOmegaComplement} below which is proved in section~\ref{sectionProofPropositionAstarAHOmegaComplement}.

\begin{restatable}{proposition}{restatePropositionAstarAHOmegaComplement}
\label{PropositionAstarAHOmegaComplement}
Let $\mathcal{H}$ be defined as in~\eqref{definitionMathcalH}. Let $T$ and $T^\perp$ denote the tangent space to the positive semidefinite cone at $\bs X_0 = \bs x_0\bs x_0^*$ where $\bs x_0 = [1, \tilde{\bs x}_0]$ and $\tilde{\bs x}_0$ is a $k$-sparse vector of the unit sphere, i.e.  $\tilde{\bs x}_0\in \mathbb{S}^{n+1}$. We let $\Omega$ denote the support of $\bs X_0$ and use $\mathcal{P}_{T\cap \Omega}$ to encode the orthogonal projector onto the intersection $T\cap \Omega$. Then as soon as $m\gtrsim k\delta^{-1}$, (up to log factors) for any matrix $\bs H\in \Omega^c$, with probability at least $1-n^{-\alpha}$, if we let 
\begin{align}
\bs H = \left[\begin{array}{cc}
1& \bs h^*\\
\bs h & \tilde{\bs H}
\end{array}\right],
\end{align}
 we have 
\begin{align}
\left\|\mathcal{P}_{T\cap \Omega}\mathcal{H}(\bs H) - \mathbb{E}\mathcal{P}_{T\cap \Omega}\mathcal{H}(\bs H)\right\|_F\leq  \delta \|\bs h_{S^c}\|_2 + 8\text{\upshape Tr}(\bs H_{B})
\end{align}
with probability at least $1-n^{-\alpha}$ as soon as $m\gtrsim \delta^{-1} k$. 

\end{restatable}

Combining Proposition~\ref{PropositionAstarAHOmegaComplement} with~\eqref{relationBetweenHTOmegaHTperpOmegaHOmegaComplement001}, and letting $\bs H_T = \bs x_0\bs x^* + \bs x\bs x_0^*$,  we finally get 
\begin{align}
\left\|\bs H_{T\cap \Omega}\right\|_F&\leq 2 \|\bs H_{T^\perp\cap \Omega}\|_F + \delta\|\bs h_{S^c}\|_2 + \text{\upshape Tr}(\bs H_B)\\
&\leq 2 \|\bs H_{T^\perp\cap \Omega}\|_F + \delta\|\bs x_{S^c}\|_{\ell_1}+ \delta\text{Tr}(\bs H_{T}^\perp)+ \text{\upshape Tr}(\bs H_B)\label{connectionBetweenHTOmegaAndHTperpOmegaHOmegac}
\end{align} 
In the last line, we used the decomposition 
\begin{align}
\bs H_{\Omega^c} &= \mathcal{P}_{\Omega^c}\left(\bs H_T\right) + \mathcal{P}_{\Omega^c}\left(\bs H_T^\perp \right)\\
&= \bs x_0(\bs x_{S^c})^* + \bs x_{S^c}\bs x_0^* + \mathcal{P}_{\Omega^c}\left(\bs H_T^\perp \right)
\end{align}
where we define $\mathcal{P}_{\Omega^c}\left(\bs H_T^\perp \right)$ blockwise as 
\begin{align}
\mathcal{P}_{\Omega^c}\left(\bs H_{T}^\perp\right) = \left[\begin{array}{cc}
0 & \bs h_{T^\perp\cap \Omega^c}^*\\
 \bs h_{T^\perp\cap \Omega^c} & \tilde{\bs H}_{T^\perp\cap \Omega^c}
\end{array}\right]
\end{align}
Noting that 
\begin{align}
\|\bs h_{T^\perp \cap \Omega^c}\|_2 &\leq \|\mathcal{P}_{\Omega^c}\left(\bs H_T^\perp\right)\|_F\\
& \leq  \|\mathcal{P}_{\Omega^c}\left(\bs H_T^\perp\right)\|_F\\
&\leq   \|\bs H_T^\perp\|_F\\
&\leq  \|\bs H_T^\perp \|_1  = \text{Tr}(\bs H_{T}^\perp)
\end{align}
Substituting~\eqref{connectionBetweenHTOmegaAndHTperpOmegaHOmegac} into~\eqref{relationBetweenHTOmegaHTperpOmegaHOmegaComplement001}, we have
\begin{align}
\|\bs H_{T\cap \Omega}\|_F& \leq 3 \|\bs H_{T^\perp\cap \Omega}\|_F + \delta\|\bs x_{S^c}\|_{\ell_1}+ \delta\text{Tr}(\bs H_{T}^\perp)+ \text{\upshape Tr}(\bs H_B)\\
&\leq 3\text{\upshape Tr}\left(\bs H_{T^\perp\cap \Omega}\right) + \delta \text{\upshape Tr}(\bs H_{T^\perp}) + \delta\|\bs x_{S^c}\|_{\ell_1}+  \text{\upshape Tr}(\bs H_B)
\end{align}
as well as
\begin{align}
\|\bs H_{\Omega}\|_F\leq 4\text{\upshape Tr}\left(\bs H_{T^\perp\cap \Omega}\right) + \delta \text{\upshape Tr}(\bs H_{T^\perp}) + \delta\|\bs x_{S^c}\|_{\ell_1}+  \text{\upshape Tr}(\bs H_B)
\end{align}
Using this last bound together with~\eqref{lowerBoundl1NormFinal01} and  ~\eqref{afterApplyingHansonWright02}, we then obtain
\begin{align}
\|\bs X\|_{\ell_1}&\geq \|\bs X_0\|_{\ell_1} - \mu_0 (1+\log^{\alpha}(n))\|\bs H_\Omega\|_F \|\tilde{\bs x}_0\|_{\ell_1} - \mu_0\|\bs H_\Omega\|_F\sqrt{k} \|\tilde{\bs x}_0\|_{\ell_1} \\
&- \left\| \bs H_{ \Omega}\right\|_F \mu_0^2 (1+\log^{\alpha}(n)) \|\tilde{\bs x}_0\|^2_{\ell_1} - \mu_0^2 \|\bs H_\Omega\|_F \sqrt{k}\|\tilde{\bs x}_0\|^2_{\ell_1} \\
& +   \|\bs H_{\Omega^c}\|_{\ell_1}  - \|\bs H_{\Omega}\|_F \log(n) \left(1+2^{-1}\gamma \mu_0\right)\\
& + \gamma\text{Tr}(\bs H_{B}) + \gamma \text{Tr}(\bs H_{T^\perp\cap \Omega}) - 2^{-1}\gamma \mu_0^2 \left(1+\log^{\alpha}(n)\right) \|\bs H_{\Omega}\|_F\\
&- 2^{-1}\gamma \mu_0^2 \sqrt{k} \|\bs H_{\Omega}\|_F \\
&\geq \|\bs X_0\|_{\ell_1} \\
&- \mu_0 (1+\log^{\alpha}(n))\left(4\text{\upshape Tr}\left(\bs H_{T^\perp\cap \Omega}\right) + \delta \text{\upshape Tr}(\bs H_{T^\perp}) + \delta\|\bs x_{S^c}\|_{\ell_1}+  \text{\upshape Tr}(\bs H_B)\right)  \|\tilde{\bs x}_0\|_{\ell_1} \\
&- \mu_0 \sqrt{k}\left(4\text{\upshape Tr}\left(\bs H_{T^\perp\cap \Omega}\right) + \delta \text{\upshape Tr}(\bs H_{T^\perp}) + \delta\|\bs x_{S^c}\|_{\ell_1}+  \text{\upshape Tr}(\bs H_B)\right)  \|\tilde{\bs x}_0\|_{\ell_1} \\
&- \mu_0^2 (1+\log^{\alpha}(n)) \left(4\text{\upshape Tr}\left(\bs H_{T^\perp\cap \Omega}\right) + \delta \text{\upshape Tr}(\bs H_{T^\perp}) + \delta\|\bs x_{S^c}\|_{\ell_1}+  \text{\upshape Tr}(\bs H_B)\right)  \|\tilde{\bs x}_0\|^2_{\ell_1} \\
&- \mu_0^2 \sqrt{k}\left(4\text{\upshape Tr}\left(\bs H_{T^\perp\cap \Omega}\right) + \delta \text{\upshape Tr}(\bs H_{T^\perp}) + \delta\|\bs x_{S^c}\|_{\ell_1}+  \text{\upshape Tr}(\bs H_B)\right)  \|\tilde{\bs x}_0\|^2_{\ell_1} \\
&+ \|\bs H_{\Omega^c}\|_{\ell_1}\\
&  - \log(n) \left(1+2^{-1}\gamma \mu_0\right)\left(4\text{\upshape Tr}\left(\bs H_{T^\perp\cap \Omega}\right) + \delta \text{\upshape Tr}(\bs H_{T^\perp}) + \delta\|\bs x_{S^c}\|_{\ell_1}+  \text{\upshape Tr}(\bs H_B)\right)\\
&- 2^{-1}\gamma \mu_0^2 \left(1+\log^{\alpha}(n)\right)  \left(4\text{\upshape Tr}\left(\bs H_{T^\perp\cap \Omega}\right) + \delta \text{\upshape Tr}(\bs H_{T^\perp}) + \delta\|\bs x_{S^c}\|_{\ell_1}+  \text{\upshape Tr}(\bs H_B)\right)\\
&- 2^{-1}\gamma \mu_0^2 \sqrt{k} \left(1+\log^{\alpha}(n)\right) \left(4\text{\upshape Tr}\left(\bs H_{T^\perp\cap \Omega}\right) + \delta \text{\upshape Tr}(\bs H_{T^\perp}) + \delta\|\bs x_{S^c}\|_{\ell_1}+  \text{\upshape Tr}(\bs H_B)\right) \\
&+ \gamma (1-\delta)\text{Tr}(\bs H_{B}) + \gamma (1-\delta) \text{Tr}(\bs H_{T^\perp\cap \Omega}) + \delta \gamma \text{\upshape Tr}(\bs H_{T^\perp}) \\
&\geq \|\bs X_0\|_{\ell_1} -|W_0|\left(4\text{\upshape Tr}\left(\bs H_{T^\perp\cap \Omega}\right) + \delta \text{\upshape Tr}(\bs H_{T^\perp}) + \delta\|\bs x_{S^c}\|_{\ell_1}+  \text{\upshape Tr}(\bs H_B)\right)\\
& + \|\bs H_{\Omega^c}\|_{\ell_1}  + \gamma (1-\delta)\text{Tr}(\bs H_{B}) + \gamma (1-\delta) \text{Tr}(\bs H_{T^\perp\cap \Omega}) + \delta \gamma \text{\upshape Tr}(\bs H_{T^\perp}) 
\end{align}
where we define $W_0$ as 
\begin{align}\begin{split}
W_0 &= \left(\mu_0 \|\tilde{\bs x}_0\|_{\ell_1} + \mu_0^2 \|\tilde{\bs x}_0\|_{\ell_1}^2 + 2^{-1}\gamma \mu_0^2 + 1+ 2^{-1}\gamma\mu_0  \right) (1+\log^{\alpha}(n))\\
&+  \left(\mu_0 \|\tilde{\bs x}_0\|_{\ell_1} + \mu_0^2 \|\tilde{\bs x}_0\|_{\ell_1}^2 + 2^{-1}\gamma \mu_0^2 \right) \sqrt{k}
\end{split}\label{definitionW0}
\end{align}
from which, using $\|\bs x_0\|_1\leq \mu_0 \sqrt{k}$ we can write $|W_0|\lesssim\left(\mu_0(k+\gamma) +( k+\gamma)\mu_0^2\sqrt{k}\right)$. Grouping the terms and rearranging, we get 
\begin{align}
\|\bs X\|_{\ell_1}&\geq\|\bs X_0\|_{\ell_1} + \left(\gamma(1-\delta) - 4|W_0|\right) \text{\upshape Tr}(\bs H_{T^\perp\cap \Omega})\label{secondToLast01}\\
&+ \left(\gamma(1-\delta) - |W_0|\right) \text{\upshape Tr}(\bs H_{B})\\
&+ \left(\delta \gamma -|W_0|\delta \right)\text{\upshape Tr}(\bs H_{T^\perp})\label{secondToLastEndMinus1}\\
&+ \left(\mu_0^{-1} - |W_0|\delta \right)\|\bs x_{S^c}\|_{\ell_1}\label{secondToLastEnd}
\end{align}
Taking $\gamma \gtrsim  \mu_0^2  k^{3/2}$, one can make the terms~\eqref{secondToLast01} to~\eqref{secondToLastEndMinus1} positive.  In particular, noting that $\mu_0\sqrt{k}>1$, and taking $\gamma  \asymp \mu_0^2  k^{3/2}$, for the last term to be positive, we need
\begin{align}
1 \geq \mu_0^3   k^{3/2} \delta
\end{align}
which gives the conclusion.

\subsection{\label{auxialliaryResultsLinearMeasMatrixL1}Auxilliary results}

\subsubsection{\label{sectionpropositionInjectivityOnTcapOmega}Proof of Proposition~\ref{propositionInjectivityOnTcapOmega}}

We start by recalling the statement of the proposition below.

\restateInjectivityOnTcapOmega*

\begin{proof}
We will show that the operator $\mathcal{P}_{T\cap \Omega}\mathcal{A}^*\mathcal{A}\mathcal{P}_{T\cap \Omega}$ defined on $T\cap \Omega$ as 
\begin{align*}
\mathcal{P}_{T\cap \Omega}\mathcal{A}^*\mathcal{A}\mathcal{P}_{T\cap \Omega}\; :\; \bs X \mapsto \mathcal{P}_{T\cap \Omega}\left(\bs e_1\bs e_1^*\right)\langle \bs e_1\bs e_1^*, \bs X  \rangle + \frac{1}{2} \sum_{i=1}^m \mathcal{P}_{T\cap \Omega}\left(\bs e_1\bs c_i^* + \bs c_i\bs e_1^*\right) \langle \bs e_1\bs c_i^* + \bs c_i\bs e_1^* , \bs X\rangle  
\end{align*}
satisfies 
\begin{align}
\left\|\mathcal{P}_{T\cap \Omega}\mathcal{A}^*\mathcal{A}\mathcal{P}_{T\cap \Omega} (\bs X) - \bs X\right\|_F \lesssim \sqrt{\frac{k}{m}}\|\bs X\|_F
\end{align}
with probability at least $1-n^{-\alpha}$ for any $\bs X\in T\cap \Omega$.  First note that any matrix in $T$ has rank at most 2.  We define $
\bs Y$ as
\begin{align}
    \mathcal{P}_{T\cap \Omega}(\bs Y)& = \frac{1}{m} \sum_{i=1}^m \mathcal{P}_{T\cap \Omega}\left(\bs e_1\bs c_i^* + \bs c_i\bs e_1^*\right) \langle \mathcal{P}_{T\cap \Omega}\left(\bs c_i\bs e_1^* + \bs e_1\bs c_i^*\right), \bs u\bs u^*\rangle  = \frac{1}{m} \sum_{i=1}^m \xi_i \bs Y_i\label{sumscertificates01Asymmetric}
\end{align}
where $\bs u$ is any vector of $\mathbb{S}^{n+1}$ with $|\text{\upshape supp}(\bs u)| = k$. Using $\|\bs x_0\|_2^2 = 2$
\begin{align}
    \bs Y_{i}& = \frac{1}{2} \langle \bs x_0, \bs c_i\rangle \left(\bs x_0\bs e_1^* + \bs e_1\bs x_0^*\right) +\frac{1}{2} \left(\bs x_0\bs c_{i\Omega}^* + \bs c_{i\Omega}\bs x_0^*\right) - \frac{1}{2}\langle \bs x_0, \bs c_i\rangle \bs x_0\bs x_0^* \\
    \xi_{i} &= \frac{1}{2}\langle \bs x_0, \bs c_i\rangle \left(\bs x_0\bs e_1^* + \bs e_1\bs x_0^*\right) + \frac{1}{2}\left(\bs x_0\bs c_{i\Omega}^* +\bs c_{i\Omega}\bs x_0^*\right) - \frac{1}{2}\langle \bs x_0, \bs c_i\rangle \bs x_0\bs x_0, \bs u\bs u^* \rangle   \\
    & =  u_1 \langle \bs x_0, \bs u\rangle  \langle \bs x_0, \bs c_i\rangle + \langle \bs x_0, \bs u\rangle \langle \bs c_{i\Omega}, \bs u\rangle  - \frac{1}{2}|\langle \bs x_0, \bs u\rangle |^2 \langle \bs x_0, \bs c_i\rangle 
\end{align}
We start by showing that $\mathbb{E}\bs Y = \mathcal{P}_{T\cap \Omega }(\bs u\bs u^*)$. We decompose the sum~\eqref{sumscertificates01Asymmetric} as $\sum_{i=1}^m \xi_{i}\bs Y_{i} = \bs H_1+\bs H_2+\bs H_3$ with
\begin{align}
    \mathbb{E}\bs H_1& = \mathbb{E}\left\{\frac{1}{2m }\sum_{i=1}^m\langle \bs x_0, \bs c_i \rangle \left(\bs x_0\bs e_1^* + \bs e_1\bs x_0^*\right) \left[u_1\langle \bs x_0, \bs u\rangle \langle \bs x_0, \bs c_i\rangle  + \langle \bs x_0, \bs u\rangle\langle \bs c_i, \bs u\rangle\right]\right\}\\
    &+\mathbb{E}\left\{\frac{1}{2m}\sum_{i=1}^m\langle \bs x_0, \bs c_i \rangle \left(\bs x_0\bs e_1^* + \bs e_1\bs x_0^*\right)\left[- \frac{1}{2}|\langle \bs x_0, \bs u\rangle |^2 \langle \bs x_0, \bs c_i\rangle  \right]\right\}\\
    & =  \frac{1}{2}u_1 \langle \bs x_0, \bs u\rangle \|\tilde{\bs x}_0\|^2 \left(\bs x_0\bs e_1^*+\bs e_1\bs x_0^*\right)+ \frac{1}{2}\langle \bs x_0, \bs u\rangle \langle \tilde{\bs x}_0, \tilde{\bs u}\rangle \left(\bs x_0\bs e_1^* + \bs e_1\bs x_0^*\right)\\
    &- \frac{1}{4}\|\tilde{\bs x}_0\|^2 |\langle \bs x_0, \bs u\rangle |^2 \left(\bs x_0\bs e_1^* + \bs e_1\bs x_0^*\right)\\
    & =\frac{1}{4}\|\tilde{\bs x}_0\|^2 |\langle \bs x_0, \bs u\rangle |^2 \left(\bs x_0\bs e_1^* + \bs e_1\bs x_0^*\right)
\end{align}
\begin{align}
    \mathbb{E}\bs H_2& = \mathbb{E}\left\{\frac{1}{2m}\sum_{i=1}^m \left(\bs x_0\bs c_i^* + \bs c_i\bs x_0^*\right) \left[u_1\langle \bs x_0, \bs u\rangle \langle \bs x_0, \bs c_i\rangle +  \langle \bs x_0, \bs u\rangle \langle \bs c_i, \bs u\rangle \right]\right\}\\
    & +\mathbb{E}\left\{\frac{1}{2m}\sum_{i=1}^m \left(\bs x_0\bs c_i^* + \bs c_i\bs x_0^*\right) \left[-\frac{1}{2}|\langle \bs x_0, \bs u\rangle|^2 \langle \bs x_0, \bs c_i\rangle  \right]\right\}\\
    & = \frac{1}{2}\left(\bs x_0\tilde{\bs x}_0^* + \tilde{\bs x}_0\bs x_0^*\right) u_1 \langle \bs x_0, \bs u\rangle+ \frac{1}{2}\langle \bs x_0, \bs u\rangle \left(\bs x_0\tilde{\bs u}_\Omega  + \tilde{\bs u}_\Omega\bs x_0^*\right)\\
&- \frac{1}{4}|\langle \bs x_0, \bs u\rangle|^2 \left(\bs x_0\tilde{\bs x}_0^* +\tilde{\bs x}_0\bs x_0^*\right)
\end{align}
\begin{align}
    \mathbb{E}\bs H_3& = -\frac{1}{2}\mathbb{E}\left\{\frac{1}{m}\sum_{i=1}^m \bs x_0\bs x_0^* \langle \bs x_0, \bs c_i\rangle \left[u_1\langle \bs x_0, \bs u\rangle \langle \bs x_0, \bs c_i\rangle + \langle \bs x_0, \bs u\rangle \langle \bs c_i, \bs u\rangle   \right]\right\}\\
    &  -\frac{1}{2}\mathbb{E}\left\{\frac{1}{m}\sum_{i=1}^m \bs x_0\bs x_0^* \langle \bs x_0, \bs c_i\rangle \left[-\frac{1}{2}|\langle \bs x_0, \bs u\rangle |^2 \langle \bs x_0, \bs c_i\rangle \right] \right\}\\
    & = -\frac{1}{2} \bs x_0\bs x_0^* \|\tilde{\bs x}_0\|^2_2 u_1 \langle \bs x_0, \bs u\rangle - \frac{1}{2}\langle \bs x_0, \bs u\rangle \bs x_0\bs x_0^* \langle \tilde{\bs x}_0, \tilde{\bs u}\rangle + \frac{1}{4}|\langle \bs x_0, \bs u\rangle |^2 \bs x_0\bs x_0^* \|\tilde{\bs x}_0\|^2\\
    & = -\frac{1}{4}|\langle \bs x_0, \bs u\rangle |^2 \bs x_0\bs x_0^* \|\tilde{\bs x}_0\|^2
\end{align}
\begin{align}
    \mathbb{E}\frac{1}{m}\sum_{i=1}^m \xi_{i}\bs Y_{i}& =\frac{1}{4}\|\tilde{\bs x}_0\|^2 |\langle \bs x_0, \bs u\rangle |^2 \left(\bs x_0\bs e_1^* + \bs e_1\bs x_0^*\right) -\frac{1}{4}|\langle \bs x_0, \bs u\rangle |^2 \bs x_0\bs x_0^* \|\tilde{\bs x}_0\|^2\\
&+ \frac{1}{2}\left(\bs x_0\tilde{\bs x}_0^* + \tilde{\bs x}_0\bs x_0^*\right) u_1 \langle \bs x_0, \bs u\rangle+ \frac{1}{2}\langle \bs x_0, \bs u\rangle \left(\bs x_0\tilde{\bs u}_\Omega  + \tilde{\bs u}_\Omega\bs x_0^*\right)\\
&- \frac{1}{4}|\langle \bs x_0, \bs u\rangle|^2 \left(\bs x_0\tilde{\bs x}_0^* +\tilde{\bs x}_0\bs x_0^*\right)
\end{align}
Adding $2\mathcal{P}_{T\cap \Omega}(\bs e_1\bs e_1^*)\langle \mathcal{P}_{T\cap \Omega}(\bs e_1\bs e_1^*), \bs u\bs u^* \rangle $,
\begin{align}
    & \mathcal{P}_{T\cap \Omega}\left(\bs e_1\bs e_1^*\right) \langle \mathcal{P}_{T\cap \Omega}\left(\bs e_1\bs e_1^*\right), \bs u\bs u^*\rangle \\
    & =  \mathcal{P}_{T\cap \Omega}\left(\bs e_1\bs e_1^*\right)\left(\frac{1}{2}\langle \bs x_0, \bs u\rangle u_1 - \frac{1}{4} |\langle \bs x_0, \bs u\rangle |^2 +  \frac{1}{2}u_1 \langle \bs x_0, \bs u\rangle   \right)\\
    & =\frac{1}{2} \left(\bs e_1\bs x_0^* + \bs x_0\bs e_1^* - \frac{1}{2}\bs x_0\bs x_0^*\right) \left(\langle \bs x_0, \bs u\rangle u_1 - \frac{1}{4} |\langle \bs x_0, \bs u\rangle |^2    \right)
\end{align}
we get 
\begin{align}
    &\mathbb{E}\bs Y +  \mathcal{P}_{T\cap \Omega}\left(\bs e_1\bs e_1^*\right) \langle \mathcal{P}_{T\cap \Omega}\left(\bs e_1\bs e_1^*\right), \bs u\bs u^*\rangle\label{developingAstarAweird01}\\
     =& \frac{1}{4}\|\tilde{\bs x}_0\|^2 |\langle \bs x_0, \bs u\rangle |^2 \left(\bs x_0\bs e_1^* + \bs e_1\bs x_0^*\right) -\frac{1}{4}|\langle \bs x_0, \bs u\rangle |^2 \bs x_0\bs x_0^* \|\tilde{\bs x}_0\|^2 \\
&+\frac{1}{2}\left(\bs x_0\tilde{\bs x}_0^* + \tilde{\bs x}_0\bs x_0^*\right) u_1 \langle \bs x_0, \bs u\rangle+ \frac{1}{2}\langle \bs x_0, \bs u\rangle \left(\bs x_0\tilde{\bs u}_\Omega  + \tilde{\bs u}_\Omega\bs x_0^*\right)\\
&- \frac{1}{4}|\langle \bs x_0, \bs u\rangle|^2 \left(\bs x_0\tilde{\bs x}_0^* +\tilde{\bs x}_0\bs x_0^*\right)\\
&+ \left(\bs e_1\bs x_0^* + \bs x_0\bs e_1^* - \frac{1}{2}\bs x_0\bs x_0^*\right) \left(\langle \bs x_0, \bs u\rangle u_1 - \frac{1}{4} |\langle \bs x_0, \bs u\rangle |^2    \right) \\
 =& - \frac{1}{2}|\langle \bs x_0, \bs u\rangle|^2 \left(\bs x_0\tilde{\bs x}_0^* +\tilde{\bs x}_0\bs x_0^*\right)\\
& +\frac{1}{2}\left(\bs x_0\tilde{\bs x}_0^* + \tilde{\bs x}_0\bs x_0^*\right) u_1 \langle \bs x_0, \bs u\rangle+ \frac{1}{2}\langle \bs x_0, \bs u\rangle \left(\bs x_0\tilde{\bs u}_\Omega  + \tilde{\bs u}_\Omega\bs x_0^*\right)\\
& + \left(\bs e_1\bs x_0^* + \bs x_0\bs e_1^* - \frac{1}{2}\bs x_0\bs x_0^*\right) \left(\langle \bs x_0, \bs u\rangle u_1 - \frac{1}{4} |\langle \bs x_0, \bs u\rangle |^2    \right)\\
 =& \bs x_0\bs x_0^* u_1 \langle \bs x_0, \bs u\rangle + \langle \bs x_0, \bs u \rangle \left(\bs x_0\bs u_{\Omega}^* + \bs u_{\Omega} \bs x_0^*\right)  \frac{1}{2}\\
& - \frac{1}{2} |\langle \bs x_0, \bs u\rangle |^2 \left(\bs x_0\tilde{\bs x}_0^* + \tilde{\bs x}_0\bs x_0^*\right)\\
& - \frac{1}{2} \bs x_0\bs x_0^* \langle \bs x_0, \bs u\rangle u_1 + \frac{1}{8} \bs x_0\bs x_0^* |\langle \bs x_0, \bs u\rangle |^2 - \frac{1}{4} |\langle \bs x_0, \bs u\rangle |^2 \left(\bs x_0\bs e_1^* + \bs e_1\bs x_0^*\right)\\
= & \mathcal{P}_{T\cap \Omega} \left(\bs u\bs u^*\right) + \frac{1}{4} |\langle \bs x_0, \bs u\rangle |^2 \bs x_0\bs x_0^* + \bs x_0\bs x_0^* u_1 \langle \bs x_0, \bs u\rangle\\
& - \frac{1}{2} |\langle \bs x_0, \bs u\rangle |^2 \bs x_0\bs x_0^* - \frac{1}{4} |\langle \bs x_0, \bs u\rangle |^2 \left(\bs x_0\tilde{\bs x}_0^* + \tilde{\bs x}_0\bs x_0^*\right)\\
& - \frac{1}{2} \bs x_0\bs x_0^* \langle \bs x_0, \bs u\rangle u_1 + \frac{1}{8} \bs x_0\bs x_0^* |\langle \bs x_0, \bs u\rangle |^2 \\
 =& \mathcal{P}_{T\cap \Omega}(\bs u\bs u^*) + \frac{1}{2}\bs x_0\bs x_0^* u_1 \langle \bs x_0, \bs u\rangle - \frac{1}{8} |\langle \bs x_0, \bs u\rangle |^2 \bs x_0\bs x_0^* \\
& - \frac{1}{4} |\langle \bs x_0, \bs u\rangle |^2 \left(\bs x_0\tilde{\bs x}_0^* + \tilde{\bs x}_0\bs x_0^*\right)   \\
 =& \mathcal{P}_{T\cap \Omega}\left(\bs u\bs u^*\right) + \frac{1}{2} \bs x_0\bs x_0^* u_1 \langle \bs x_0, \bs u\rangle - \frac{1}{8} |\langle \bs x_0, \bs u\rangle |^2 \bs x_0\bs x_0^* \\
&- \frac{1}{2} |\langle \bs x_0, \bs u\rangle |^2 \bs x_0\bs x_0^* \\
&+ \frac{1}{4} |\langle \bs x_0, \bs u\rangle |^2 \left(\bs x_0\bs e_1^* + \bs e_1\bs x_0^*\right)\\
 =& \mathcal{P}_{T\cap \Omega} \left(\bs u\bs u^*\right) + 2 \mathcal{P}_{T\cap \Omega}\left(\bs e_1\bs e_1^*\right) \langle \bs I_{T\cap \Omega}, \bs u\bs u^*\rangle + \frac{1}{8} |\langle \bs x_0, \bs u\rangle |^2 \bs x_0\bs x_0^* \\
& - \frac{1}{8} |\langle \bs x_0, \bs u\rangle |^2 \bs x_0\bs x_0^* \\
&+ \frac{1}{2} \bs x_0\bs x_0^* u_1 \langle \bs x_0, \bs u\rangle - \frac{1}{2} |\langle \bs x_0, \bs u\rangle |^2 \bs x_0\bs x_0^*\\
 = &\mathcal{P}_{T\cap \Omega} \left(\bs u\bs u^*\right) + 2\mathcal{P}_{T\cap \Omega} \left(\bs e_1\bs e_1^*\right) \langle \bs I_{T\cap \Omega}, \bs u\bs u^*\rangle\\
&+ 2\bs I_{T\cap \Omega}    \langle \mathcal{P}_{T\cap \Omega}(\bs e_1\bs e_1^*), \bs u\bs u^* \rangle + \frac{1}{8}\bs x_0\bs x_0^* \langle\bs x_0\bs x_0^*, \bs u\bs u^*\rangle  \\
 =& \mathcal{P}_{T\cap \Omega}\left(\bs u\bs u^*\right) +2 \mathcal{P}_{T\cap \Omega} \left(\bs e_1\bs e_1^*\right)\langle \bs I_{T\cap \Omega}, \bs u\bs u^*\rangle + 2\bs I_{T\cap \Omega} \langle \mathcal{P}_{T\cap \Omega}(\bs e_1\bs e_1^*), \bs u\bs u^*\rangle\\
& - \frac{3}{8} |\langle \bs x_0, \bs u\rangle |^2 \bs x_0\bs x_0^*\\
 =& \mathcal{P}_{T\cap \Omega}\left(\bs u\bs u^*\right) +2 \mathcal{P}_{T\cap \Omega} \left(\bs e_1\bs e_1^*\right)\langle \bs I_{T\cap \Omega}, \bs u\bs u^*\rangle + 2\bs I_{T\cap \Omega} \langle \mathcal{P}_{T\cap \Omega}(\bs e_1\bs e_1^*), \bs u\bs u^*\rangle\\
& - 6 \bs I_{T\cap \Omega} \langle \bs I_{T\cap \Omega}, \bs u\bs u^*\rangle 
 \label{ExpectationProjectionTOmegaLambda1}
\end{align}
From~\eqref{developingAstarAweird01} to~\eqref{ExpectationProjectionTOmegaLambda1}, we can thus write
\begin{align}
\mathcal{P}_{T\cap \Omega}\mathcal{H}(\bs u\bs u^*)\mathcal{P}_{T\cap \Omega} =& \frac{1}{m } \sum_{i=1}^m \mathcal{P}_{T\cap \Omega}\left(\bs c_i\bs e_1^* + \bs e_1\bs c_i^*\right) \langle \mathcal{P}_{T\cap \Omega}\left(\bs c_i\bs e_1^* + \bs e_1\bs c_i^*\right), \bs u\bs u^*\rangle \\
& + 2\mathcal{P}_{T\cap \Omega}\left(\bs e_1\bs e_1^*\right)\langle \mathcal{P}_{T\cap \Omega}\left(\bs e_1\bs e_1^*\right), \bs u\bs u^*\rangle \\
& - 2  \bs I_{T\cap \Omega} \langle \mathcal{P}_{T\cap \Omega}\left(\bs e_1\bs e_1^*\right), \bs u\bs u^*\rangle - 2  \mathcal{P}_{T\cap \Omega}\left(\bs e_1\bs e_1^*\right) \langle\bs I_{T\cap \Omega} \bs u\bs u^*\rangle \\
&+ 6 \bs I_{T\cap \Omega}\langle \bs I_{T\cap \Omega}, \bs u\bs u^*\rangle 
\end{align}
for which 
\begin{align}
\mathbb{E}\mathcal{P}_{T\cap \Omega}\mathcal{H}(\bs u\bs u^*)\mathcal{P}_{T\cap \Omega}  = \mathcal{P}_{T\cap \Omega}(\bs u\bs u^*) 
\end{align}
We now show concentration of this operator around its expectation.  
Let $Z_i$ to denote the variable 
\begin{align*}
Z_i &= \mathcal{P}_{T\cap \Omega}\left(\bs c_i\bs e_1^* + \bs e_1\bs c_i^*\right)\langle \mathcal{P}_{T\cap \Omega}\left(\bs c_i\bs e_1^* + \bs e_1\bs c_i^*, \bs u\bs u^*\right)\rangle 
\end{align*}
From this we have
\begin{align}
\mathbb{E} \sum_{i=1}^m Z_iZ_i^* & = \frac{1}{m} \mathbb{E} \sum_{i=1}^m \mathcal{P}_{T\cap \Omega} \left(\bs c_i\bs e_1^* + \bs e_1\bs c_i^*\right) \mathcal{P}_{T\cap \Omega}\left(\bs c_i\bs e_1^* + \bs e_1\bs c_i^*\right) \left|\langle \mathcal{P}_{T\cap \Omega}\left(\bs e_1\bs c_i^* + \bs c_i\bs e_1^*\right), \bs u\bs u^*\rangle \right|^2
\end{align}
we now let $\bs H$ to denote the matrix 
\begin{align}
\bs H& = \langle \bs x_0, \bs c_i\rangle \left(\bs x_0\bs e_1^* + \bs e_1\bs x_0^*\right) + \bs c_i\bs x_0^* + \bs x_0\bs c_i^* - 2\langle \bs x_0, \bs c_i\rangle \bs x_0\bs x_0^*  
\end{align}
Moreover, we let $h_i$ to denote the scalar $\langle \bs H, \bs u\bs u^*\rangle $. I.e. 
\begin{align}
h_i& = \langle \langle \bs x_0, \bs c_i\rangle \left(\bs x_0\bs e_1^* + \bs e_1\bs x_0^*\right) + \left(\bs c_i\bs x_0^* + \bs x_0\bs c_i^*\right) - 2\langle \bs x_0, \bs c_i\rangle \bs x_0\bs x_0^*  , \bs u\bs u^*   \rangle\\
& = 2 \langle \bs x_0, \bs c_i\rangle \langle \bs x_0, \bs u\rangle u_1 + 2\langle \bs c_i, \bs u\rangle \langle \bs x_0, \bs u\rangle - 2|\langle \bs x_0, \bs u\rangle |^2 \langle \bs x_0, \bs c_i\rangle   
\end{align}
From this, we have 
\begin{align}
h_i^2 & \lesssim  |\langle \bs x_0, \bs c_i \rangle |^2 u_1^2 |\langle \bs x_0, \bs u\rangle |^2 + |\langle \bs c_i, \bs u\rangle |^2 |\langle \bs x_0, \bs u\rangle |^2 + |\langle \bs x_0, \bs u\rangle |^4 |\langle \bs x_0, \bs c_i\rangle |^2
\end{align}
Similarly,  noting that for any two matrices $\bs A$ and $\bs B$ we have $\bs A\bs A^* + \bs B\bs B^* \succeq 0 -\bs A\bs B^* - \bs B\bs A^*$ as well as $\bs A\bs A^* + \bs B\bs B^*\succeq \bs A\bs B^* + \bs B\bs A^*$ hence $(\bs A+\bs B)(\bs A+\bs B)^*  \lesssim \bs A\bs A^* + \bs B\bs B^*$ and $\|\left(\bs A+\bs B\right)\left(\bs A+\bs B\right)^* \|\lesssim \|\bs A\bs A^* + \bs B\bs B^*\|$ and from this, that for any two scalars $a$ and $b$,  and psd matrix $\bs A$, we have $\|\bs A(a+b)^2\| \lesssim \|\bs A (a^2 + b^2)\|$, we can focus on the decomposition
\begin{align}
\left\|\mathbb{E} \left\{Z_iZ_i^* \right\}\right\| &\leq \left\|\mathbb{E} \left\{|\langle \bs x_0, \bs c_i\rangle |^4 u_1^2 |\langle \bs x_0, \bs u\rangle |^2 \left(\bs x_0\bs e_1^* + \bs e_1\bs x_0^*\right)\left(\bs x_0\bs e_1^* + \bs e_1\bs x_0^*\right)\right\}\right\| \\
&+ \left\|\mathbb{E}\left\{|\langle \bs x_0, \bs c_i\rangle |^2 |\langle \bs c_i, \bs u\rangle |^2 |\langle \bs x_0, \bs u\rangle |^2 \left(\bs x_0\bs e_1^* + \bs e_1\bs x_0^*\right) \left(\bs x_0\bs e_1^* + \bs e_1\bs x_0^*\right)  \right\}\right\|\\
& + \left\|\mathbb{E} \left\{|\langle \bs x_0, \bs u\rangle |^4 | \langle \bs x_0, \bs c_i\rangle |^4 \left(\bs x_0\bs e_1^* + \bs e_1\bs x_0^*\right)\left(\bs x_0\bs e_1^* + \bs e_1\bs x_0^*\right)\right\}\right\|\\
& + \left\|\mathbb{E} \left\{ |\langle \bs x_0, \bs c_i\rangle |^2 u_1^2 |\langle \bs x_0, \bs u\rangle |^2 \left(\bs c_i\bs x_0^* + \bs x_0\bs c_i^*\right) \left(\bs c_i\bs x_0^* + \bs x_0\bs x_0^*\right) \right\} \right\|\\
&+ \left\|\mathbb{E} \left\{|\langle \bs c_i, \bs u\rangle |^2 |\langle \bs x_0, \bs u\rangle |^2 \left(\bs c_i\bs x_0^* + \bs x_0\bs c_i^*\right) \left(\bs c_i\bs x_0^* + \bs x_0\bs c_i^*\right) \right\}\right\|\label{termcix001}\\
&+ \left\|\mathbb{E} \left\{|\langle \bs x_0, \bs u\rangle |^4 |\langle \bs x_0, \bs c_i\rangle |^2 \left(\bs c_i\bs x_0^* +\bs x_0\bs c_i^*\right) \left(\bs c_i\bs x_0^* + \bs x_0\bs c_i^*\right) \right\}\right\|\\
&+ \left\|\mathbb{E} \left\{|\langle \bs x_0, \bs c_i\rangle | 4 u_1^2 |\langle \bs x_0, \bs u\rangle |^2 \bs x_0\bs x_0^*\right\}\right\|\label{termcix002}\\
&+ \left\|\mathbb{E} \left\{|\langle \bs c_i, \bs u\rangle |^2 |\langle\bs x_0, \bs c_i \rangle |^2 |\langle \bs x_0, \bs u\rangle |^2 \bs x_0\bs x_0^* \right\}\right\|\label{termcix003}\\
&+ \left\|\mathbb{E} \left\{|\langle \bs x_0, \bs u\rangle |^4 |\langle \bs x_0, \bs c_i\rangle |^4 \bs x_0\bs x_0^*\right\}\right\|
\end{align}
Developing each of those terms, we get 
\begin{align}
&\left\|\mathbb{E} \left\{ |\langle \bs x_0, \bs c_i\rangle |^4 u_1^2 |\langle \bs x_0, \bs u\rangle |^2 \left(\bs x_0\bs x_0^* + \bs x_0\bs e_1^* + \bs e_1\bs x_0^* + \bs e_1\bs e_1^*\right)\right\} \right\|\\
&\lesssim u_1^2 |\langle \bs x_0, \bs u\rangle |^2 \left(|\langle \bs x_0\bs x_0^* , \bs x_0\bs x_0^*\rangle | - \|\bs x_0\|_4^4 + \|\bs x_0\|_2^4\right)\lesssim 1
\end{align}
\begin{align}
&\left\|\mathbb{E}\left\{|\langle \bs x_0, \bs c_i\rangle |^2 |\langle \bs c_i, \bs u \rangle|^2 |\langle \bs x_0, \bs u\rangle |^2 \left(\bs x_0\bs x_0^* + \bs x_0\bs e_1^* + \bs e_1\bs x_0^* + \bs e_1\bs e_1^*\right) \right\} \right\|\\
&\lesssim |\langle \bs x_0, \bs u\rangle |^2 \left\|\sum_{k\neq \ell} (x_0)_k (x_0)_\ell \bs u_k\bs u_\ell + \|\bs x_0\|^2 \|\bs u\|^2\right\|\\
&\lesssim |\langle \bs x_0, \bs u \rangle |^2 \left\|\sum_{k, \ell} (x_0)_k ( x_0)_\ell \bs u_k\bs u_\ell - \sum_{k} (x_0)_k^2 u_k^2 + \|\bs x_0\|^2 \|\bs u\|^2 \right\|\\
&\lesssim |\langle \bs x_0, \bs u\rangle |^2 \left(|\langle \bs x_0\bs x_0^*, \bs u\bs u^*\rangle | + \|\bs x_0\|^2 \|\bs u\|^2 \right)
\end{align}
\begin{align}
& \left\|\mathbb{E} \left\{|\langle \bs x_0, \bs u\rangle |^4 |\langle \bs x_0, \bs c_i\rangle |^4 \left(\bs x_0\bs x_0^* + \bs x_)\bs e-1^* + \bs e_1\bs x_0^* + \bs e_1\bs e_1^*\right) \right\} \right\|\\
&\lesssim |\langle \bs x_0, \bs u\rangle |^4 \left(|\langle \bs x_0\bs x_0^* , \bs x_0\bs x_0^*\rangle | - \|\bs x_0\|_4^4 + \|\bs x_0\|_2^4 \right)\lesssim 1
\end{align}
For~\eqref{termcix001} to~\eqref{termcix003}, expanding, we get
\begin{align}
\eqref{termcix001}&= \left\|\mathbb{E}\left\{|\langle \bs x_0, \bs c_i\rangle |^2 u_1^2 |\langle \bs x_0, \bs u\rangle |^2 \left(\left(\bs c_i\bs x_0^* + \bs x_0\bs c_i^*\right) \langle \bs c_i, \bs x_0\rangle +\bs x_0\bs x_0^* \|\bs c_i\|^2 + \bs c_i\bs c_i^* \|\bs x_0\|^2 \right) \right\}\right\|
\end{align}
Developing each of the terms, we obtain
\begin{align}
&\left\|\mathbb{E}\left\{|\langle \bs x_0, \bs c_i\rangle |^2 u_1^2 |\langle\bs x_0, \bs u\rangle |^2 \left(\bs c_i\bs x_0^* + \bs x_0\bs c_i^*\right) \langle \bs c_i,\bs x_0\rangle \right\}\right\|\\
&\leq  \left\| \left(\sum_{k\neq \ell} (x_0)_k (x_0)_\ell  \left(\bs e_k\bs x_0^* + \bs x_0\bs e_k^* \right) (x_0)_\ell \right) u_1^2 |\langle \bs x_0, \bs u\rangle |^2\right\|\\
&+\left\| \left(\sum_{k\neq \ell } (x_0)_k (x_0)_\ell \left(\bs e_\ell \bs x_0^* + \bs x_0\bs e_\ell^* \right) (x_0)_k \right) u_1^2 |\langle \bs x_0, \bs u\rangle |^2\right\|\\
& + \left\|\sum_k (x_0)_k^2 u_1 |\langle \bs x_0, \bs u\rangle |^2 \left(\bs e_k\bs x_0^* + \bs x_0\bs e_k^*\right) (x_0)_k\right\|\\
& \leq \left\|4\sum_{\ell} (\bs x_0\bs x_0^*) (x_0)_\ell^2 u_1^2 |\langle \bs x_0, \bs u\rangle |^2 + \sum_{k} (x_0)_k^2 \left(\bs e_k\bs x_0^* + \bs x_0\bs e_k^*\right) (x_0)_k\right\|\\
& +\left\| 4\sum_{k} (x_0)_k^2 \left(\bs e_k\bs x_0^* + \bs x_0\bs e_k^*\right) (x_0)_k u_1^2 |\langle \bs x_0, \bs u\rangle |^2\right\|\\
&\lesssim \|\bs x_0\|^2 \left\|\bs x_0\bs x_0^* u_1^2 |\langle \bs x_0, \bs u\rangle |^2\right\|+ \|\bs x_0\|^2 \left\|\bs x_0\bs x_0^* \right\| u_1^2 |\langle \bs x_0, \bs u\rangle |^2 \lesssim 1 \label{term2cix0Final01}
\end{align}
\begin{align}
&\left\|\mathbb{E}\left\{|\langle \bs x_0, \bs c_i\rangle |^2 u_1^2 |\langle \bs x_0, \bs u\rangle |^2 \bs x_0\bs x_0^* \|\bs c_i\|^2  \right\}\right\|\\
&\lesssim u_1^2 k \|\bs x_0\|^2 \left\|\bs x_0\bs x_0^*\right\| + \left\|\bs x_0\bs x_0^* \left(\langle \bs x_0\bs x_0^* , \mathbf{1}\mathbf{1}^*\rangle - \|\bs x_0\|^2 k\right)\right\|\lesssim k \label{term3cix0Final01}
\end{align}
\begin{align}
&\left\|\mathbb{E}\left\{|\langle \bs x_0, \bs c_i \rangle |^2 u_1^2 |\langle \bs x_0, \bs u\rangle |^2 \|\bs x_0\|^2 \bs c_i\bs c_i^* \right\}\right\|\\
&\lesssim u_1^2 |\langle \bs x_0, \bs u\rangle |^2 \left\|\bs x_0\bs x_0^* - \text{diag}\left(\bs x_0\bs x_0^*\right) + \|\bs x_0\|^2 \bs I \right\|\lesssim 1\label{term4cix0Final01}
\end{align}
In the line above, we used $\text{diag}\left(\bs x_0\bs x_0^*\right) = \sum_{k}\bs e_k\bs e_k^* (x_0)_k^2$. Grouping~\eqref{term2cix0Final01} to~\eqref{term4cix0Final01}, we get 
\begin{align}
\left\|\mathbb{E}\left\{|\langle \bs x_0, \bs c_i \rangle |^2 u_1^2 |\langle \bs x_0, \bs u\rangle |^2 \left(\bs c_i\bs x_0^* + \bs x_0\bs c_i^*\right)^2 \right\} \right\|\lesssim k 
\end{align}
Using similar ideas, we can write 
\begin{align}
&\left\|\mathbb{E}\left\{|\langle \bs c_i, \bs u\rangle |^2 |\langle \bs x_0, \bs u\rangle |^2 \left(\bs c_i\bs x_0^* + \bs x_0\bs c_i^*\right)^2\right\}\right\| \\
&\lesssim |\langle\bs x_0, \bs u \rangle |^2 \left\|\mathbb{E}\left\{|\langle \bs c_i, \bs u\rangle |^2 \left(\bs c_i\bs c_i^* \|\bs x_0\|^2  + \left(\bs c_i\bs x_0^* + \bs x_0\bs c_i^*\right) \langle \bs c_i, \bs x_0\rangle + \bs x_0\bs x_0^* \|\bs c_i\|^2\right)\right\}\right\|\\
&\lesssim |\langle \bs x_0, \bs u\rangle |^2 \left\|\left(\bs u\bs u^* - \text{diag}\left(\bs u\bs u^*\right) + \|\bs u\|^2\bs I\right)\|\bs x_0\|^2\right\|+ |\langle \bs x_0, \bs u\rangle |^2 k \|\bs u\|^2 \\
&+ |\langle \bs x_0, \bs u\rangle |^2 \left( \left\|2\sum_{k\neq \ell} u_ku_\ell \left(\bs e_k\bs x_0^* + \bs x_0\bs e_k^*\right) (x_0)_\ell\right\| +  \left\|\sum_{k} (u_k)^2 \left(\bs e_k\bs x_0^* + \bs x_0\bs e_k^*\right) (x_0)_k\right\| \right)\\
&\lesssim |\langle \bs x_0, \bs u \rangle |^2 \|\bs u\|^2 + k\|\bs u\|^2 |\langle \bs x_0, \bs u\rangle |^2 + |\langle \bs x_0, \bs u\rangle |^2 \left\|2\sum_{\ell} u_\ell (x_0)_\ell \left(\bs u\bs x_0^* + \bs x_0\bs u^*\right)\right\|\\
&+ |\langle \bs x_0, \bs u\rangle |^2 \left\|2\sum_{k} u_k^2 (x_0)_k \left(\bs e_k\bs x_0^* + \bs x_0\bs e_k^*\right)\right\|+ |\langle \bs x_0, \bs u\rangle |^2 \|\bs u\|^2 \|\bs x_0\bs x_0^*\|\\
&\lesssim |\langle \bs x_0, \bs u\rangle |^2 \|\bs u\|^2+ k\|\bs u\|^2 |\langle \bs x_0, \bs u\rangle |^2 + |\langle \bs x_0, \bs u\rangle |^2 |\langle \bs x_0, \bs u\rangle | \|\bs u\bs x_0^* + \bs x_0\bs u^*\|\\
&+ |\langle \bs x_0, \bs u\rangle |^2 \|\bs u\|^2 \|\bs x_0\bs x_0^*\| + |\langle \bs x_0, \bs u\rangle |^2 \|\bs u\|^2 \|\bs x_0\bs x_0^*\| \lesssim k
\end{align}
Finally, 
\begin{align}
\left\|\mathbb{E}\left\{|\langle \bs x_0, \bs u\rangle |^4 |\langle \bs x_0, \bs c_i\rangle |^2 \left(\bs c_i\bs x_0^* + \bs x_0\bs c_i^*\right)^2  \right\}\right\|\lesssim k
\end{align}
\begin{align}
\left\|\mathbb{E}\left\{|\langle \bs x_0, \bs c_i\rangle |^4 u_1^2 |\langle \bs x_0, \bs u\rangle |^2 \bs x_0\bs x_0^*\right\}\right\| &\lesssim |\langle \bs x_0, \bs u\rangle |^2 \|\bs x_0\bs x_0^*\| u_1^2 \\
&+ \left|\sum_{k, \ell} (x_0)_k^2 (x_0)_\ell^2 -\sum_{k} (x_0)_k^2 + \|\bs x_0\|_2^4 \right|
\end{align}
Using $\|\bs x_0\|_4 \leq \|\bs x_0\|_2$, we get 
\begin{align}
&\|\mathbb{E}\left\{|\langle \bs c_i, \bs u\rangle |^2 |\langle \bs x_0, \bs u\rangle |^2 |\langle \bs x_0, \bs c_i\rangle |^2 \bs x_0\bs x_0^* \right\}\|\\
&\lesssim |\langle \bs x_0, \bs u\rangle |^2 \left|\langle \bs u\bs u^* , \bs x_0\bs x_0^*\rangle - \sum_{k} u_k^2 (x_0)_k^2 + \|\bs u\|^2 \|\bs x_0\|^2 \right| \left\|\bs x_0\bs x_0^*\right\|
\end{align}
\begin{align}
&\left\|\mathbb{E}\left\{|\langle \bs x_0, \bs u\rangle |^4 |\langle \bs x_0, \bs c_i\rangle |^4 \bs x_0\bs x_0^* \right\}\right\|\\
& \lesssim |\langle \bs x_0, \bs u\rangle |^4 \|\bs x_0\bs x_0^*\| \left|\langle \bs x_0\bs x_0^*, \bs x_0\bs x_0^*\rangle - \sum_{k} (x_0)_k^4 + \|\bs x_0\|_2^4 \right|\lesssim 1
\end{align}
All in all, we thus have
\begin{align}
\left\|\mathbb{E}\sum_{i=1}^m Z_iZ_i^* - \mathbb{E}Z_i\mathbb{E}Z_i^*\right\|\lesssim k \label{boundVarianceLambda101}
\end{align}

The variables $Z_i = \mathcal{P}_{T\cap \Omega}(\bs c_i\bs e_1^* + \bs e_1\bs c_i^*)\langle \mathcal{P}_{T\cap \Omega}(\bs e_1\bs c_i^* + \bs c_i\bs e_1^*), \bs u\bs u^*\rangle $ have subexponential entries. Using 
\begin{align}
\left\|Z_i - \mathbb{E}Z_i\right\|&\lesssim \left\|\mathcal{P}_{T\cap \Omega} \left(\bs c_i\bs e_1^* + \bs e_1\bs c_i^*\right) \langle \mathcal{P}_{T\cap \Omega}(\bs e_1\bs c_i^* + \bs c_i\bs e_1^*), \bs u\bs u^*\rangle \right\|_{\psi_1}\\
&+ \left\|\mathbb{E} \left\{\mathcal{P}_{T\cap \Omega} \left(\bs c_i\bs e_1^* + \bs e_1\bs c_i^*\right) \langle \mathcal{P}_{T\cap \Omega} \left(\bs e_1\bs c_i^* + \bs c_i\bs e_1^*\right), \bs u\bs u^* \rangle  \right\} \right\|_{\psi_1}
\end{align} 
Now using 
\begin{align}
\mathcal{P}_{T\cap \Omega}(\bs c_i\bs e_1^* +\bs e_1\bs c_i^*)& = \langle \bs x_0, \bs c_i\rangle \left(\bs x_0\bs e_1^* + \bs e_1\bs x_0^*\right) + \left(\bs x_0\bs c_i^* + \bs c_i\bs x_0^*\right) - 2\langle \bs x_0, \bs c_i \rangle \bs x_0\bs x_0^*  \\
\langle \mathcal{P}_{T\cap \Omega}(\bs c_i\bs e_1^* +\bs e_1\bs c_i^*), \bs u\bs u^*\rangle & = 2\langle \bs x_0, \bs u\rangle u_1 \langle \bs x_0, \bs c_i\rangle + 2\langle \bs x_0, \bs u\rangle \langle \bs c_i, \bs u\rangle - 2\langle \bs x_0, \bs c_i\rangle |\langle \bs x_0, \bs u\rangle |^2     
\end{align}
we can consider the decomposition 
\begin{align}
&\left\|\mathbb{E}\left\{\mathcal{P}_{T\cap \Omega}\left(\bs c_i\bs e_1^* + \bs e_1\bs c_i^* \right) \langle \mathcal{P}_{T\cap \Omega}\left(\bs e_1\bs c_i^* + \bs c_i\bs e_1^*\right), \bs u\bs u^*\rangle \right\}\right\|\\
&\leq  \left\|\mathbb{E}\left\{2|\langle \bs x_0, \bs c_i\rangle|^2 \left(\bs x_0\bs e_1^* + \bs e_1\bs x_0^*\right) \langle \bs x_0, \bs u\rangle u_1   \right\}\right\|\\
& + \left\|\mathbb{E}\left\{2\langle \bs x_0, \bs c_i\rangle \langle \bs x_0, \bs u\rangle \langle \bs c_i, \bs u\rangle \left(\bs x_0\bs e_1^* + \bs e_1\bs x_0^*\right)  \right\} \right\|\\
&+ \left\|\mathbb{E}\left\{2|\langle \bs x_0, \bs c_i\rangle |^2 |\langle \bs x_0, \bs u\rangle |^2 \left(\bs x_0\bs e_1^* + \bs e_1\bs x_0^*\right) \right\}\right\|\\
&+ \left\|\mathbb{E}\left\{2\langle \bs x_0, \bs u\rangle u_1 \langle \bs x_0, \bs c_i\rangle \left(\bs x_0\bs c_i^* + \bs c_i\bs x_0^*\right)  \right\} \right\|\\
&+ \left\|\mathbb{E}\left\{2 \left(\bs x_0\bs c_i^* + \bs c_i\bs x_0^*\right) \langle \bs x_0, \bs u\rangle \langle \bs c_i, \bs u \rangle \right\}\right\|\\
&+ \left\|\mathbb{E}\left\{2\langle \bs x_0, \bs c_i\rangle |\langle \bs x_0, \bs u\rangle |^2 \left(\bs c_i\bs x_0^* + \bs x_0\bs c_i^*\right)  \right\}\right\|\\
&+ \left\|\mathbb{E}\left\{4|\langle \bs x_0, \bs c_i\rangle |^2 \bs x_0\bs x_0^* \langle \bs x_0, \bs u\rangle u_1  \right\}\right\|\\
&+ \left\|\mathbb{E}\left\{4 \langle \bs x_0, \bs u\rangle \langle \bs c_i, \bs u\rangle \langle \bs x_0, \bs c_i\rangle \bs x_0\bs x_0^*   \right\} \right\|\\
&+ \left\|\mathbb{E} \left\{4|\langle \bs x_0, \bs c_i\rangle |^2 |\langle \bs x_0, \bs u\rangle |^2 \bs x_0\bs x_0^*\right\} \right\|\\
&\leq \left\|\left(\bs x_0\bs e_1^* + \bs e_1\bs x_0^*\right) \langle \bs x_0, \bs u\rangle u_1 \left(\|\bs x_0\bs x_0^*\|_F^2 - \|\bs x_0\|^2 + \|\bs x_0\|^4\right) \right\|\label{finalNormExpectationendlambda1caseSDPProof01}\\
&+ \left\|2|\langle \bs x_0, \bs u\rangle |^2 \left(\bs x_0\bs e_1^* + \bs e_1\bs x_0^*\right) \right\|\\
&+ \left\|2 |\langle \bs x_0, \bs u\rangle |^2 \left(\bs x_0\bs e_1^* + \bs e_1\bs x_0^*\right) \left(\|\bs x_0\bs x_0^*\|_F^2 - \|\bs x_0\|^2 + \|\bs x_0\|^4 \right) \right\|\\
& + \left\|4\langle \bs x_0, \bs u\rangle u_1 \bs x_0\bs x_0^* \right\| + \left\|2\left(\bs x_0\bs u^* + \bs u\bs x_0^*\right) \langle \bs x_0, \bs u\rangle  \right\|\\
&+ \left\|4\bs x_0\bs x_0^* |\langle \bs x_0, \bs u\rangle |^2 \right\| + \left\|4\bs x_0\bs x_0^* \langle \bs x_0, \bs u\rangle u_1 \left(\|\bs x_0\bs x_0^*\|_F^2 - \|\bs x_0\|^2 + \|\bs x_0\|^4 \right)  \right\|\\
&+ \left\|4 |\langle \bs x_0, \bs u\rangle|^2 \bs x_0\bs x_0^*\right\| + \left\|4|\langle \bs x_0, \bs u\rangle |^2 \bs x_0\bs x_0^* \left(\|\bs x_0\bs x_0^*\|_F^2 -\|\bs x_0\|^2 + \|\bs x_0\|^4 \right) \right\|\lesssim 1\label{finalNormExpectationendlambda1caseSDPProof0end}
\end{align}
On the other hand, we have
\begin{align}
&\left\| |\langle \bs x_0, \bs c_i \rangle|^2 \left(\bs x_0\bs e_1^* = \bs e_1\bs x_0^*\right) \langle \bs x_0, \bs u\rangle u_1   \right\|_{\psi_1} \leq |\langle \bs x_0, \bs u\rangle u_1| \|\langle \bs x_0, \bs c_i\rangle \|_{\psi_2}^2 \lesssim 1\label{finalOrliczNormfirstlambda1caseSDPProof}\\
&\left\|\langle \bs x_0, \bs c_i\rangle \langle \bs x_0, \bs u\rangle \langle \bs c_i, \bs u\rangle \left(\bs x_0\bs e_1^* + \bs e_1\bs x_0^*\right)   \right\|_{\psi_1} \leq \|\langle \bs x_0, \bs c_i\rangle \|_{\psi_2} \|\langle \bs c_i, \bs u\rangle \|_{\psi_2} \lesssim 1\\
& \left\||\langle \bs x_0, \bs c_i\rangle |^2 |\langle \bs x_0, \bs u\rangle |^2 \left(\bs x_0\bs e_1^* + \bs e_1\bs x_0^*\right)\right\|_{\psi_1} \lesssim \left\|\langle \bs x_0, \bs c_i\rangle \right\|_{\psi_2}^2 \lesssim 1\\
& \left\| \left(\bs x_0\bs c_i^* + \bs c_i\bs x_0^*\right) \langle \bs x_0, \bs u\rangle \langle \bs c_i, \bs u\rangle   \right\|_{\psi_1} \leq \|\langle \bs c_i, \bs u\rangle \|_{\psi_2} \|\bs x_0\bs c_i^* + \bs c_i\bs x_0^*\|_{\psi_2}\lesssim \sqrt{k}\\
&\left\|2\left(\bs x_0\bs c_i^* + \bs c_i\bs x_0^*\right) \langle \bs x_0, \bs u\rangle \langle \bs c_i, \bs u\rangle  \right\|_{\psi_1} \lesssim \left\|\|\bs x_0\bs c_i^* + \bs c_i\bs x_0^*\| \right\|_{\psi_2} \left\|\langle \bs c_i, \bs u\rangle \right\|_{\psi_2} \lesssim \sqrt{k}\\
&\left\|2\langle \bs x_0, \bs c_i\rangle |\langle \bs x_0, \bs u\rangle |^2 \left(\bs c_i\bs x_0^* + \bs x_0\bs c_i^*\right) \right\| \lesssim \sqrt{k}
\end{align}
Finally,
\begin{align}
\left\| 4|\langle \bs x_0, \bs c_i\rangle |^2 \bs x_0\bs x_0^* \langle \bs x_0, \bs u\rangle u_1 \right\|_{\psi_1} &\lesssim \|\langle \bs x_0, \bs c_i\rangle \|_{\psi_1}^2 \lesssim 1\\
 \left\|4 \langle \bs x_0, \bs u\rangle \langle \bs c_i, \bs u\rangle \langle \bs x_0, \bs c_i\rangle \bs x_0\bs x_0^*   \right\|_{\psi_1}& \leq \|\langle \bs c_i, \bs u\rangle \|_{\psi_2} \|\langle \bs c_i, \bs x_0\rangle \|_{\psi_2}\lesssim 1\\
 \left\|4|\langle \bs x_0, \bs c_i\rangle |^2 |\langle \bs x_0, \bs u\rangle |^2 \bs x_0\bs x_0^*\right\|_{\psi_1}&\leq \|\langle \bs x_0, \bs c_i\rangle \|_{\psi_2}^2 \lesssim 1\label{finalOrliczNormendlambda1caseSDPProof}
\end{align}
Combining~\eqref{finalOrliczNormfirstlambda1caseSDPProof} to~\eqref{finalOrliczNormendlambda1caseSDPProof} with~\eqref{finalNormExpectationendlambda1caseSDPProof01}~\eqref{finalNormExpectationendlambda1caseSDPProof0end}, we can write  
\begin{align}
\|Z_i - \mathbb{E}Z_i\|_{\psi_1}\lesssim \sqrt{k}
\end{align}
Using this,  together with~\eqref{boundVarianceLambda101} as well as~\eqref{ExpectationProjectionTOmegaLambda1}, and applying the matrix version of Bernstein's inequality (see Proposition~\ref{BernsteinMatrix}) gives 
\begin{align}
\left\|\mathcal{H}(\bs X) - \mathcal{P}_{T\cap \Omega}(\bs X)\right\| &=\left\|\frac{1}{m}\sum_{i=1}^m \mathcal{P}_{T\cap \Omega}(\bs c_i\bs e_1^* + \bs e_1\bs c_i^*)\langle \mathcal{P}_{T\cap \Omega}\left( \bs e_1\bs c_i^* + \bs c_i\bs e_1^*\right), \bs X\rangle \right\|\\
& \lesssim \sqrt{\frac{k\delta\log^{\alpha}(n)}{m}}\|\bs X\|_F
\end{align}
with probability at least $1-n^{-\alpha}$.

\end{proof}

\subsubsection{\label{sectionrelationTTperpOmegaLinearMeasurements}Proof of Proposition~\ref{relationTTperpOmegaLinearMeasurements}}

\restaterelationTTperpOmegaLinearMeasurements*

\begin{proof}
Using the definition of $\mathcal{A}$,  considering a decomposition 
\begin{align}
\bs X = \left[\begin{array}{cc}
1 & \bs x^*\\
\bs x& \tilde{\bs X}
\end{array}\right]
\end{align}
and focusing on the random part first, we have 
\begin{align}
\begin{split}
&\mathbb{E}\left\{\mathcal{P}_{T\cap \Omega}\mathcal{A}^*\mathcal{A}(\bs X)\right\} = \mathbb{E}\left\{\frac{1}{2m}\sum_{i=1}^m \mathcal{P}_{T\cap \Omega} \left(\bs c_i\bs e_1^* + \bs e_1\bs c_i^*\right)\langle \bs c_i\bs e_1^* + \bs e_1\bs c_i^*, \bs X\rangle \right\}\\
 &= \mathbb{E} \left\{\frac{1}{2m}\sum_{i=1}^m\left[ \langle \bs x_0, \bs c_i\rangle \left(\bs x_0\bs e_1^* + \bs e_1\bs x_0^*\right) +  \left(\bs x_0\bs c_i^* + \bs c_i\bs x_0^*\right) - \bs x_0\bs x_0^* \langle \bs x_0, \bs c_i\rangle \right] \langle \bs c_i, \bs x\rangle  \right\}\end{split}\label{ExpectationAstarATfor relationHTHTper}
\end{align}
Focusing on the first term above,
\begin{align}
&\mathbb{E} \left\{\frac{1}{2m}\sum_{i=1}^m\left[ \langle \bs x_0, \bs c_i\rangle \left(\bs x_0\bs e_1^* + \bs e_1\bs x_0^*\right) +  \left(\bs x_0\bs c_i^* + \bs c_i\bs x_0^*\right) - \bs x_0\bs x_0^* \langle \bs x_0, \bs c_i\rangle \right] \langle \bs c_i, \bs x\rangle  \right\}\\
& = \left[\langle \bs x_0, \bs x\rangle\left(\bs x_0\bs e_1^* + \bs e_1\bs x_0^*\right) + \bs x_0\bs x^* + \bs x\bs x_0^* - \bs x_0\bs x_0^* \langle \bs x_0, \bs x\rangle  \right]
\end{align}
Substituting this in~\eqref{ExpectationAstarATfor relationHTHTper},  we can write 
\begin{align}
\left\|\mathbb{E}\mathcal{P}_{T\cap \Omega}\mathcal{A}^*\mathcal{A}\bs X\right\|\lesssim \|\bs X\|_F\label{normExpectationZiTtoTperpOmega}
\end{align}
Moreover, for the deterministic part, 
\begin{align}
\mathcal{D}(\bs X)&\equiv 2\mathcal{P}_{T\cap \Omega} \left(\bs e_1\bs e_1^*\right) \langle \bs e_1\bs e_1^*, \bs X\rangle - 2\mathcal{P}_{T\cap \Omega}(\bs e_1\bs e_1^*) \langle \bs I, \bs X\rangle  \\
& - 2\mathcal{P}_{T\cap \Omega} (\bs I) \langle \bs e_1\bs e_1^* , \bs X\rangle + 6\bs I_{T\cap \Omega} \langle \bs I, \bs X\rangle   
\end{align}
we also have
\begin{align}
\left\|\mathcal{D}(\bs X)\right\|\lesssim \|\bs X\|_F 
\end{align}
To control the deviation
\begin{align}
&\frac{1}{2m}\sum_{i=1}^m  \mathcal{P}_{T\cap \Omega}\left(\bs c_i\bs e_1^* + \bs e_1\bs c_i^*\right) \langle \bs c_i\bs e_1^* + \bs e_1\bs c_i^*, \bs X\rangle\\
&  - \mathbb{E}\frac{1}{2m}\sum_{i=1}^m \mathcal{P}_{T\cap \Omega} \left(\bs c_i\bs e_1^* + \bs e_1\bs c_i^*\right) \langle \bs c_i\bs e_1^* + \bs e_1\bs c_i^*\rangle 
\end{align}
As before, we define $Z_i$ as
\begin{align}
Z_i& =   \mathcal{P}_{T\cap \Omega}\left(\bs c_i\bs e_1^* + \bs e_1\bs c_i^*\right) \langle \bs c_i\bs e_1^* + \bs e_1\bs c_i^*, \bs X\rangle  - \mathbb{E} \mathcal{P}_{T\cap \Omega} \left(\bs c_i\bs e_1^* + \bs e_1\bs c_i^*\right) \langle \bs c_i\bs e_1^* + \bs e_1\bs c_i^*\rangle
\end{align}
The matrices $Z_i$ have subexponential entries. Using~\eqref{normExpectationZiTtoTperpOmega} together with standard results on subexponential tails, one can write 
\begin{align}
\|Z_i\|_{\psi_1}& \lesssim \left\|\mathcal{P}_{T\cap \Omega}\left(\bs e_1\bs c_i^* + \bs c_i\bs e_1^*\right) \right\|_{\psi_2} \|\langle \bs c_i\bs e_1^* + \bs e_1\bs c_i^*, \bs X\rangle \|_{\psi_2} + \|\bs X\|_F\\
&\lesssim \|\bs X\|_F \left(\left\|\bs c_i\bs e_1^* + \bs e_1\bs c_i^* \right\|_{\psi_2} + \left\|\langle \bs x_0, \bs c_i\rangle \left(\bs x_0\bs e_1^* + \bs e_1\bs x_0^*\right) \right\|_{\psi_2} + \left\|\bs x_0\bs x_0^* \langle \bs x_0, \bs c_i\rangle \right\|_{\psi_2}\right)\\
&\lesssim \sqrt{k}\|\bs X\|_F\label{boundOrliczConnectionTTperFirst}
\end{align}
On the other hand, if we let 
\begin{align}
\sigma^2 = \left\|\sum_{i=1}^m Z_i Z_i^*\right\|, 
\end{align}
Using $\mathbb{E}\left\{XX^*\right\} - \mathbb{E}X\mathbb{E}X^* \succeq 0$, and therefore focusing on the norm
\begin{align}
\left\|\sum_{i=1}^m \mathbb{E}\left\{\mathcal{P}_{T\cap \Omega}\left(\bs e_1\bs c_i^* + \bs c_i\bs e_1^*\right)^2 \left|\langle \bs e_1\bs c_i^* + \bs c_i\bs e_1^*, \bs X\rangle \right|^2 \right\}\right\|
\end{align}
Moreover, using $0\preceq \mathbb{E} \left\{ \left(\bs A_i+\bs B_i\right) \left(\bs A_i + \bs B_i\right)^* \right\}\preceq 2\mathbb{E}\left\{\bs A_i\bs A_i^*\right\} + 2\mathbb{E}\left\{\bs B_i\bs B_i^*\right\}$, we can write 
\begin{align}
\left\|\mathbb{E}\left\{\left(\bs A_i+\bs B_i\right)\left(\bs A_i + \bs B_i\right)^*\right\}\right\|&\lesssim \left\|\mathbb{E}\left\{\bs A_i\bs A_i^*\right\}\right\|+ \left\|\mathbb{E}\left\{\bs B_i\bs B_i^*\right\}\right\|
\end{align}
together with
\begin{align}
\mathcal{P}_{T\cap \Omega}\left(\bs c_i\bs e_1^* + \bs e_1\bs c_i^*\right)& = \langle \bs x_0, \bs c_i\rangle \left(\bs x_0\bs e_1^* + \bs e_1\bs x_0^*\right) + \left(\bs x_0\bs c_i^* + \bs c_i\bs x_0^*\right)  - \bs x_0\bs x_0^* \langle \bs x_0, \bs c_i\rangle 
\end{align}
we can write 
\begin{align}
\sigma^2& \lesssim \left\|\frac{1}{m}\sum_{i=1}^m \mathbb{E} |\langle \bs x_0, \bs c_i \rangle |^2 \left(\bs x_0\bs e_1^* + \bs e_1\bs x_0^*\right)^2 |\langle \bs c_i\bs e_1^* + \bs e_1\bs c_i^*, \bs X\rangle |^2 \right\|\label{term1TversTperpProof01}\\
&+ \left\|\frac{1}{m}\sum_{i=1}^m \mathbb{E} \left(\bs x_0\bs c_i^* + \bs c_i\bs x_0^*\right)^2 \left|\langle \bs c_i\bs e_1^* + \bs e_1\bs c_i^*\rangle \right|\right\|\label{term2TversTperpProof01}\\
&+ \left\|\frac{1}{m}\sum_{i=1}^m \mathbb{E}\left\{\bs x_0\bs x_0^* |\langle \bs x_0, \bs c_i\rangle |^2 \left|\langle \bs c_i\bs e_1^* + \bs e_1\bs c_i^*, \bs X\rangle \right|^2\right\}\right\|\label{term3TversTperpProof01}
\end{align}
For~\eqref{term1TversTperpProof01} and~\eqref{term3TversTperpProof01}, we use
\begin{align}
\eqref{term1TversTperpProof01} + \eqref{term3TversTperpProof01}\lesssim &\left|\langle \tilde{\bs x}_0\tilde{\bs x}_0^*, \bs x\bs x^*\rangle - \langle \text{\upshape diag}(\tilde{\bs x}_0\tilde{\bs x}_0^*), \text{\upshape diag}(\bs x\bs x^*)\rangle  \right| \left\|\bs x_0\bs e_1^* + \bs e_1\bs x_0^*\right\|^2\\
&+ \|\tilde{\bs x}_0\|^2 \|\bs x\|^2 \left\|\bs x_0\bs e_1^* + \bs e_1\bs x_0^*\right\|^2\\
\lesssim &\|\bs x\|^2\label{boundVarianceConnectionTTperFirst}
\end{align}
For~\eqref{term2TversTperpProof01}, we use 
\begin{align}
&\left\|\mathbb{E}\left\{\left(\bs x_0\bs x_0^*\|\bs c_i\|^2 + \bs c_i\bs c_i^* + \langle \bs x_0, \bs c_i\rangle \left(\bs x_0\bs c_i^* + \bs c_i\bs x_0^*\right) \right) \left|\langle \bs c_i\bs e_1^* + \bs e_1\bs c_i^*, \bs X\rangle \right|^2\right\}\right\|\\
\lesssim &k \|\bs x\|^2 + \left\|\bs x\bs x^* - \text{\upshape diag}(\bs x\bs x^*) + \|\bs x\|^2 \bs I\right\| \\
&+\left\|\mathbb{E}\left\{\left(\sum_{k, \ell} \left(\bs x_0\bs e_\ell^* + \bs e_\ell \bs x_0^*\right) c_i[k]c_i[\ell] x_0[k]\right) \sum_{a, b}c_i[a]c_i[b]x_ax_b\right\}\right\|\label{term2TversTperpProof02}
\end{align}
For the last term in~\eqref{term2TversTperpProof02}, we use 
\begin{align}
&\left\|\sum_{k\neq \ell} \left(\bs x_0\bs e_\ell^* + \bs e_\ell \bs x_0^* \right) x[\ell]x_0[k]x[k] + \sum_{k=\ell } \left(\bs x_0\bs e_k^* + \bs e_k\bs x_0^*\right) x_0[k] \|\bs x\|^2 \right\|\\
&\lesssim \left\|\left(\bs x_0\bs x^* + \bs x\bs x_0^*\right) \langle \bs x_0, \bs x\rangle - \sum_{k} \left(\bs x_0\bs e_k^* + \bs e_k\bs x_0^*\right) x^2[k] x_0[k]  \right\|\\
&\, + \left\|\bs x_0\bs x_0^* \|\bs x\|^2 \right\|\\
&\lesssim \|\bs x\|^2\label{boundVarianceConnectionTTperLast}
\end{align}
Grouping~\eqref{boundVarianceConnectionTTperFirst} to~\eqref{boundVarianceConnectionTTperLast}, with~\eqref{boundOrliczConnectionTTperFirst} and applying proposition~\ref{BernsteinMatrix}, we get 
\begin{align}
&\frac{1}{2m}\left\|\sum_{i=1}^m \mathcal{P}_{T\cap \Omega}\left(\bs c_i\bs e_1^* + \bs e_1\bs c_i^*\right) \langle \bs c_i\bs e_1^* + \bs e_1 \bs c_i^*\rangle -\mathbb{E}\sum_{i=1}^m  \mathcal{P}_{T\cap \Omega}\left(\bs c_i\bs e_1^* + \bs e_1\bs c_i^*\right) \langle \bs c_i\bs e_1^* + \bs e_1 \bs c_i^*\rangle \right\|\\
&\lesssim \delta\|\bs X\|_F
\end{align}
with probability at least $1-n^{-\alpha}$ as soon as $m\gtrsim k\log(n)^{\alpha}$. 
\end{proof}

\subsubsection{\label{sectionProofPropositionAstarAHOmegaComplement}Proof of Proposition~\ref{PropositionAstarAHOmegaComplement}}

\restatePropositionAstarAHOmegaComplement*

\begin{proof}
Using $S$ to denote the support of $\bs x_0$, and noting that $(1,1)\notin \Omega^c$,  first note that we have 
\begin{align}
&(2m)^{-1}\mathbb{E}\sum_{i=1}^m  \mathcal{P}_{T\cap \Omega}\left(\bs c_i\bs e_1^* + \bs e_1\bs c_i^*\right) \langle\bs e_1\bs c_i^* + \bs c_i\bs e_1^*, \bs X_{\Omega^c}\rangle \\
& =m^{-1}\mathbb{E} \sum_{i=1}^m \left(\langle \bs x_0, \bs c_i\rangle \left(\bs x_0\bs e_1^* + \bs e_1\bs x_0^*\right) + (\bs c_{i})_{\Omega}\bs x_0 + \bs x_0(\bs c_i)_{\Omega}^* - 2\langle \bs x_0, \bs c_i\rangle \bs x_0\bs x_0^*   \right) \langle \bs c_i, \bs x_{S^c}\rangle \\
& = 0
\end{align}
Moreover,  for the deterministic part, noting that $(1,1)\notin\Omega^c$, we get 
\begin{align}
\left\|\mathcal{D}(\bs X_{\Omega^c})\right\| &=  \left\|2\mathcal{P}_{T\cap \Omega}(\bs e_1\bs e_1^* + \bs I)\langle \bs e_1\bs e_1^*, \bs X_{\Omega^c} \rangle + \bs I_{T\cap \Omega}\langle 6\bs I -2\bs e_1\bs e_1^*, \bs X_{\Omega^c} \rangle\right\|\\
&\leq 8\text{\upshape Tr}(\bs X_{\Omega^c})
\end{align}

To control the deviation
\begin{align}
\left\|\frac{1}{2m}\sum_{i=1}^m \mathcal{P}_{T\cap \Omega}\left(\bs c_i\bs e_1^* + \bs e_1\bs c_i^*\right) \langle \bs c_i\bs e_1^* + \bs e_1\bs c_i^*\rangle - \frac{1}{2m}\sum_{i=1}^m \mathcal{P}_{T\cap \Omega}\left(\bs c_i\bs e_1^* + \bs e_1\bs c_i^*\right) \langle \bs c_i\bs e_1^* + \bs e_1\bs c_i^*, \bs X_{\Omega^c}\rangle   \right\|
\end{align}
we again turn to Proposition~\ref{BernsteinMatrix}.  Letting $Z_i$ to denote the variables
\begin{align}
Z_i& = \frac{1}{2} \mathcal{P}_{T\cap \Omega}\left(\bs c_i\bs e_1^* + \bs e_1\bs c_i^*\right) \langle \bs c_i\bs e_1^* + \bs e_1\bs c_i^*, \bs X_{\Omega^c}\rangle 
\end{align}
as well as $\sigma^2 = \left\|m^{-1}\sum_{i=1}^m Z_iZ_i^*\right\|$, we can write 
\begin{align}
\|Z_i\|_{\psi_1}&\leq \left\|\mathcal{P}_{T\cap \Omega}\left(\bs c_i\bs e_1^* + \bs e_1\bs c_i^*\right)\right\|_{\psi_2} \left\|\langle \bs c_i\bs e_1^* + \bs e_1\bs c_i^*, \bs X_{\Omega^c}\rangle \right\|_{\psi_2}\\
&+ \left\|\left\|\bs x_0(\bs x_{S^c})^* + (\bs x_{S^c})\bs x_0^*\right\|\right\|_{\psi_1}\\
&\lesssim \|\bs x_{S^c}\| + \left\|\bs x_{S^c}\right\| \left(\|\langle \bs x_0, \bs c_i\rangle \left(\bs x_0\bs e_1^* + \bs e_1\bs x_0^*\right) \|_{\psi_2} + \left\|\left(\bs x_0\bs c_i^* + \bs c_i\bs x_0^*\right)\right\|_{\psi_2}\right)\\
&  + \left\|\bs x_{S^c}\right\| \left(\left\|\bs x_0\bs x_0^* \langle \bs x_0, \bs c_i\rangle \right\|_{\psi_2} \right)\\
&\lesssim \|\bs x_{S^c}\| + \|\bs x_{S^c}\| \sqrt{k}\label{boundOrliczOmegaC01}
\end{align}
On the other hand, we have 
\begin{align}
\sigma^2 &= \left\|\frac{1}{m}\sum_{i=1}^m \mathbb{E}Z_iZ_i^*\right\|\\
&\leq \left\|\frac{1}{m}\sum_{i=1}^m \mathbb{E} \bs A_i\bs A_i^* \left|\langle \bs c_i\bs e_1^* + \bs e_1\bs c_i^*, \bs X_{\Omega^c}\rangle \right|^2\right\|
\end{align}
where the $\bs A_i$ are defined as 
\begin{align}
\bs A_i& = \left(\bs x_0\bs c_i^* + \bs c_i\bs x_0^*\right) + \langle \bs x_0, \bs c_i\rangle \left(\bs x_0\bs e_1^* + \bs e_1\bs x_0^*\right) - 2\bs x_0\bs x_0^* \langle \bs x_0, \bs c_i\rangle  
\end{align}
Noting that $\sum_{i} \left(\bs A_i' +\bs A_i'' \right)\left(\bs A_i' +\bs A_i'' \right)^* \precsim \sum_{i} \bs A_i'(\bs A_i')^* + \bs A_i''(\bs A_i'')^* $, we get
\begin{align}
\sigma^2 &\lesssim \left\|\frac{1}{m}\sum_{i=1}^m \mathbb{E}\left\{\left(\bs x_0\bs c_i^* + \bs c_i\bs x_0^*\right)^2 \left|\langle \bs c_i\bs e_1^* + \bs e_1\bs c_i^*, \bs X_{\Omega^c}\rangle \right| \right\}\right\|\\
&+ \left\|\frac{1}{m}\sum_{i=1}^m \mathbb{E}\left\{|\langle \bs x_0,\bs c_i\rangle |^2 \left(\bs x_0\bs e_1^* + \bs e_1\bs x_0^*\right)^2 |\langle \bs c_i\bs e_1^* + \bs e_1\bs c_i^*, \bs X_{\Omega^c}\rangle |\right\}\right\|\\
&+ \left\|\frac{1}{m}\sum_{i=1}^m \mathbb{E}\left\{\bs x_0\bs x_0^* |\langle \bs x_0, \bs c_i\rangle |^2 \left|\langle \bs c_i\bs e_1^* + \bs e_1\bs c_i^*, \bs X_{\Omega^c}\rangle \right| \right\} \right\|\\
&\lesssim \left\|\frac{1}{m} \sum_{i=1}^m \mathbb{E} \left\{\|\bs c_i\|^2 \bs x_0\bs x_0^* + \bs c_i\bs c_i^* + \langle \bs c_i, \bs x_0\rangle \left(\bs c_i\bs x_0^* + \bs x_0\bs c_i^*\right) |\langle \bs x_{S^c}, \bs c_i\rangle |^2 \right\} \right\|\\
&+ \left\|\frac{1}{m}\sum_{i=1}^m \mathbb{E} \left\{|\langle \bs x_0, \bs c_i\rangle |^2 \left(\bs e_1\bs e_1^* + \bs e_1\bs x_0^* + \bs x_0\bs e_1^* +\bs x_0\bs x_0^*\right) |\langle \bs c_i, \bs x_{S^c}\rangle|^2 \right\} \right\|\\
&+ \left\|\frac{1}{m}\sum_{i=1}^m \mathbb{E}\left\{\bs x_0\bs x_0^* |\langle \bs x_0, \bs c_i\rangle |^2 |\langle \bs c_i, \bs x_{S^c}\rangle |^2 \right\}\right\|\\
&\lesssim \left\|k \|\bs x_{S^c}\|^2 \bs x_0\bs x_0^* + \left(\bs x_{S^c}(\bs x_{S^c})^* - \text{\upshape diag}(\bs x_{S^c}(\bs x_{S^c})^*) + \|\bs x_{S^c}\|^2\bs I\right)\right\|\\
& + \|\bs x_{S^c}\|^2 \|\bs x_0\|^2 + \|\bs x_{S^c}\|^2 \|\bs x_0\|^2\label{stepExpansionCix0}\\
&+ \left\|\left(\bs e_1\bs e_1^* + \bs e_1\bs x_0^* + \bs x_0\bs e_1^* +\bs x_0\bs x_0^*\right)\right\| \left|\langle \bs x_0\bs x_0^*, \bs x_{S^c}(\bs x_{S^c})^*\rangle\right\|\\
& -\left\|\left(\bs e_1\bs e_1^* + \bs e_1\bs x_0^* + \bs x_0\bs e_1^* +\bs x_0\bs x_0^*\right)\right\| \left|\langle \text{\upshape diag}(\bs x_0\bs x_0^*), \text{\upshape diag}(\bs x_{S^c}(\bs x_{S^c})^*)\rangle + \|\bs x_0\|^2 \|\bs x_{S^c}\|^2  \right|\\
&+ \left\|\bs x_0\bs x_0^*\right\| \left|\langle \bs x_0\bs x_0^* , \bs x_{S^c}(\bs x_{S^c})^*\rangle - \langle \text{\upshape diag}(\bs x_0\bs x_0^*), \text{\upshape diag}(\bs x_{S^c}(\bs x_{S^c})^*)\rangle + \|\bs x_0\|^2 \|\bs x_{S^c}\|^2  \right|\\
&\lesssim \|\bs x_{S^c}\|^2 k\label{boundVarianceOmegacrelationHTOmega}
\end{align}
In~\eqref{stepExpansionCix0} we use
\begin{align}
&\left\|\mathbb{E}\left\{\sum_{a, b} c_i[a]x_0[a]c_i[b]\left(\bs e_b\bs x_0^* + \bs x_0\bs e_b^*\right)\sum_{m,n} c_i[m]x_{S^c}[m]c_i[n]x_{S^c}[n] \right\}\right\|\\
& = \left\| \sum_{a\neq b} x_0[a]\left(\bs x_0\bs e_b^* + \bs e_b\bs x_0^*\right) x_{S^c}[a]x_{S^c}[b] \right\| + \left\|2\bs x_0\bs x_0^* \|\bs x_{S^c}\|^2 \right\| + \|\bs x_{S^c}\|^2 \|\bs x_0\|^2\\
&\lesssim \left\|\langle \bs x_0, \bs x_{S^c} \left(\bs x_0(\bs x_{S^c})^* + \bs x_{S^c}\bs x_0^*\right)\rangle \right\| + \left\|\bs x_{S^c}\right\|^2 \left\|\bs x_0\right\|^2
\end{align}
Combining~\eqref{boundVarianceOmegacrelationHTOmega} and~\eqref{boundOrliczOmegaC01} and applying Proposition~\ref{BernsteinMatrix}, we get
\begin{align}
&\left\|m^{-1}\sum_{i=1}^m \mathcal{P}_{T\cap \Omega}\left(\bs c_i\bs e_1^* + \bs e_1\bs c_i^*\right) \langle \bs c_i\bs e_1^* + \bs e_1\bs c_i^*, \bs X_{\Omega^c}\rangle\right.\\
&\left. - \mathbb{E}m^{-1} \sum_{i=1}^m \mathcal{P}_{T\cap \Omega}\left(\bs c_i\bs e_1^* + \bs e_1\bs c_i^*\right)\langle \bs c_i\bs e_1^* + \bs e_1\bs c_i^*, \bs X_{\Omega^c}\rangle \right\|\lesssim \sqrt{\frac{k\|\bs x_{S^c}\|^2}{m}} \vee \frac{k\|\bs x_{S^c}\|}{m}
\end{align}
with probability at least $1-e^{-t}$. Taking $t=\log^{\alpha}(n)$ gives
\begin{align}
&\left\|m^{-1}\sum_{i=1}^m \mathcal{P}_{T\cap \Omega}\left(\bs c_i\bs e_1^* + \bs e_1\bs c_i^*\right) \langle \bs c_i\bs e_1^* + \bs e_1\bs c_i^*, \bs X_{\Omega^c}\rangle\right.\\
&\left. - \mathbb{E}m^{-1} \sum_{i=1}^m \mathcal{P}_{T\cap \Omega}\left(\bs c_i\bs e_1^* + \bs e_1\bs c_i^*\right)\langle \bs c_i\bs e_1^* + \bs e_1\bs c_i^*, \bs X_{\Omega^c}\rangle \right\|\lesssim \delta \|\bs x_{S^c}\|
\end{align}
with probability at least $1-n^{-\alpha}$ as soon as $m\gtrsim k\delta^{-1}$ (up to log factors).

\end{proof}

\subsection{\label{numericalExperimentsQuadratic}Numerical Experiments}

To back the theory, we provide a number of numerical experiments. We start by studying the recovery of a vector $\bs x_0\in [1,\tilde{\bs x}_0]$ with support $S$, of size $|S|=k$, sampled uniformly at random from $[1,2,\ldots, n]$ and $\left.\tilde{\bs x}_0\right|_{S}\in \text{Uni}\left( \mathbb{S}^{|S|-1}\right)$ so that $\|\tilde{\bs x}_0\|_2=1$.  We take $n = 30$, $k = 5$ and $m$ ranging from $2$ to $30$.  Recovery from a system of purely quadratic equations would occur at $m \geq k^2= 25$ in this case. From Figure~\ref{figureRecoveryLinearMatrixL1} we see that recovery of $\bs x_0$ in the linear setting with the vector and matrix $\ell_1$ norms occur, as expected well before this limit.  For each of the two formulations $\mathsf{SDP}_1(\tilde{\bs A}, \lambda)$ and $\mathsf{SDP}_2(\tilde{\bs A}, \lambda)$, we conduct a number of $10$ experiments and display the average relative errors in Figure~\ref{figureRecoveryLinearMatrixL1} (top). 

To further corroborate the assumptions of Propositions~\ref{propositionOnlyLinearVectorL1} and~\ref{propositionSDP2lambda1}, we study the recovery of a vector $\bs x_0$ with $S=\text{supp}(\bs x_0)$, $|S|=k$ sampled uniformly at random from $[1,2,\ldots, n]$ and defined on $S$ as $x_0[i] = w/\sqrt{w^2 + (1-w)^2(k-1)} $ for $i=S[1]$ and $(1-w)/\sqrt{w^2 + (1-w)^2(k-1)}$ otherwise. The experiments are again repeated 10 times and for various values of $w$. The average relative errors are shown in Fig~\ref{figureRecoveryLinearMatrixL1} (bottom left). Those results in particular suggest that the incoherence and random sign conditions in Proposition~\ref{propositionSDP2lambda1} arise as theoretical artefacts.  

Finally, we study the recovery of the matrix $\bs X_0 = \bs x_0\bs x_0^*$, $\bs x_0 = [1, \tilde{\bs x}_0]$ for a $k$-sparse $\tilde{\bs x}_0\in \mathbb{S}^{n-1}$ without the trace constraint. As in the previous settings we repeat our experiments 10 times and display the average relative error in Fig~\ref{figureRecoveryLinearMatrixL1} (bottom right). Those experiments clearly indicate the need for this constraint.
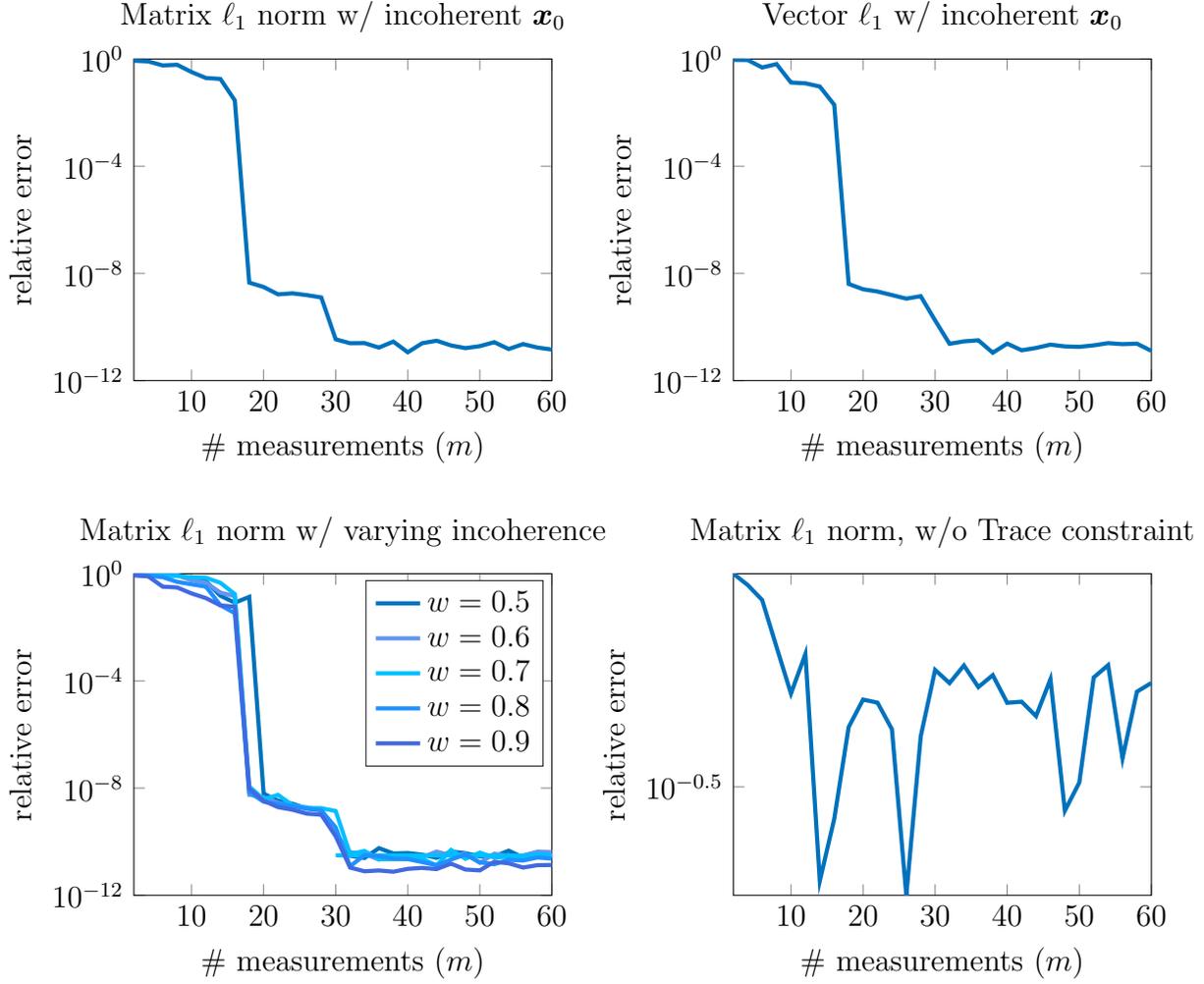
\begin{figure}[h!]
\hspace{-1cm}
\begin{minipage}{0.3\linewidth}
%
%
\definecolor{mycolor1}{rgb}{0.00000,0.44700,0.74100}%
\begin{tikzpicture}

\begin{axis}[%
width=2.228in,
height=1.7154in,
at={(1.011in,0.642in)},
scale only axis,
xmin=2,
xmax=60,
ymode=log,
ymin=1e-12,
ymax=1,
yminorticks=true,
axis background/.style={fill=white},
xlabel={\# measurements ($m$)},
    ylabel={relative error},
title = {Matrix $\ell_1$ norm w/ incoherent $\bs x_0$}
]
\addplot [color=mycolor1, forget plot, line width=1.6pt]
  table[row sep=crcr]{%
2	0.861801639374832\\
4	0.813893599552643\\
6	0.575902983656732\\
8	0.609789420597854\\
10	0.326863836628767\\
12	0.195682013656673\\
14	0.180840638261492\\
16	0.0285097997230055\\
18	4.52598898715635e-09\\
20	3.11004017248476e-09\\
22	1.65245315975616e-09\\
24	1.80302040033212e-09\\
26	1.54516368637627e-09\\
28	1.26558490693445e-09\\
30	3.46503391373232e-11\\
32	2.47618167073508e-11\\
34	2.50457314591715e-11\\
36	1.70340147179476e-11\\
38	2.85074838164305e-11\\
40	1.1387600475716e-11\\
42	2.4831231404017e-11\\
44	3.07316109062812e-11\\
46	2.04827451067447e-11\\
48	1.62557716358814e-11\\
50	1.92969579069761e-11\\
52	2.71243966214095e-11\\
54	1.50545240520523e-11\\
56	2.30207829205235e-11\\
58	1.71016251206012e-11\\
60	1.43877794937807e-11\\
};
\end{axis}

\begin{axis}[%
width=2.228in,
height=1.7154in,
at={(0in,0in)},
scale only axis,
xmin=0,
xmax=1,
ymin=0,
ymax=1,
axis line style={draw=none},
ticks=none,
axis x line*=bottom,
axis y line*=left
]
\end{axis}
\end{tikzpicture}%
\end{minipage}\hspace{3cm}
\begin{minipage}{0.3\linewidth}
%
%
\definecolor{mycolor1}{rgb}{0.00000,0.44700,0.74100}%
\begin{tikzpicture}

\begin{axis}[%
width=2.228in,
height=1.7154in,
at={(1.011in,0.642in)},
scale only axis,
xmin=2,
xmax=60,
ymode=log,
ymin=1e-12,
ymax=1,
yminorticks=true,
axis background/.style={fill=white},
xlabel={\# measurements ($m$)},
    ylabel={relative error},
title = {Vector $\ell_1$ w/ incoherent $\bs x_0$}
]
\addplot [color=mycolor1, forget plot, line width=1.6pt]
  table[row sep=crcr]{%
2	0.926629875514028\\
4	0.916954096912032\\
6	0.482793037666978\\
8	0.646397080096777\\
10	0.132634861948755\\
12	0.123840678270426\\
14	0.0940136384878549\\
16	0.0195819073177046\\
18	4.04551553307582e-09\\
20	2.55994829959576e-09\\
22	2.10972807061004e-09\\
24	1.55428554326959e-09\\
26	1.1432427486223e-09\\
28	1.41090081751153e-09\\
30	1.756855134483e-10\\
32	2.36261634215116e-11\\
34	2.90650816084253e-11\\
36	3.19082717558013e-11\\
38	1.0948835711531e-11\\
40	2.40148597877748e-11\\
42	1.3506000212959e-11\\
44	1.65208022676803e-11\\
46	2.19743755005522e-11\\
48	1.88269036732003e-11\\
50	1.81612642421509e-11\\
52	2.06346314573567e-11\\
54	2.49876303645274e-11\\
56	2.28877027160594e-11\\
58	2.38845541057002e-11\\
60	1.2877863968809e-11\\
};
\end{axis}

\begin{axis}[%
width=2.228in,
height=1.7154in,
at={(0in,0in)},
scale only axis,
xmin=0,
xmax=1,
ymin=0,
ymax=1,
axis line style={draw=none},
ticks=none,
axis x line*=bottom,
axis y line*=left
]
\end{axis}
\end{tikzpicture}%
\end{minipage}\\

\hspace{-1cm}
\begin{minipage}{0.3\linewidth}
%
%
\definecolor{mycolor1}{rgb}{0.00000,0.44700,0.74100}%
\definecolor{mycolor2}{RGB}{100,149,237}%
\definecolor{mycolor3}{RGB}{0,191,255}%
\definecolor{mycolor4}{RGB}{30,144,255}%
\definecolor{mycolor5}{RGB}{65,105,225}%
\begin{tikzpicture}

\begin{axis}[%
width=2.228in,
height=1.7154in,
at={(1.011in,0.642in)},
scale only axis,
xmin=2,
xmax=60,
ymode=log,
ymin=1e-12,
ymax=1,
yminorticks=true,
axis background/.style={fill=white},
legend style={legend cell align=left, align=left, draw=white!15!black},
xlabel={\# measurements ($m$)},
    ylabel={relative error},
title = {Matrix $\ell_1$ norm w/ varying incoherence}
]
\addplot [color=mycolor1,line width=1.6pt]
  table[row sep=crcr]{%
2	0.925197116297647\\
4	0.917725451108134\\
6	0.889225554313859\\
8	0.882532581603353\\
10	0.64226835121396\\
12	0.517358683401099\\
14	0.15082798051969\\
16	0.0826625300916496\\
18	0.136083814985821\\
20	6.19721526464422e-09\\
22	3.38719078576636e-09\\
24	2.60124943866368e-09\\
26	1.72519096429066e-09\\
28	1.55180160222642e-09\\
30	2.13801172615475e-10\\
32	2.97498131913101e-11\\
34	2.46299342686762e-11\\
36	5.73180710774561e-11\\
38	3.63417405754863e-11\\
40	3.76912647341846e-11\\
42	3.26802199270494e-11\\
44	2.50309236560166e-11\\
46	4.29146254778031e-11\\
48	3.71195215341752e-11\\
50	2.81844855980281e-11\\
52	2.82422667185788e-11\\
54	4.6100705515047e-11\\
56	2.57199424852138e-11\\
58	4.06049560082506e-11\\
60	3.00034015318309e-11\\
};
\addlegendentry{$w = 0.5$}

\addplot [color=mycolor2, line width=1.6pt]
  table[row sep=crcr]{%
2	0.870545314784002\\
4	0.968357661766201\\
6	0.844557500605038\\
8	0.864484402051707\\
10	0.572870038509288\\
12	0.461178019843884\\
14	0.196909017128781\\
16	0.144602480899879\\
18	5.63687313481581e-09\\
20	4.56619008456269e-09\\
22	2.81005863116249e-09\\
24	2.14305459759871e-09\\
26	1.91341817072067e-09\\
28	1.46916062494851e-09\\
30	1.95613810594799e-10\\
32	4.04404091616567e-11\\
34	3.51215566703931e-11\\
36	2.14277867139821e-11\\
38	3.01972628338014e-11\\
40	2.97752045624595e-11\\
42	2.62624889104089e-11\\
44	4.26202605357148e-11\\
46	3.06182717985208e-11\\
48	3.52919194263932e-11\\
50	1.62134761755026e-11\\
52	3.55271654027725e-11\\
54	2.91791249253362e-11\\
56	3.06194734498135e-11\\
58	4.22016186600926e-11\\
60	4.0880844874068e-11\\
};
\addlegendentry{$w = 0.6$}

\addplot [color=mycolor3, line width=1.6pt]
  table[row sep=crcr]{%
2	0.928438290959634\\
4	0.857179442009879\\
8	0.836626894985408\\
10	0.739754331224973\\
12	0.704553725523316\\
14	0.45820393997363\\
16	0.173242811443695\\
18	1.15619968978578e-08\\
20	3.81870961418726e-09\\
22	5.58083126736101e-09\\
24	2.29908909181014e-09\\
26	1.80155640815409e-09\\
28	1.78131587293996e-09\\
30	1.39319768680301e-09\\
32	3.47528768486355e-11\\
34	4.49862457499352e-11\\
36	2.19939434619979e-11\\
38	2.25521919115418e-11\\
40	2.30039746450152e-11\\
42	2.29255146597808e-11\\
44	1.29291206743624e-11\\
46	4.91557829646846e-11\\
48	2.30704903798733e-11\\
50	3.8954302343264e-11\\
52	2.55662920137981e-11\\
54	2.74579559081506e-11\\
56	2.83996641488324e-11\\
58	3.43919357262031e-11\\
60	3.03239627762451e-11\\
30	3.09862504974821e-11\\
};
\addlegendentry{$w = 0.7$}

\addplot [color=mycolor4, line width=1.6pt]
  table[row sep=crcr]{%
2	0.881368688129894\\
4	0.833967014038345\\
6	0.7409149050182\\
8	0.495677280774816\\
10	0.403499173355253\\
12	0.32879720819566\\
14	0.0667536923206447\\
16	0.0335477309475007\\
18	7.06689405319547e-09\\
20	3.09453439281059e-09\\
22	2.57537766496833e-09\\
24	2.10852723898248e-09\\
26	1.67587494058229e-09\\
28	1.46483434728885e-09\\
30	3.40653026713729e-10\\
32	1.20830516420922e-11\\
34	3.15450150148444e-11\\
36	2.85704303211755e-11\\
38	2.29659791819127e-11\\
40	2.23980627442815e-11\\
42	1.7373375302974e-11\\
44	1.31819981047591e-11\\
46	2.15699183673764e-11\\
48	3.33197984369525e-11\\
50	1.70227483077658e-11\\
52	1.44195678581289e-11\\
54	2.10660435476984e-11\\
56	1.98138189323582e-11\\
58	2.5750342045887e-11\\
60	2.32986154545152e-11\\
};
\addlegendentry{$w =0.8$}

\addplot [color=mycolor5, line width=1.6pt]
  table[row sep=crcr]{%
2	0.868055088228917\\
4	0.791408039953162\\
6	0.328703828392677\\
8	0.30689225851675\\
10	0.183739254496404\\
12	0.122381432627517\\
14	0.0655162774308057\\
16	0.0584434831156798\\
18	9.62099097876106e-09\\
20	3.34408835126706e-09\\
22	1.94563686044322e-09\\
24	1.55421823742771e-09\\
26	1.11049674703123e-09\\
28	1.0145128327093e-09\\
30	1.59005472516268e-10\\
32	1.0828813169972e-11\\
34	7.85835793891934e-12\\
36	8.45295089708167e-12\\
38	7.57352106220992e-12\\
40	9.67804763123854e-12\\
42	1.05317920976585e-11\\
44	9.48837382588084e-12\\
46	1.51998187134349e-11\\
48	9.14429241665776e-12\\
50	8.49379706859166e-12\\
52	1.78214610268793e-11\\
54	1.55756550450229e-11\\
56	1.08606754947002e-11\\
58	1.32917559120313e-11\\
60	1.34992402005251e-11\\
};
\addlegendentry{$w = 0.9$}

\end{axis}

\begin{axis}[%
width=2.228in,
height=1.7154in,
at={(0in,0in)},
scale only axis,
xmin=0,
xmax=1,
ymin=0,
ymax=1,
axis line style={draw=none},
ticks=none,
axis x line*=bottom,
axis y line*=left
]
\end{axis}
\end{tikzpicture}%
\end{minipage}\hspace{3cm}
\begin{minipage}{0.3\linewidth}
%
%
\definecolor{mycolor1}{rgb}{0.00000,0.44700,0.74100}%
\begin{tikzpicture}

\begin{axis}[%
width=2.228in,
height=1.7154in,
at={(1.011in,0.642in)},
scale only axis,
xmin=2,
xmax=60,
ymode=log,
ymin=0.183163557882738,
ymax=0.926305171033744,
yminorticks=true,
axis background/.style={fill=white},
xlabel={\# measurements ($m$)},
    ylabel={relative error},
title = {Matrix $\ell_1$ norm,  w/o Trace constraint}
]
\addplot [color=mycolor1, forget plot, line width=1.6pt]
  table[row sep=crcr]{%
2	0.926305171033744\\
4	0.874999948998156\\
6	0.811512578749527\\
8	0.640052237283235\\
10	0.506297606997037\\
12	0.617729556956604\\
14	0.197010044007035\\
16	0.269118146053049\\
18	0.427615687829486\\
20	0.491130193683393\\
22	0.483639190265549\\
24	0.423059159059947\\
26	0.183163557882738\\
28	0.407940166325098\\
30	0.570350560976841\\
32	0.533532155208053\\
34	0.58333519021193\\
36	0.523000742977577\\
38	0.555544510454287\\
40	0.483189472509491\\
42	0.48590124677407\\
44	0.451839593937599\\
46	0.54430471495143\\
48	0.28054827622542\\
50	0.323330722689312\\
52	0.549005628282317\\
54	0.58419322451903\\
56	0.365375490765714\\
58	0.511057075141469\\
60	0.533941614259256\\
};
\end{axis}

\begin{axis}[%
width=2.228in,
height=1.7154in,
at={(0in,0in)},
scale only axis,
xmin=0,
xmax=1,
ymin=0,
ymax=1,
axis line style={draw=none},
ticks=none,
axis x line*=bottom,
axis y line*=left
]
\end{axis}
\end{tikzpicture}%
\end{minipage}
\caption{\label{figureRecoveryLinearMatrixL1}Top: Phase transition for the recovery of $\bs X_0$ by the semidefinite programs $\mathsf{SDP}_1(\tilde{\bs A}, 1)$ and $\mathsf{SDP}_2(\tilde{\bs A}, 1)$. The simulations were run for $n=30$, $k=4$ as well as varying value of $m$. Bottom: (Left) Evolution of the recovery for various levels of incoherence (Right) Evolution of the recovery without the Trace constraint.}
\end{figure}

%

\section{\label{sectionFinalResult}General $\mathsf{SDP}_2(\tilde{\bs A}, \lambda)$ setting}

In this section, we consider the semidefinite program $\mathsf{SDP}(\tilde{\bs A}, \lambda)$ which interpolates between the linear and quadratic measurements on $\bs x\in \mathbb{S}^{n-1}$. The linear map $\mathcal{A}_{\lambda}$ is now defined mathematically as
\begin{align}
\begin{array}{lll}
\mathcal{S}^{n\times n} & \rightarrow & \mathbb{R}^m\\
\bs X &\mapsto&  \left\{\langle\lambda\left(\bs c_i\bs e_1^* + \bs e_1\bs c_i^*\right)+ (1-\lambda)\langle\bs A_i, \bs X \rangle \right\}_{1\leq i\leq m}
\end{array}\label{definitionMathcalAInterpolationLinearQuadratic01}
\end{align}

As explained in section~\ref{quadraticEquations}, it will not be possible to ensure recovery of the rank one matrix $\bs X_0 = \bs x_0\bs x_0^*$ in the case $\lambda=0$ as the quadratic part of $\mathcal{A}$ only applies to the central part of $\bs X$ (i.e. the second order pseudo moments). Clearly if we let $\tilde{\bs X}_0 = \tilde{\bs x}_0\tilde{\bs x}_0^*$, then the matrix 
\begin{align}
\left[\begin{array}{cc}
1 & \bs 0^*\\
\bs 0& \tilde{\bs x}_0\tilde{\bs x}_0^* 
\end{array}\right]
\end{align} 
will always have smaller $\ell_1$ norm than $\bs X_0 = \bs x_0\bs x_0^*$ and will perfectly satisfy the constraints.  It seems therefore reasonable to expect that beyond a certain threshold, it will not be possible to recover the whole matrix anymore and we should instead focus solely on the central part (from which it will be possible to extract $\bs x_0$).  This phenomenon is illlustrated in section~\ref{sectionNumericalExperimentsInterpolation}.

\subsection{General outline}

Keeping the inverse from section~\ref{sectionLinearMeasurementsMatrixL1} and simply viewing the additional contribution as a perturbation, for any $\bs X\in T\cap \Omega$, we have
\begin{align}
&\bs X +  \lambda^{-1} \left(\mathcal{P}_{T\cap \Omega}\mathcal{H}\mathcal{P}_{T\cap \Omega}\right)^{-1}\sum_{i=1}^m \mathcal{P}_{T\cap \Omega} \left(\bs e_1\bs c_i^* + \bs c_i\bs e_1^*\right) \langle  (1-\lambda)\bs A_i, \bs X\rangle\\
&= \left(\mathcal{P}_{T\cap \Omega}\mathcal{H}\mathcal{P}_{T\cap \Omega}\right)^{-1} (\lambda m)^{-1}\sum_{i=1}^m \mathcal{P}_{T\cap \Omega} \left(\bs e_1\bs c_i^* + \bs c_i\bs e_1^*\right) \langle \lambda \left(\bs e_1\bs c_i^* + \bs c_i\bs e_1^*\right) + (1-\lambda)\bs A_i, \bs X\rangle \\
&+ \left(\mathcal{P}_{T\cap \Omega}\mathcal{H}\mathcal{P}_{T\cap \Omega}\right)^{-1}\left(2\mathcal{P}_{T\cap \Omega} \left(\bs e_1\bs e_1^*\right) \langle \bs e_1\bs e_1^*, \bs X\rangle - 2\mathcal{P}_{T\cap \Omega}(\bs e_1\bs e_1^*) \langle \bs I, \bs X\rangle   \right)\\
& + \left(\mathcal{P}_{T\cap \Omega}\mathcal{H}\mathcal{P}_{T\cap \Omega}\right)^{-1}\left(- 2\mathcal{P}_{T\cap \Omega} (\bs I) \langle \bs e_1\bs e_1^* , \bs X\rangle+6\bs I_{T\cap \Omega} \langle \bs I, \bs X\rangle   \right)
\end{align}
From this, combing Proposition~\ref{propositionInjectivityOnTcapOmega} as well as the results of section~\ref{sectionLinearMeasurementsMatrixL1} with Lemma~\ref{lemmaCrossTermsAici01Quadratic} below,  we can further write 
\begin{align}
\begin{split}
\left\|\bs H_{T\cap \Omega} \right\| \left(1 - \frac{1-\lambda}{\lambda}\sqrt{\frac{k\log^{\alpha}(n)}{m}}\right) &\lesssim \text{\upshape Tr}(\bs H_B) + \delta \left(\|\bs x_{S^c}\|_1 +  \text{Tr}(\bs H_{T^\perp})\right)  + \|\bs H_{T^\perp\cap \Omega}\|_F\\
& + \left(1-\lambda\right)\lambda^{-1} \left( \text{\upshape Tr}(\bs H_B)+ \text{Tr}(\bs H_{T^\perp\cap \Omega})\right)\\
&+\left(1-\lambda\right)\lambda^{-1} \delta \left(\|\bs H_{\Omega^c}\|_{\ell_1} + \text{Tr}(\bs H_{T^\perp})\right)
\end{split}\label{afterInterpolation04}
\end{align}
with probability $1-o_n(1)$ as soon as $m\gtrsim k\delta^{-1}$. 
\begin{restatable}{lemma}{restateCrossTermsQuadratic}
\label{lemmaCrossTermsAici01Quadratic}
Let $\bs A_i$, $\bs c_i$ respectively denote i.i.d standard normal random matrices and vectors.  Let $\Omega$ be defined as in section~\eqref{quadraticEquations} and let $T$ denote the tangent space to the cone of positive semidefinite matrices at $\bs X_0 = \bs x_0\bs x_0^*$
\begin{align}
\left\| m^{-1}\sum_{i=1}^m \mathcal{P}_{T\cap \Omega} (\bs e_1\bs c_i^* + \bs c_i\bs e_1^*) \langle \bs A_i, \bs X\rangle \right\| \leq \delta \|\bs X\|_F\label{boundInterpolation01}
\end{align}
with probability at least $1-n^{-\alpha}$ as soon as $m\gtrsim \delta^{-1}k$ (up to log factors)
\end{restatable}

Substituting~\eqref{afterInterpolation04} into~\eqref{afterApplyingHansonWright02} and combining with~\eqref{lowerBoundl1NormFinal01}, we get

\begin{align}
\|\bs X\|_{\ell_1}&\geq \|\bs X_0\|_{\ell_1} -|W_0|\left(4\text{\upshape Tr}\left(\bs H_{T^\perp\cap \Omega}\right) + \delta \text{\upshape Tr}(\bs H_{T^\perp}) + \delta\|\bs x_{S^c}\|_{\ell_1}+  \text{\upshape Tr}(\bs H_B)\right)\\
&-|W_0|(1-\lambda)\lambda^{-1}\left(4\text{\upshape Tr}\left(\bs H_{T^\perp\cap \Omega}\right) + \delta \text{\upshape Tr}(\bs H_{T^\perp}) + \delta\|\bs H_{\Omega^c}\|_{\ell_1}+  \text{\upshape Tr}(\bs H_B)\right)\\
& + \|\bs H_{\Omega^c}\|_{\ell_1}  + \gamma (1-\delta)\text{Tr}(\bs H_{B}) + \gamma (1-\delta) \text{Tr}(\bs H_{T^\perp\cap \Omega}) + \delta \gamma \text{\upshape Tr}(\bs H_{T^\perp}) 
\end{align}
where $W_0$ is defined as in~\eqref{definitionW0} with $|W_0|\lesssim\left(\mu_0(k+\gamma) +( k+\gamma)\sqrt{k}\mu_0^2\right)$.  Further developing, we obtain
\begin{align}
\|\bs X\|_{\ell_1}&\geq\|\bs X_0\|_{\ell_1} + \left(\gamma(1-\delta) - 4|W_0| - |W_0|(1-\lambda)\lambda^{-1}\right) \text{\upshape Tr}(\bs H_{T^\perp\cap \Omega})\\
&+ \left(\gamma(1-\delta) - |W_0|- |W_0|(1-\lambda)\lambda^{-1}\right) \text{\upshape Tr}(\bs H_{B})\\
&+ \left(\delta \gamma -|W_0|\delta- |W_0|\delta(1-\lambda)\lambda^{-1} \right)\text{\upshape Tr}(\bs H_{T^\perp})\\
&+ \left(2^{-1}\mu_0^{-1} - |W_0|\delta \right)\|\bs x_{S^c}\|_{\ell_1}\\
&+ 2^{-1}\left(1- |W_0|\delta(1-\lambda)\lambda^{-1}\right)\|\bs H_{\Omega^c}\|_{\ell_1} 
\end{align}
which immediately implies $\mu_0(1-\lambda)/\lambda \vee \mu_0^2\sqrt{k}(1-\lambda)/\lambda <1$ (up to log factors). From this, one can choose 
\begin{align}
\gamma > \mu_0^2\left( \frac{1-\lambda}{\lambda}\right) k^{3/2}
\end{align}
And the last term then gives 
\begin{align}
1>\delta \mu_0^2\left( \frac{1-\lambda}{\lambda}\right) k^{3/2}
\end{align}
which can be satisfied as soon as $m\gtrsim k^4\mu_0^4 (1-\lambda)^2/\lambda^2$.

\subsection{\label{sectionNumericalExperimentsInterpolation}Numerical Experiments}

Again, to support the theory, we perform a series of numerical experiments.  We start by studying the phase transition with respect to the recovery of a sparse vector $\bs x_0\in \mathbb{R}^{20}$, $\text{\upshape supp}(\bs x_0) = S$, $|S| =4$ which we generate as $\left.\bs x_0\right|_S \sim \text{\upshape Uni}(\mathbb{S}^{3})$, $S\sim \text{Uni}(n,k)$ for various values of $m$ (number of measurements) and $\lambda$ (linear vs quadratic tradeoff). For each $(m, \lambda)$ pair,  and for each of the two formulations $\mathsf{SDP}_1(\tilde{\bs A}, \lambda)$ and $\mathsf{SDP}_2(\tilde{\bs A}, \lambda)$, a total of $40$ experiments are run and the results are stored as $1$ if $\|\tilde{\bs X} - \tilde{\bs X}_0\|/\|\tilde{\bs X}_0\|<.01$ and $0$ otherwise. The averages are shown in Fig.~\ref{phaseIllustrationmanual} (top and bottom).  Note that the relative error is computed solely with respect to the central part $\tilde{\bs X}$ of $\bs X$ (leaving aside the first order pseudo moments which cannot be recovered for $\lambda=0$). The setting of Theorem~\ref{mainTheorem} (for an incoherent vector) is illustrated in the $(\lambda, k)$-space alongside the empirical phase transition in a similar framework ($k$-sparse vector $\bs x_0$ with entries generated as $\left.\bs x_0\right|_S\sim \text{\upshape Uni}(\mathbb{S}^{k-1})$ and support $S$ generated uniformly at random from $\mathbb{N}_n$) for $m=n=20$,  $k\in [0,15]$ and $\lambda\in[0,1]$. We represent the bounds $k\lesssim m\lambda^2/(1-\lambda)^2$, $k\lesssim m$ for appropriate constants selected to match the observed phase transition, as well as $\lambda > \sqrt{k}$.  

The numerical simulations of Fig~\ref{phaseIllustrationmanual} indicate a sharp reduction in the number of measurements needed around $\lambda = 0$.  To investigate this phenomenon more closely,  we perform additional simulations around that regime, taking $n = 30$, $k = 4$, a number of measurements $m = 50$ to ensure the quadratic regime (see~\cite{amini2008high, li2013sparse}) and generate the relative error on the second order pseudo-moments only for $50$ values of $\lambda$ between $0$ and $1$.  We repeat the experiment $10$ times and display the average relative error in Fig~\ref{phaseTransitionRecoveryColumnInnerMatrix} for both $\mathsf{SDP}_1(\tilde{\bs A}, \lambda)$ (top) and $\mathsf{SDP}_2(\tilde{\bs A}, \lambda)$ (bottom).  Once again, those simulations show, for $\mathsf{SDP}_2(\tilde{\bs A}, \lambda)$ only, a sharp decrease in the number of measurements needed when moving from a small external field to no external field at all.

\begin{figure}
\vspace{-7cm}
\begin{minipage}{.4\linewidth}
\vspace{7.4cm}
\begin{tikzpicture}
\node [color=black] at (3,1.5) {$\bs X \neq \bs X_0$};
\node [color=blue] at (4.3,4) {$\bs X = \bs X_0$};
  \begin{axis}[
      width = 8cm, height= 7cm,
      xlabel = {$\lambda$},
      ylabel = {$m$},
            xmin = 0,
ymin = 1,
xmax =1,
      ylabel style={rotate=-90}
    ]
    
    \addplot[name path=p,domain=0.04:1.00, smooth, thick]{-95.66*x^3 + 136.69*x^2 - 86.2*x + 50.17} ;
    \path[name path=axis] (axis cs:0.04,0) -- (axis cs:1,0);

\path[name path=axis2] (axis cs:0,0) -- (axis cs:.04,0);
\path[name path=axis3] (axis cs:0,48) -- (axis cs:.04,48);
    
\addplot[gray!30] fill between[of=p and axis];
\addplot[gray!30] fill between[of=axis2 and axis3];

  \end{axis}
\end{tikzpicture}
\end{minipage}
\hspace{.2cm}
\begin{minipage}{.4\linewidth}
%
%
\begin{tikzpicture}

\begin{axis}[%
width=2.528in,
height=2.054in,
,ylabel={$m$},xlabel={$\lambda$},ylabel style={rotate=-90},
at={(1.011in,0.642in)},
scale only axis,
axis on top,
xmin=-0.0138888888888889,
xmax=1.01388888888889,
ymin=3.81818181818182,
ymax=40.1818181818182,
axis background/.style={fill=white}
]
\addplot [forget plot] graphics [xmin=-0.23505050505050505, xmax=1.1870505050505, ymin=-26.5, ymax=68.5] {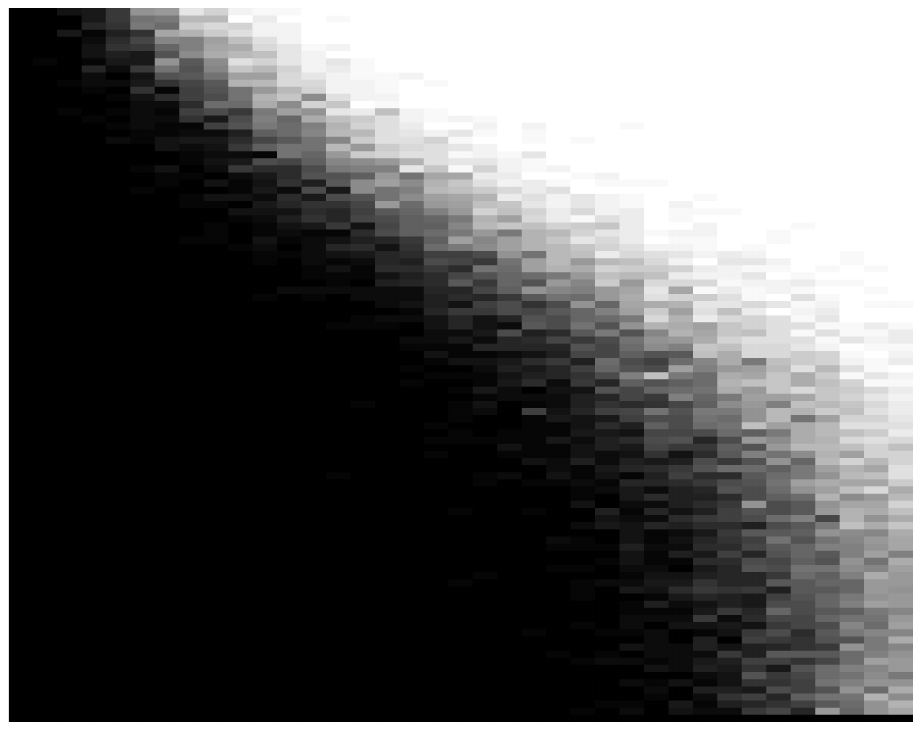};
\end{axis}

\begin{axis}[%
width=7.778in,
height=5.833in,
at={(0in,0in)},
scale only axis,
xmin=0,
xmax=1,
ymin=0,
ymax=1,
axis line style={draw=none},
ticks=none,
axis x line*=bottom,
axis y line*=left
]
\end{axis}
\end{tikzpicture}%
\end{minipage}

\vspace{-8cm}
\begin{minipage}{.4\linewidth}
\vspace{7.4cm}
\begin{tikzpicture}
\node [color=black] at (2,.7) {$\bs X \neq \bs X_0$};
\node [color=blue] at (5,4) {$\bs X = \bs X_0$};
\node [color=red] at (1.1,3) {$\tilde{\bs X} = \tilde{\bs X}_0$};
\draw[red, thick] (.9,3.2) -- (.1,4.5);
  \begin{axis}[
   width = 8cm, height= 7cm,
      xlabel = {$\lambda$},
      ylabel = {$m$},
      xmin = 0,
xmax =1,
      ylabel style={rotate=-90}
    ]
    
    \addplot[name path=p,domain=0.04:1.00, smooth, thick]{5.17319*x^3 + 19.69*x^2 - 59.80*x + 42.94} ;
    \path[name path=axis] (axis cs:0.04,5) -- (axis cs:1,5);

\path[name path=axis2] (axis cs:0,5) -- (axis cs:.04,5);
\path[name path=axis3] (axis cs:0,48) -- (axis cs:.04,48);
    
\addplot[gray!30] fill between[of=p and axis];
\addplot[gray!30] fill between[of=axis2 and axis3];

\draw[line width=4pt, red] (axis cs:0,22) -- (axis cs:0,48);

  \end{axis}

\end{tikzpicture}%
\end{minipage}
\hspace{.2cm}
\begin{minipage}{.4\linewidth}
%
%
\begin{tikzpicture}

\begin{axis}[%
width=2.528in,
height=2.054in,
at={(1.011in,0.642in)},
scale only axis,
axis on top,
xmin=-0.00505050505050505,
xmax=1.0050505050505,
ymin=3.5,
ymax=40.5,
axis background/.style={fill=white},ylabel={$m$},xlabel={$\lambda$},ylabel style={rotate=-90}
]
\addplot [forget plot] graphics [xmin=-0.21505050505050505, xmax=1.1670505050505, ymin=-26.5, ymax=68.5] {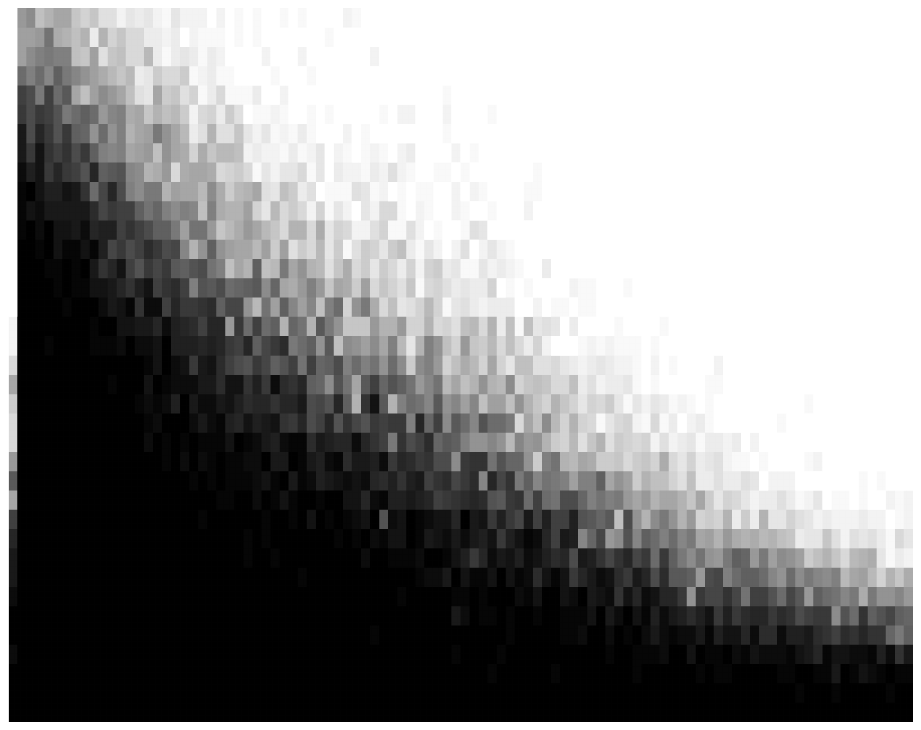};
\end{axis}

\begin{axis}[%
width=7.778in,
height=5.833in,
at={(0in,0in)},
scale only axis,
xmin=0,
xmax=1,
ymin=0,
ymax=1,
axis line style={draw=none},
ticks=none,
axis x line*=bottom,
axis y line*=left
]
\end{axis}
\end{tikzpicture}%
\end{minipage}

\caption{\label{phaseIllustrationmanual}Phase transitions observed for the semidefinite programs $\mathsf{SDP}_1(\tilde{\bs A}, \lambda)$ and $\mathsf{SDP}_2(\tilde{\bs A}, \lambda)$ (bottom). In this case the support size was set to $k=4$ and the ambient dimension was set to $n=20$. The total number of experiments is set to 40 and an experiment is considered as succesful if $\|\bs X - \bs X_0\|_F/\|\bs X_0\|_F<.01$ where $\bs X$ is the matrix returned by \textsc{cvx}. The various regimes are illustrated on the left where we have highlighted in red the regime in which the SDP successfully recovers the matrix $\tilde{\bs X}$ while failing to recover the whole matrix $\bs X$. The behavior of the SDP in this last regime is further illustrated in Fig.~\ref{phaseTransitionRecoveryColumnInnerMatrix}. }
\end{figure}

%
%
%

\begin{figure}
%
%
\begin{tikzpicture}

\begin{axis}[%
width=3.028in,
height=2.34in,
ylabel={$k$},xlabel={$\lambda$},ylabel style={rotate=-90},
at={(1.811in,0.642in)},
scale only axis,
axis on top,
xmin=-0.0128205128205128,
xmax=1.01282051282051,
ymin=0.5,
ymax=15.5,
axis background/.style={fill=white}
]
\addplot [forget plot] graphics [xmin=-0.258205128205128, xmax=1.21282051282051, ymin=-11.5, ymax=26.5] {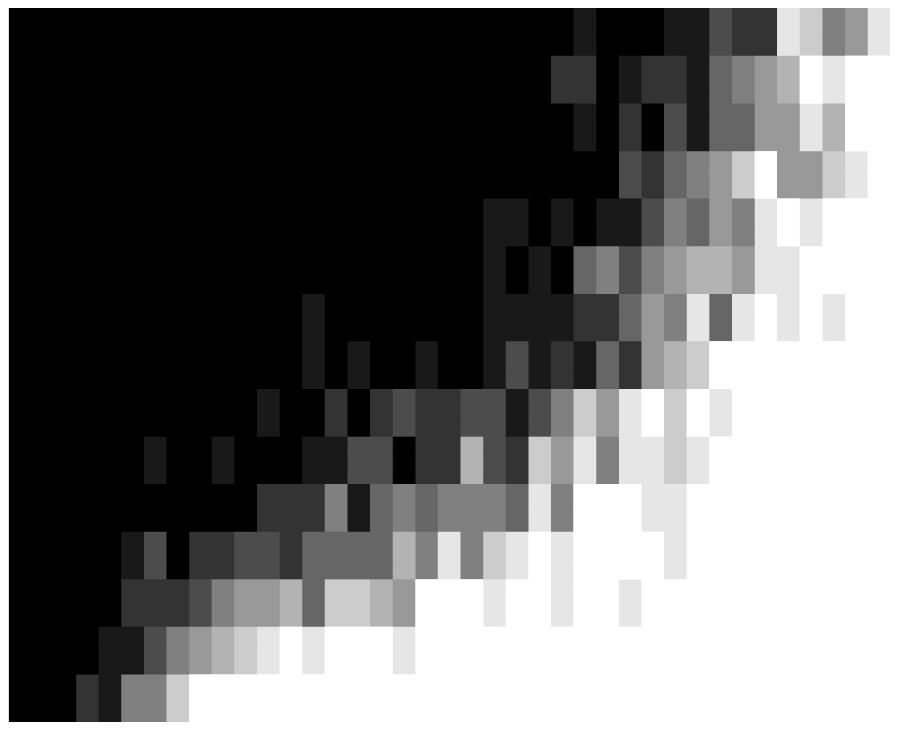};
\addplot [name path=A,color=green, forget plot, line width=1.5pt]
  table[row sep=crcr]{%
0	0\\
0.0256410256410256	0.000173130193905817\\
0.0512820512820513	0.000730460189919649\\
0.0769230769230769	0.00173611111111111\\
0.102564102564103	0.00326530612244898\\
0.128205128205128	0.00540657439446367\\
0.153846153846154	0.00826446280991736\\
0.179487179487179	0.011962890625\\
0.205128205128205	0.0166493236212279\\
0.230769230769231	0.0225\\
0.256410256410256	0.0297265160523187\\
0.282051282051282	0.0385841836734694\\
0.307692307692308	0.0493827160493827\\
0.333333333333333	0.0625\\
0.358974358974359	0.0784\\
0.384615384615385	0.09765625\\
0.41025641025641	0.120982986767486\\
0.435897435897436	0.149276859504132\\
0.461538461538462	0.183673469387755\\
0.487179487179487	0.225625\\
0.512820512820513	0.277008310249307\\
0.538461538461538	0.340277777777778\\
0.564102564102564	0.418685121107266\\
0.58974358974359	0.5166015625\\
0.615384615384615	0.64\\
0.641025641025641	0.797193877551021\\
0.666666666666667	1\\
0.692307692307692	1.265625\\
0.717948717948718	1.6198347107438\\
0.743589743589744	2.1025\\
0.769230769230769	2.77777777777778\\
0.794871794871795	3.75390625\\
0.82051282051282	5.22448979591837\\
0.846153846153846	7.5625\\
0.871794871794872	11.56\\
0.897435897435897	19.140625\\
0.923076923076923	36.0000000000001\\
0.948717948717949	85.5624999999999\\
0.974358974358974	360.999999999999\\
1	inf\\
};
\addplot [color=blue, forget plot, line width=1.5pt]
  table[row sep=crcr]{%
0.66	0\\
0.66	1\\
0.66	2\\
0.66	3\\
0.66	4\\
0.66	5\\
0.66	6\\
0.66	7\\
0.66	8\\
0.66	9\\
0.66	10\\
0.66	11\\
0.66	12\\
0.66	13\\
0.66	14\\
0.66	15\\
0.66	15.4\\
};
\addplot [name path=B,color=green, dashed, forget plot, line width=1.5pt]
  table[row sep=crcr]{%
0	8.4\\
0.0282051282051282	8.4\\
0.0564102564102564	8.4\\
0.0846153846153846	8.4\\
0.112820512820513	8.4\\
0.141025641025641	8.4\\
0.169230769230769	8.4\\
0.197435897435897	8.4\\
0.225641025641026	8.4\\
0.253846153846154	8.4\\
0.282051282051282	8.4\\
0.31025641025641	8.4\\
0.338461538461538	8.4\\
0.366666666666667	8.4\\
0.394871794871795	8.4\\
0.423076923076923	8.4\\
0.451282051282051	8.4\\
0.47948717948718	8.4\\
0.507692307692308	8.4\\
0.535897435897436	8.4\\
0.564102564102564	8.4\\
0.592307692307692	8.4\\
0.620512820512821	8.4\\
0.648717948717949	8.4\\
0.676923076923077	8.4\\
0.705128205128205	8.4\\
0.733333333333333	8.4\\
0.761538461538462	8.4\\
0.78974358974359	8.4\\
0.817948717948718	8.4\\
0.846153846153846	8.4\\
0.874358974358974	8.4\\
0.902564102564103	8.4\\
0.930769230769231	8.4\\
0.958974358974359	8.4\\
0.987179487179487	8.4\\
1.01538461538462	8.4\\
1.04358974358974	8.4\\
1.07179487179487	8.4\\
1.1	8.4\\
};

\path[name path=C] (axis cs:0,0) -- (axis cs:1,0);


\tikzfillbetween[
    of=A and C, soft clip={domain=0:.855}] {pattern color=green, pattern=north east lines};
\tikzfillbetween[
    of=B and C, soft clip={domain=.855:1}] {pattern color=green, pattern=north east lines};

\end{axis}

\begin{axis}[%
width=7.778in,
height=5.833in,
at={(0in,0in)},
scale only axis,
xmin=0,
xmax=1,
ymin=0,
ymax=1,
axis line style={draw=none},
ticks=none,
axis x line*=bottom,
axis y line*=left
]
\end{axis}

\end{tikzpicture}%
\caption{Evolution of the recovery for the semidefinite program~\eqref{sdpquadraticLinear} for various support sizes $k$ and various magnitudes of the interpolation parameter $\lambda$. In this case the support size was set to $k=4$ and the vector size was set to $n=20$. The total number of experiments is set to 40 and an experiment is considered as succesful if $\|\tilde{\bs X} - \bs X_0\|_F/\|\bs X_0\|_F<.01$ where $\bs X$ is the matrix returned by \textsc{cvx}. The green lines indicate the sample complexities $k\leq m\lambda^2/(1-\lambda)^2/5$ and $k\leq  m/4.8$ (The multiplicative constants were chosen to match the empirical phase transition). The blue line indicates the lower bound $\lambda>k^{1/2}/(k^{1/2}+1)$. The setting of Theorem~\ref{mainTheorem} for strongly incoherent vectors, i.e. $\mu_0k^{1/2}\asymp 1$ is illustrated by the area hatched in green. }
\end{figure}

\begin{figure}

    \begin{minipage}{.45\linewidth}
\hspace{-1cm}
%
%
\definecolor{mycolor1}{rgb}{0.00000,0.44700,0.74100}%
\begin{tikzpicture}

\begin{axis}[%
width=2.428in,
height=1.954in,
at={(1.011in,0.642in)},
scale only axis,
xmin=0,
xmax=1,
ymode=log,
ymin=1e-12,
ymax=1,
yminorticks=true,
xlabel={$\lambda$},ylabel style={rotate=-90},
axis background/.style={fill=white}]
\addplot [color=mycolor1, forget plot,line width=2pt]
  table[row sep=crcr]{%
0	0.707106781186548\\
.02	0.488699753074324\\
.04	0.289924620755412\\
.06	0.148710049433817\\
.08	0.236317011751795\\
.1	0.152218476132037\\
.12	0.167116612162094\\
.14	0.0489057759557445\\
.16	0.121733373170036\\
.18	0.0151037734795895\\
.2	0.113644402455601\\
.22	0.0075622744283673\\
.24	0.0404759453685622\\
.26	0.0599200925811696\\
.28	0.0613260493750187\\
.3	0.0250481713210799\\
.32	0.0138086355512643\\
.34	2.4862135496553e-10\\
.36	0.0354882021569276\\
.38	0.0732398338478439\\
.4	0.00166041356014295\\
.42	1.67524794962368e-08\\
.44	6.9631658123186e-09\\
.46	1.95948678371456e-10\\
.48	0.00589728538485635\\
.5	2.46113660202836e-09\\
.52	1.26609373723646e-10\\
.54	2.88780520900752e-09\\
.56	2.52590751192157e-09\\
.58	2.62861708982063e-10\\
.6	2.16744150294967e-09\\
.62	2.05589873787076e-10\\
.64	1.38120308222496e-10\\
.66	1.06447458814097e-10\\
.68	1.13211448750492e-10\\
.7	8.94268284414377e-11\\
.72	9.41381768411223e-11\\
.74	1.33046881267128e-10\\
.76	1.04942941808227e-10\\
.78	1.08471981558789e-10\\
.8	1.1095836286073e-10\\
.82	1.32447967758319e-10\\
.84	5.88657410244242e-11\\
.86	8.1713294500891e-11\\
.88	8.66178235572038e-11\\
.9	1.06814622102025e-10\\
.92	7.08432264686367e-11\\
.94	7.00352623058077e-11\\
.96	6.31959165014668e-11\\
.98	7.44359623105767e-11\\
1	9.73570786882417e-11\\
};
%
\end{axis}

\begin{axis}[%
width=2.428in,
height=1.954in,
at={(0in,0in)},
scale only axis,
xmin=0,
xmax=1,
ymin=0,
ymax=1,
axis line style={draw=none},
ticks=none,
axis x line*=bottom,
axis y line*=left
]
\end{axis}
\end{tikzpicture}%
\end{minipage}
\begin{minipage}{.45\linewidth}
\hspace{-2.6cm}\vspace{-.4cm}
%
%
\definecolor{mycolor1}{rgb}{0.00000,0.44700,0.74100}%
\begin{tikzpicture}

\begin{axis}[%
width=2.428in,
height=1.954in,
at={(1.011in,0.642in)},
scale only axis,
xmin=0,
xmax=1,
ymode=log,
ymin=1e-12,
ymax=1,
yminorticks=true,
xlabel={$\lambda$},ylabel style={rotate=-90},
axis background/.style={fill=white}
]
\addplot [color=mycolor1, forget plot,line width=2pt]
  table[row sep=crcr]{%
0	0.0762681004627713\\
.02	0.0764296064215244\\
.04	0.115776809227246\\
.06	0.178573195865403\\
.08	0.0695368996468805\\
.1	0.0908785665949652\\
.12	0.219737297862986\\
.14	0.173371769754537\\
.16	0.0969808855136793\\
.18	0.0800327160875975\\
.2	0.147660850397208\\
.22	0.0250807828952093\\
.24	0.127329994299783\\
.26	0.144113945672819\\
.28	0.0444506300889511\\
.3	0.121977039810627\\
.32	0.0466596685791171\\
.34	0.0711617084156318\\
.36	0.0210743489475558\\
.38	0.0420026650385536\\
.4	0.00775337764658844\\
.42	4.87969545267645e-09\\
.44	0.0234010657419767\\
.46	0.00984003060506706\\
.48	4.76261818653961e-10\\
.5	5.31055169452492e-09\\
.52	1.2430563524251e-08\\
.54	3.39056843318342e-09\\
.56	3.44973787390983e-10\\
.58	6.6077807688062e-09\\
.6	3.64503251035118e-10\\
.62	2.90946933344492e-10\\
.64	2.53644549301481e-10\\
.66	2.33277813367697e-10\\
.68	2.25516793768179e-10\\
.7	2.30995950073741e-10\\
.72	3.12009786740235e-10\\
.74	2.96576602612956e-10\\
.76	2.38049007447961e-10\\
.78	1.93233671148461e-10\\
.8	2.91656932598342e-10\\
.82	2.83653974158873e-10\\
.84	2.01732451098627e-10\\
.86	1.48447143087232e-10\\
.88	1.77475911524968e-10\\
.9	2.27322463806119e-10\\
.92	2.09901568207585e-10\\
.94	1.98187448105543e-10\\
.96	1.94296340357099e-10\\
.98	1.67530289577671e-10\\
1	1.7861221475827e-10\\
};
\end{axis}

\begin{axis}[%
width=2.428in,
height=1.954in,
at={(0in,0in)},
scale only axis,
xmin=0,
xmax=1,
ymin=0,
ymax=1,
axis line style={draw=none},
ticks=none,
axis x line*=bottom,
axis y line*=left
]
\end{axis}
\end{tikzpicture}%
\end{minipage}
\begin{minipage}{.45\linewidth}
\hspace{-1cm}
%
%
\definecolor{mycolor1}{rgb}{0.00000,0.44700,0.74100}%
\begin{tikzpicture}

\begin{axis}[%
width=2.428in,
height=1.954in,
at={(1.011in,0.642in)},
scale only axis,
xmin=0,
xmax=1,
ymode=log,
ymin=1e-12,
ymax=1,
yminorticks=true,xlabel={$\lambda$},ylabel style={rotate=-90},
axis background/.style={fill=white}
]
\addplot [color=mycolor1, forget plot,line width=2pt]
  table[row sep=crcr]{%
0	0.707106781186548\\
.02	0.348394835487082\\
.04	0.25992289966914\\
.06 0.0563599827518177\\
.08	0.209913892735856\\
.1	2.54302551118754e-08\\
.12	1.1892712484522e-08\\
.14	2.22069217595453e-08\\
.16	9.14316933980127e-09\\
.18	1.43071213255648e-08\\
.2	2.37853464973275e-08\\
.22	1.13593727258668e-09\\
.24	7.91525704403228e-09\\
.26	1.00199809309368e-08\\
.28	4.56328337919368e-09\\
.3	5.33223573390809e-09\\
.32	8.54218305880522e-10\\
.34	4.96600312150329e-10\\
.36	6.33131616816669e-09\\
.38	4.10065497800899e-09\\
.4	7.87359911086991e-10\\
.42	1.88125806961419e-09\\
.44	4.42041708248658e-09\\
.46	4.26728691093929e-10\\
.48	1.05751720946896e-09\\
.5	8.5471115993097e-10\\
.52	1.55192448645892e-10\\
.54	3.27373608113146e-09\\
.56	2.76422489393089e-10\\
.58	2.91870435547821e-12\\
.6	1.99975479484318e-11\\
.62	2.96411088254189e-11\\
.64	1.59697906759252e-10\\
.66	1.29557217859553e-11\\
.68	4.86901878834046e-12\\
.7	5.97806149358988e-11\\
.72	1.71150144916825e-11\\
.74	8.49875929367104e-12\\
.76	7.71553224782635e-12\\
.78	7.40349759329004e-12\\
.8	1.23889256671768e-11\\
.82	6.40049509557755e-11\\
.84	6.15677541364696e-12\\
.86	8.44035771107842e-12\\
.88	1.02608575478123e-11\\
.9	5.22330335573654e-12\\
.92	6.9024497681041e-12\\
.94	1.01320808180547e-11\\
.96	8.1192785303854e-12\\
.98	1.21628458277225e-11\\
1	9.27339396978897e-12\\
};
\end{axis}

\begin{axis}[%
width=2.428in,
height=1.954in,
at={(0in,0in)},
scale only axis,
xmin=0,
xmax=1,
ymin=0,
ymax=1,
axis line style={draw=none},
ticks=none,
axis x line*=bottom,
axis y line*=left
]
\end{axis}
\end{tikzpicture}%
\end{minipage}
\begin{minipage}{.45\linewidth}
\hspace{-1cm}
%
%
\definecolor{mycolor1}{rgb}{0.00000,0.44700,0.74100}%
\begin{tikzpicture}

\begin{axis}[%
width=2.428in,
height=1.954in,
at={(1.011in,0.642in)},
scale only axis,
xmin=0,
xmax=1,
ymode=log,
ymin=1e-12,
ymax=1,
yminorticks=true,
xlabel={$\lambda$},ylabel style={rotate=-90},
axis background/.style={fill=white}
]
\addplot [color=mycolor1, forget plot,line width=2pt]
  table[row sep=crcr]{%
0	1.93638321900315e-09\\
.02	0.0219015693929812\\
.04	0.0334398395886849\\
.06	0.0219327877689793\\
.08	0.0492352571437857\\
.1	3.25238745742614e-08\\
.12	0.00863471083369332\\
.14	6.26066749083661e-09\\
.16	3.43851083882484e-09\\
.18	1.97873277816575e-09\\
.2	1.90600443237248e-09\\
.22	1.66755476663834e-09\\
.24	7.69500535284038e-09\\
.26	3.57643168575154e-09\\
.28	2.0874286088635e-09\\
.3	1.81945756888655e-09\\
.32	2.40732652591012e-09\\
.34	2.8510043910822e-09\\
.36	1.20785992020527e-09\\
.38	2.15132901148835e-09\\
.4	2.11046112271008e-09\\
.42	1.3290438072085e-09\\
.44	1.42699122705337e-09\\
.46	1.46089017724377e-09\\
.48	1.33361526878256e-09\\
.5	9.08721273512973e-10\\
.52	4.27778755045799e-10\\
.54	4.66249055546394e-10\\
.56	9.82588632005039e-10\\
.58	1.39623744168086e-10\\
.6	1.71751865565429e-10\\
.62	1.25345797546245e-10\\
.64	5.72691143455776e-11\\
.66	5.1442284060662e-11\\
.68	5.74963659148019e-11\\
.7	8.7598911762941e-11\\
.72	3.92110902775482e-11\\
.74	2.17964258594876e-11\\
.76	4.53906901841572e-11\\
.78	4.0494200468466e-11\\
.8	7.31260146070879e-11\\
.82	7.36534507238798e-11\\
.84	2.08278562235152e-11\\
.86	4.32141058704165e-11\\
.88	1.07105148566457e-10\\
.9	3.00489730618111e-11\\
.92	7.22770375141804e-12\\
.94	6.25888229907746e-11\\
.96	5.03200427480621e-11\\
.98	3.38909452364506e-11\\
1	4.38262341768506e-11\\
};
\end{axis}

\begin{axis}[%
width=2.428in,
height=1.954in,
at={(0in,0in)},
scale only axis,
xmin=0,
xmax=1,
ymin=0,
ymax=1,
axis line style={draw=none},
ticks=none,
axis x line*=bottom,
axis y line*=left
]
\end{axis}
\end{tikzpicture}%
\end{minipage}
\caption{\label{phaseTransitionRecoveryColumnInnerMatrix} Phase transition for the semidefinite programs $\mathsf{SDP}_1(\tilde{\bs A}, \lambda)$ and $\mathsf{SDP}_2(\tilde{\bs A}, \lambda)$ in the ``large $m$" (quadratic recovery) regime (area between the red line and the black line in Fig.~\ref{phaseIllustrationmanual}) in which the recovery is possible for the purely quadratic measurement operator only if we restrict to second order part of the matrix. Here we take $n = 30$, $k = 4$ $m=50$ and we collect a number of $10$ experiments which are averaged to provide the blue curves.  $\lambda$ is sampled every $.02$ values between $0$ and $1$. The experiments are carried out in the ``large $m$" (i.e. $m>k^2$) regime as shown by the exact recovery of the solution at $\lambda=0$}
\end{figure}
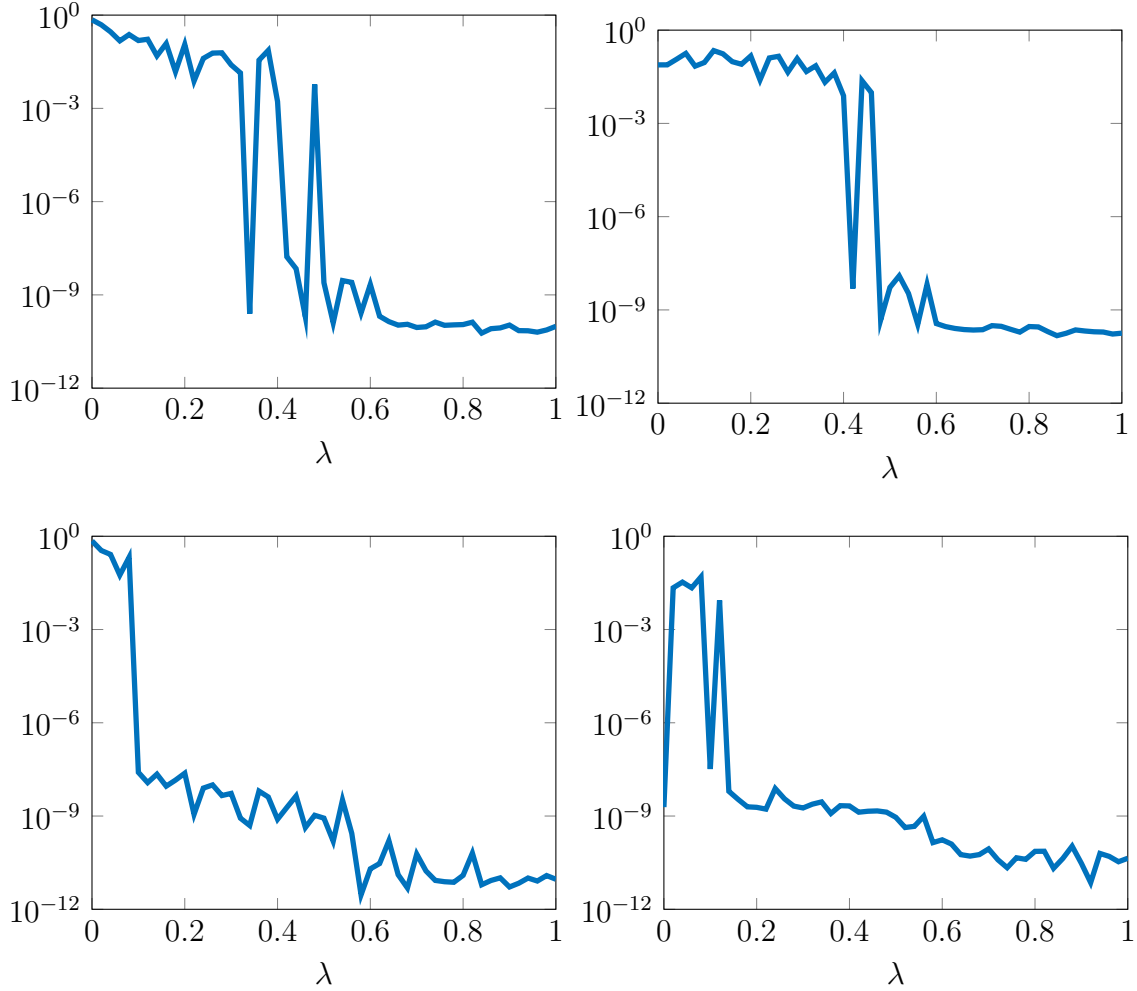

%
%
%
%
%
%
%
%
%
%


\subsection{Auxilliary results}

\subsubsection{Proof of lemma~\ref{lemmaCrossTermsAici01Quadratic}}

\restateCrossTermsQuadratic*

\begin{proof}
From the independence of the $\bs c_i, \bs A_i$, letting 
\begin{align}
Z_i = \mathcal{P}_{T\cap \Omega}(\bs e_1\bs c_i^* + \bs c_i\bs e_1^*) \langle \bs A_i, \bs X\rangle
\end{align}
we have $\mathbb{E}Z_i = 0$. Moreover, 
\begin{align}
&\left\|m^{-1}\sum_{i=1}^m Z_iZ_i^*\right\|\\
& = \left\|m^{-1}\sum_{i=1}^m \left(\frac{(\bs x_0\bs c_i^* + \bs c_i\bs x_0^*)}{2} + \langle \bs c_i, \bs x_0\rangle \frac{(\bs e_1\bs x_0^* + \bs x_0\bs e_1^*)}{2} - \frac{\langle \bs x_0, \bs c_i\rangle \bs x_0\bs x_0^*}{2} \right)^2\right\|\\
& \lesssim \left\|\mathbb{E} \left|\langle \bs A_i, \bs X\rangle \right|^2 \frac{(\bs x_0\bs c_i^* + \bs c_i\bs x_0^*)^2}{4} \right\|+ \left\|\mathbb{E} |\langle \bs A_i, \bs X\rangle |^2 |\langle \bs c_i, \bs x_0\rangle |^2 \frac{(\bs e_1\bs x_0^* + \bs x_0\bs e_1^*)^2}{4}\right\|\\
&+ \left\|\mathbb{E} \frac{|\langle \bs x_0, \bs c_i\rangle |^2}{4} |\langle \bs A_i, \bs X\rangle |\bs x_0\bs x_0^* \|\bs x_0\|^2\right\|\\
&\lesssim \left\|\frac{\|\bs X\|_F^2}{4} \mathbb{E}\left\{\langle \bs x_0, \bs c_i\rangle \left(\bs x_0\bs c_i^* + \bs c_i\bs x_0^*\right) + \bs c_i\bs c_i^*\|\bs x_0\|^2 + \bs x_0\bs x_0^*\|\bs c_i\|^2 \right\}\right\|\\
&+ \left\|\frac{\|\bs X\|_F^2}{4} \|\bs x_0\|^2 \frac{\left(\bs e_1\bs x_0^* + \bs x_0\bs e_1^* + \bs e_1\bs e_1^* \|\bs x_0\|^2 + \bs x_0\bs x_0^*\right)}{4}\right\|+ \left\|\frac{\|\bs X\|_F^2}{4} \|\bs x_0\|^4 \bs x_0\bs x_0^*\right\|\\
&\lesssim \|\bs X\|_F^2 k\label{boundVarianceLemmaQuadratic01}
\end{align}
The variables $Z_i$ are subexponential and we have
\begin{align}
\|Z_i\|_{\psi_1}&= \left\|\mathcal{P}_{T\cap \Omega}\left(\bs e_1\bs c_i^* + \bs c_i\bs e_1^*\right)\langle \bs A_i, \bs X\rangle - \mathbb{E}\mathcal{P}_{T\cap \Omega}\left(\bs e_1\bs c_i^* + \bs c_i\bs e_1^*\right)\langle \bs A_i, \bs X\rangle \right\|_{\psi_1}\\
&\leq  \left\|\mathcal{P}_{T\cap \Omega}\left(\bs e_1\bs c_i^* + \bs c_i\bs e_1^*\right)\right\|_{\psi_2} \left\|\langle \bs A_i, \bs X\rangle\right\|_{\psi_2}\\
&\leq \left\|\frac{(\bs x_0\bs c_i^* + \bs c_i\bs x_0^*)}{2} + \langle \bs x_0, \bs c_i\rangle \frac{(\bs e_1\bs x_0^* +\bs x_0\bs e_1^*)}{2} - \frac{1}{2}\langle \bs x_0, \bs c_i\rangle \bs x_0\bs x_0^*  \right\|_{\psi_2} \|\bs X\|_F\\
&\lesssim \|\bs X\|_F \left(\sqrt{k}+ \|\bs x_0\|\right)\label{boundOrliczLemmaQuadratic01}
\end{align}
Combining~\eqref{boundOrliczLemmaQuadratic01} and~\eqref{boundVarianceLemmaQuadratic01} and applying Proposition~\ref{BernsteinMatrix} gives the desired result.

\end{proof}

\bibliographystyle{abbrv}
\bibliography{sample}

\end{document}